\documentclass{article}
\usepackage[margin=0.8in]{geometry}
\usepackage{mathtools}
\usepackage{amssymb}
\usepackage{amsthm}
\usepackage{enumitem}
\usepackage{bigstrut}
\usepackage{arydshln}
\usepackage{subfig}
\usepackage{color}
\usepackage{graphicx}
\usepackage{tikz-cd}
\usepackage[
colorlinks=true,
citecolor=blue,
linkcolor=blue,
]{hyperref}
\usetikzlibrary{automata, positioning,arrows,shapes,decorations.pathmorphing}
\newtheorem{conjecture}{Conjecture}
\newtheorem{theorem}{Theorem}
\newtheorem{proposition}{Proposition}
\newtheorem{lemma}{Lemma}
\newtheorem{corollary}{Corollary}
\theoremstyle{definition}
\newtheorem{example}{Example}
\newtheorem{definition}{Definition}

\newtheorem{claim}{Claim}
\newcommand{\step}[1]{\stackrel{#1}{\longrightarrow}}
\newcommand{\cqastep}[1]{\stackrel{#1}{\longrightarrow\!\!\!\!\!\!\!\!\longrightarrow}}
\newcommand{\longstep}[1]{\stackrel{#1}{\xrightarrow{\hspace*{30pt}}}}
\newcommand{\longcqastep}[1]{\stackrel{#1}{\xrightarrow{\hspace*{30pt}}\!\!\!\!\!\!\!\!\longrightarrow}}
\newcommand{\verylongstep}[1]{\stackrel{#1}{\xrightarrow{\hspace*{70pt}}}}
\newcommand{\myparagraph}[1]{{\textbf{#1}.}}
\newcommand{\problem}[1]{\mathsf{#1}}
\newcommand{\FO}{$\mathbf{FO}$}
\newcommand{\PTIME}{$\mathbf{PTIME}$}
\newcommand{\LSPACE}{$\mathbf{L}$}
\newcommand{\NL}{$\mathbf{NL}$}
\newcommand{\coNP}{$\mathbf{coNP}$}
\newcommand{\typeset}[1]{\mathcal{#1}}
\newcommand{\bone}{\typeset{B}_{1}}
\newcommand{\btwoa}{\typeset{B}_{2a}}
\newcommand{\btwob}{\typeset{B}_{2b}}
\newcommand{\bthree}{\typeset{B}_{3}}
\newcommand{\cone}{\typeset{C}_{1}}
\newcommand{\ctwo}{\typeset{C}_{2}}
\newcommand{\cthree}{\typeset{C}_{3}}
\newcommand{\done}{\typeset{D}_{1}}
\newcommand{\dtwo}{\typeset{D}_{2}}
\newcommand{\dthree}{\typeset{D}_{3}}
\newcommand{\defeq}{\vcentcolon=}
\newcommand{\formula}[1]{\left({#1}\right)}
\newcommand{\normalformula}[1]{({#1})}
\newcommand{\card}[1]{|{#1}|}
\newcommand{\pair}[2]{\langle{#1},{#2}\rangle}
\newcommand{\adom}[1]{{\mathsf{adom}}({#1})}
\newcommand{\cqa}[1]{\mathsf{CERTAINTY}({#1})}
\newcommand{\emptyword}{\varepsilon}
\newcommand{\first}[1]{{\mathsf{first}}({#1})}
\newcommand{\last}[1]{\mathsf{last}({#1})}
\newcommand{\head}[1]{{\mathsf{head}}({#1})}
\newcommand{\rear}[1]{{\mathsf{rear}}({#1})}
\newcommand{\pumpclosure}[1]{{\mathcal{L}}^{\looparrowright}({#1})}
\mathchardef\mhyphen="2D
\newcommand{\sset}[3]{{\mathsf{ST}}_{#3}({#1},{#2})}
\newcommand{\csset}[3]{{\mathsf{cqaST}}_{#3}({#1},{#2})}
\newcommand{\nfa}[1]{{\mathsf{NFA}}({#1})}
\newcommand{\snfa}[2]{{\mathsf{S\mhyphen NFA}}({#1},{#2})}

\newcommand{\chr}[1]{{\mathsf{char}}({#1})}
\newcommand{\extend}[1]{{\mathsf{ext}}({#1})}
\newcommand{\nfashortest}[1]{{\mathsf{NFA^{\min}}}({#1})}
\newcommand{\startshortest}[2]{{\mathsf{start^{\min}}}({#1},{#2})}
\newcommand{\db}{\mathbf{db}}
\newcommand{\block}{\mathbf{blk}}
\newcommand{\rep}{\mathbf{r}}
\newcommand{\sep}{\mathbf{s}}

\newcommand{\var}[1]{\mathsf{vars}({#1})}
\newcommand{\queryvars}[1]{{\mathsf{vars}}({#1})}
\newcommand{\proves}[1]{\vdash_{#1}}
\newcommand{\nproves}[1]{\not\vdash_{#1}}
\newcommand{\constant}[1]{\mathtt{#1}}
\newcommand{\qconstant}[1]{\mbox{`$\mathtt{#1}$'}}
\newcommand{\myprec}[1]{\prec_{#1}}
\newcommand{\mypreceq}[1]{\preceq_{#1}}

\newcommand{\fixedhead}[2]{{#1}_{[#2]}}
\newcommand{\pathcons}[2]{[\![{#1},{#2}]\!]}
\newcommand{\starttwo}[2]{{\mathsf{start}}({#1},{#2})}
\newcommand{\kleene}[2]{\left({#1}\right)^{#2}}
\newcommand{\shortcdot}{\!\cdot\!}
\newcommand{\swipe}[3]{\mbox{${#1}\shortcdot{#2}\shortcdot\underline{#2}\shortcdot{#3}$}}
\newcommand{\prefixswipe}[2]{{#1}\shortcdot\underline{#1}\shortcdot{#2}}
\newcommand{\symbols}[1]{{\mathsf{symbols}}({#1})}
\usepackage[auth-sc-lg]{authblk}
\author{Paraschos Koutris}
\author{Xiating Ouyang}
\affil{University of Wisconsin-Madison, WI, USA}
\author{Jef Wijsen}
\affil{University of Mons, Belgium}
\date{}

\title{Consistent Query Answering for Primary Keys on Path Queries\thanks{This paper is an evolved version of a paper with the same title and authors published at ACM PODS'21~\cite{DBLP:conf/pods/KoutrisOW21}. In particular, the proof of Lemma~\ref{lem:shrinking} in the current paper is new and replaces a flawed proof in the earlier version, and the technical treatment in the current Section~\ref{sec:datalog} strengthens some earlier results.}}

\begin{document}
\maketitle
 
\begin{abstract}
We study the data complexity of consistent query answering (CQA) on databases that may violate the primary key constraints. A repair is a maximal consistent subset of the database. For a Boolean query~$q$, the problem $\cqa{q}$ takes a database as input, and asks whether or not each repair satisfies~$q$. It is known that for any self-join-free Boolean conjunctive query $q$, $\cqa{q}$ is in \FO, \LSPACE-complete, or \coNP-complete. In particular, $\cqa{q}$ is in \FO\ for any self-join-free Boolean path query~$q$. In this paper, we show that if self-joins are allowed, the complexity of $\cqa{q}$ for Boolean path queries~$q$ exhibits a tetrachotomy between \FO, \NL-complete, \PTIME-complete, and \coNP-complete. Moreover, it is decidable, in polynomial time in the size of the query~$q$, which of the four cases applies. 
%We also show that when $\cqa{q}$ is in \FO, a consistent first-order rewriting can be effectively constructed. 
\end{abstract}

\section{Introduction}\label{sec:introduction}

%The primary key constraint is often imposed on databases as an integrity constraint. An (inconsistent) database violates the primary key constraint if it contains two distinct facts holding the same primary key. 
%Inconsistencies in databases are common in practice due to erroneous data integration, duplicate data exchange or acquisition of dirty data. Data cleaning is one common approach to deal with such inconsistencies, in which the inconsistent database is first repaired so that it satisfies the integrity constraints, and then execute the query on the repaired database. However, such methods often suffer from information loss since there are many different ways to repair the database. For that reason, we consider \emph{consistent query answering}, that is, to find the query answers that appear in \emph{every} repair. In this paper, we consider only the inconsistent databases that may violate the primary key constraint. 

Primary keys are probably the most common integrity constraints in relational database systems.
Although databases should ideally satisfy their integrity constraints,
data integration is today frequently cited as a cause for primary key violations, for example, when a same client is stored with different birthdays in two data sources.
A \emph{repair} of such an inconsistent database instance is then naturally defined as a maximal consistent subinstance.
Two approaches are then possible.
In \emph{data cleaning}, the objective is to single out the ``best'' repair, which however may not be practically possible.
In \emph{consistent query answering} (CQA)~\cite{10.1145/303976.303983}, instead of cleaning the inconsistent database instance, we change the notion of query answer: the \emph{consistent} (or \emph{certain}) \emph{answer} is defined as the intersection of the query answers over all (exponentially many) repairs. 
%\revision{Computing the consistent answers for non-Boolean queries reduces efficiently to computing that for Boolean queries \cite{10.1145/2188349.2188351,10.1145/3340531.3411911}.}  
In computational complexity studies, consistent query answering is commonly defined as the data complexity of the following decision problem, for a fixed Boolean query~$q$:  

\begin{description}
	\item[Problem:] $\cqa{q}$
	\item[Input:] A database instance $\db$.
	\item[Question:] Does $q$ evaluate to true on every repair of $\db$?
\end{description}

For every first-order query~$q$, the problem $\cqa{q}$ is obviously in \coNP. %since a repair on which $q$ is evaluated to be false (called a falsifying repair) can be verified in polynomial time by simply running the query on that repair. However, it is also known that for some boolean query $q$, the problem $\cqa{q}$ can be solved in polynomial time. In certain cases, $\cqa{q}$ even admits a \emph{first-order rewriting}: there exists another boolean query $q'$ such that $q'$ returns true on an inconsistent database $\db$ if and only if $\db$ is a ``yes''-instance for $\cqa{q}$, and we say that $\cqa{q}$ is in \FO. It has been long conjectured in the database theory community that the complexity of $\cqa{q}$ displays a dichotomy.
However, despite significant research efforts (see Section~\ref{sec:related-work}), a fine-grained complexity classification is still largely open. 
A notorious open conjecture is the following.

\begin{conjecture}\label{conj:dichotomy}
For each Boolean conjunctive query $q$, $\cqa{q}$ is 
either in \PTIME\ or \coNP-complete.
\end{conjecture}

%The above dichotomy conjecture has also been stated for the larger class of unions of conjunctive queries (possibly with self-joins).
%Fontaine~\cite{10.1145/2699912} showed that such a dichotomy for unions of conjunctive queries implies Bulatov's dichotomy theorem for conservative CSP~\cite{10.1145/1970398.1970400}.
%This demonstrates the difficulty of proving dichotomies in CQA.

On the other hand, for the smaller class of self-join-free Boolean conjunctive queries, the complexity landscape is by now well understood,
as summarized by the following theorem.

\begin{theorem}[\cite{KoutrisWTOCS20}] \label{thm:sjf-theorem}
For each self-join-free Boolean conjunctive query $q$, $\cqa{q}$ is 
in \FO, \LSPACE-complete, or \coNP-complete, and it is decidable which of the three cases applies. 
\end{theorem}

Abandoning the restriction of self-join-freeness turns out to be a major challenge.
The difficulty of self-joins is caused by the obvious observation that a single database fact can be used to satisfy more than one atom of a conjunctive query, as illustrated by Example~\ref{ex:intro1}.
Self-joins happen to significantly change the complexity landscape laid down in Theorem~\ref{thm:sjf-theorem}; this is illustrated by Example~\ref{ex:intro2}.
Self-join-freeness is a simplifying assumption that is also used outside CQA (e.g., \cite{FreireGIM15, berkholz2017answering, FreireGIM20}).

\begin{example}
\label{ex:intro1}
Take the self-join $q_1=\exists x\exists y(R(\underline{x},y)\land R(\underline{y},x))$ and its self-join-free counterpart $q_2=\exists x\exists y(R(\underline{x}, y)\land S(\underline{y}, x))$, where the primary key positions are underlined. Consider the inconsistent database instance $\db$ in Figure~\ref{tbl:instance}. We have that $\db$ is a ``no''-instance of $\cqa{q_2}$, because $q_2$ is not satisfied by the repair $\{R(\underline{a},a)$, $R(\underline{b},b)$, $S(\underline{a},b)$, $S(\underline{b},a)\}$. However, $\db$ is a ``yes''-instance of $\cqa{q_1}$. This is because every repair that contains $R(\underline{a},a)$ or $R(\underline{b},b)$ will satisfy $q_1$, while a repair that contains neither of these facts must contain $R(\underline{a},b)$ and $R(\underline{b},a)$, which together also satisfy~$q_1$.
\qed 
\end{example}

\begin{figure}
	\centering
	$$
\begin{array}{cc}
\begin{array}{c|cc}
R & \underline{1} & 2\bigstrut\\\cline{2-3}
  & a & a\\
  & a & b\\\cdashline{2-3}
  & b & a\\
  & b & b
\end{array}
&
\begin{array}{c|cc}
S & \underline{1} & 2\bigstrut\\\cline{2-3}
  & a & a\\
  & a & b\\\cdashline{2-3}
  & b & a\\
  & b & b
\end{array}
\end{array}
$$
	\caption{An inconsistent database instance $\db$.}
	\label{tbl:instance}
\end{figure}

\begin{example}
\label{ex:intro2}
Take the self-join $q_1=\exists x\exists y\exists z(R(\underline{x},z)\land R(\underline{y},z))$ and its self-join-free counterpart $q_2=\exists x\exists y\exists z(R(\underline{x},z)\land S(\underline{y},z))$.  $\cqa{q_2}$ is known to be \coNP-complete, whereas it is easily verified that $\cqa{q_1}$ is in \FO, by observing that a database instance is a ``yes''-instance of $\cqa{q_{1}}$ if and only if it satisfies $\exists x\exists y(R(\underline{x},y))$.
\qed
\end{example}

This paper makes a contribution to the complexity classification of $\cqa{q}$ for conjunctive queries, possibly with self-joins, of the form
$$q=\exists x_{1}\dotsm\exists x_{k+1}\normalformula{R_{1}(\underline{x_{1}},x_{2})\land 
R_{2}(\underline{x_{2}},x_{3})\land\dotsm\land R_{k}(\underline{x_{k}},x_{k+1})},$$
which we call \emph{path queries}.
The primary key positions are underlined.
As will become apparent in our technical treatment, the classification of path queries is already very challenging, even though it is only a first step towards Conjecture~\ref{conj:dichotomy}, which is currently beyond reach. 
If all $R_{i}$'s are distinct (i.e., if there are no self-joins), then $\cqa{q}$ is known to be in \FO\ for path queries~$q$.
However, when self-joins are allowed, the complexity landscape of $\cqa{q}$ for path queries exhibits a tetrachotomy, as stated by the following main result of our paper.

\begin{theorem}\label{thm:path-tetrachotomy}
For each Boolean path query $q$, $\cqa{q}$ is in \FO, \NL-complete, \PTIME-complete, or \coNP-complete, and it is decidable in polynomial time in the size of~$q$ which of the four cases applies.
\end{theorem}

Comparing Theorem~\ref{thm:sjf-theorem} and Theorem~\ref{thm:path-tetrachotomy},
it is striking that there are path queries $q$ for which $\cqa{q}$ is \NL-complete or \PTIME-complete, whereas these complexity classes do not occur for self-join-free queries (under standard complexity assumptions).
So even for the restricted class of path queries, allowing self-joins immediately results in a more varied complexity landscape.

Let us provide some intuitions behind Theorem~\ref{thm:path-tetrachotomy} by means of examples. 
Path queries use only binary relation names.
A database instance $\db$ with binary facts can be viewed as a directed edge-colored graph: a fact $R(\underline{a},b)$ is a directed edge from $a$ to $b$ with color~$R$.
A repair of $\db$ is obtained by choosing, for each vertex,  precisely one outgoing edge among all outgoing edges of the same color. We will use the shorthand $q = RR$ to denote the path query $q=\exists x\exists y\exists z(R(\underline{x},y)\land R(\underline{y},z))$.

%\jef{Since ``pumping'' is known from the existing Pumping Lemma, I propose to use ``sweeping.'' I suggest to introduce it here, because it is central our complexity classification.}

In general, path queries can be represented by words over the alphabet of relation names.
Throughout this paper, relation names are in uppercase letters $R$, $S$, $X$, $Y$ etc., while lowercase letters $u$, $v$, $w$ stand for (possibly empty) words.
An important operation on words is dubbed \emph{rewinding}: if a word has a factor of the form $RvR$, then rewinding refers to the operation that replaces this factor with $RvRvR$.
That is, rewinding the factor $RvR$ in the word $uRvRw$ yields the longer word $uRvRvRw$.
For short, we also say that $uRvRw$ \emph{rewinds to} the word $\swipe{u}{Rv}Rw$, where we used concatenation $(\cdot)$ and underlining for clarity.
For example, $TWITTER$ rewinds to $\prefixswipe{TWI}{TTER}$, but also to $\prefixswipe{TWIT}{TER}$ and to $\swipe{TWI}{T}{TER}$.

Let $q_1 = RR$. It is easily verified that a database instance is a ``yes''-instance of $\cqa{q_1}$ if and only if it satisfies the following first-order formula:
$$\varphi = \exists x\normalformula{\exists y R(\underline{x},y)\land\forall y\normalformula{R(\underline{x},y) \rightarrow \exists z R(\underline{y},z)}}.$$ 
Informally, every repair contains an $R$-path of length $2$ if and only if there exists some vertex $x$ such that every repair contains a path of length $2$ starting in~$x$. 

Let $q_2 = RRX$, and consider the database instance in Figure~\ref{fig:example-rrx}. Since the only conflicting facts are $R(\underline{1},2)$ and $R(\underline{1},3)$, this database instance has two repairs.
Both repairs satisfy $RRX$, but unlike the previous example, there is no vertex $x$ such that every repair has a path colored $RRX$ that starts in~$x$.
Indeed, in one repair, such path starts in~$0$; in the other repair it starts in~$1$.
For reasons that will become apparent in our theoretical development, it is significant that both repairs have paths that start in~$0$ and are colored by a word in the regular language defined by $RR\kleene{R}{*}X$.
This is exactly the language that contains $RRX$ and is closed under the rewinding operation.
In general, it can be verified with some effort that a database instance is a ``yes''-instance of $\cqa{q_{2}}$ if and only if it contains some vertex $x$ such that every repair has a path that starts in~$x$ and is colored by a word in the regular language defined by $RR\kleene{R}{*}X$.
The latter condition can be tested in \PTIME\ (and even in \NL).
%Our automaton-based approach in Section~\ref{sec:key} is motivated by our finding that the complexity classification of Theorem~\ref{thm:path-tetrachotomy} can be well explained by shifting from path queries to regular expressions induced by the sweeping operator.

\begin{figure}[!ht]
      \centering
      \begin{tikzpicture}[->,>=stealth,auto=left, scale=1.5,vnode/.style={circle,black,inner sep=1pt,scale=1},el/.style = {inner sep=3/2pt}]
		  \node[vnode] (s1) at (0, 0) {$0$};
		  \node[vnode] (s2) at (1, 0) {$1$};
		  \node[vnode] (s3) at (1.5, 0.732) {$2$};
		  \node[vnode] (s4) at (2, 0) {$3$};
		  \node[vnode] (s5) at (3, 0) {$4$};
		  
		  \path[->] (s1) edge node[el] {$R$} (s2);
		  \path[->] (s2) edge node[el] {$R$} (s3);
		  \path[->] (s2) edge node[el] {$R$} (s4);
		  \path[->] (s3) edge node[el] {$R$} (s4);
		  \path[->] (s4) edge node[el] {$X$} (s5);
		\end{tikzpicture}
      \caption{An example database instance $\db$ for $q_2 = RRX$.}
      \label{fig:example-rrx}
\end{figure}
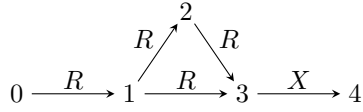

The situation is still different for $q_3 = ARRX$, for which it will be shown that $\cqa{q_{3}}$ is \coNP-complete.
Unlike our previous example, repeated rewinding of $ARRX$ into words of the language $ARR\kleene{R}{*}X$ is not helpful.
For example, in the database instance of Figure~\ref{fig:example-arrx},
every repair has a path that starts in~$0$ and is colored with a word in the language defined by $ARR\kleene{R}{*}X$.
However, the repair that contains $R(\underline{a},c)$ does not satisfy~$q_{3}$.
Unlike Figure~\ref{fig:example-rrx}, the ``bifurcation'' in Figure~\ref{fig:example-arrx} can be used as a gadget for showing \coNP-completeness in Section~\ref{sec:hardness}.

\begin{figure}[!ht]
      \centering
      \begin{tikzpicture}[->,>=stealth,auto=left, scale=2.5,vnode/.style={circle,black,inner sep=1pt,scale=1},el/.style = {inner sep=3/2pt}]
	    \node[vnode] (s1) at (-1, 0) {$0$};
	    \node[vnode] (leftmid) at (-0.5, 0) {};
	    \node[vnode] (s3) at (0, 0) {$a$};

	    \node[vnode] (mid1) at (0.5, 0.25) {$b$};
	    \node[vnode] (mid2) at (0.5, -0.25) {$c$};
	    \node[vnode] (mid3) at (1, -0.25) {};
	    
	  \node[vnode] (s4) at (1, 0.25) {};
	  \node[vnode] (s5) at (1.5, -0.25) {};
	  
	  \path[->] (s1) edge node[el,sloped] {$A$} (leftmid);
	  \path[->] (leftmid) edge node[el] {$R$} (s3);
	  \path[->] (s3) edge node[el] {$R$} (mid1);
	  \path[->] (mid1) edge node[el] {$X$} (s4);

	  \path[->] (s3) edge node[el,below] {$R$} (mid2);
	  \path[->] (mid2) edge node[el,below] {$R$} (mid3);
	  \path[->] (mid3) edge node[el,below] {$X$} (s5);
	\end{tikzpicture}
	\caption{An example database instance $\db$ for $q_3 = ARRX$.}
      \label{fig:example-arrx}
\end{figure}

%We remark that $\cqa{q}$ is always in \FO\ for self-join-free path queries by Theorem~\ref{thm:sjf-theorem}. It is therefore surprising that even for path queries, possibly with self-joins, $\cqa{q}$ already displays a tetrachotomy: it is either in \FO, \NL-complete, \PTIME-complete or \coNP-complete. It also shows two new complexity class within \PTIME\ that $\cqa{q}$ can be complete in, namely \NL-complete and \PTIME-complete. Therefore, consistent query answering on primary keys can be complete in at least five distinct complexity classes. 

%\begin{conjecture}[Pentachotomy Conjecture] \label{conj:general-pentachotomy}
%  Let $q$ be a boolean conjunctive query. Then $\cqa{q}$ is in \FO, \LSPACE-complete, \NL-complete, \PTIME-complete, or \coNP-complete, and it is decidable in polynomial time in the size of $q$ which of the five cases applies. 
%\end{conjecture}

\myparagraph{Organization} 
Section~\ref{sec:prelim} introduces the preliminaries. 
In Section~\ref{sec:classification}, the statement of Theorem~\ref{thm:main} gives the syntactic conditions for deciding the complexity of $\cqa{q}$ for path queries~$q$. 
To prove this theorem, we view the rewinding operator from the perspectives of regular expressions and automata, which are presented in Sections~\ref{sec:syntax} and~\ref{sec:key} respectively.
Sections~\ref{sec:algorithm} and~\ref{sec:hardness} present, respectively, complexity upper bounds and lower bounds of our classification. 
In Section~\ref{sec:constant}, we extend our classification result to path queries with constants. Section~\ref{sec:related-work} discusses related work, and Section~\ref{sec:conclusion} concludes this paper.

	\section{Preliminaries}
\label{sec:prelim}

We assume disjoint sets of {\em variables\/} and {\em constants\/}.
%If $\vec{x}$ is a sequence containing variables and constants, then $\sequencevars{\vec{x}}$ denotes the set of variables that occur in $\vec{x}$.
A {\em valuation\/} over a set $U$ of variables is a total mapping $\theta$ from~$U$ to the set of~constants.
%At several places, it is implicitly understood that such a valuation $\theta$ is extended to be the identity on constants and on variables not in $U$.
%If $V \subseteq U$, then $\theta[V]$ denotes the restriction of $\theta$ to $V$.

%If $\theta$ is a valuation over a set $U$ of variables, $x$ is a variable, and $a$ is a constant,
%then $\substitute{\theta}{x}{a}$ is the valuation over $U\cup\{x\}$ such that
%$\substitute{\theta}{x}{a}(x)=a$ and for every variable $y$ such that $y\neq x$, $\substitute{\theta}{x}{a}(y)=\theta(y)$.
%Notice that $x\in U$ is allowed.

\myparagraph{Atoms and key-equal facts}
We consider only $2$-ary \emph{relation names}, where the first position is called the \emph{primary key}.
If $R$ is a relation name, and $s,t$ are variables or constants,
then $R(\underline{s},t)$ is an \emph{atom}.
An atom without variables is a \emph{fact}.
Two facts are \emph{key-equal} if they use the same relation name and agree on the primary key.

\myparagraph{Database instances, blocks, and repairs}
A {\em database schema\/} is a finite set of relation names.
All constructs that follow are defined relative to a fixed database schema.

A {\em database instance\/} is a finite set $\db$ of facts using only the relation names of the schema.
%We refer to databases as ``uncertain databases'' to stress that such databases can violate primary key constraints.
%
We write $\adom{\db}$ for the active domain of $\db$ (i.e., the set of constants that occur in $\db$).
A {\em block\/} of $\db$ is a maximal set of key-equal facts of $\db$.
%The term $R$-block refers to a block of $R$-facts, i.e., facts with relation name $R$.
Whenever a database instance $\db$ is understood,
we write $R(\underline{c},*)$ for the block that contains all facts with relation name~$R$ and primary-key value~$c$.
%If $A$ is a fact of $\db$, then $\theblock{A}{\db}$ denotes the block of $\db$ that contains $A$.
A database instance $\db$ is {\em consistent\/} if it contains no two distinct facts that are key-equal (i.e., if no block of $\db$ contains more than one fact).
A {\em repair\/} of $\db$ is an inclusion-maximal consistent subset of $\db$. 
%
%We write $\repairs{\db}$ for the set of repairs of $\db$.

\myparagraph{Boolean conjunctive queries}
%A {\em Boolean query\/} is a mapping $q$ that associates a Boolean (true or false) to each uncertain database, such that $q$ is closed under isomorphism.
%We write $\db\models q$ to denote that $q$ associates true to $\db$, in which case $\db$ is said to {\em satisfy\/} $q$.
%A {\em Boolean first-order query\/} is a Boolean query that can be defined in first-order logic (with equality and constants, but without other built-in predicates).
A {\em Boolean conjunctive query\/} is a finite set 
$q=\{R_{1}(\underline{{x}_{1}},{y}_{1})$, $\dots$, $R_{n}(\underline{{x}_{n}},{y}_{n})\}$ of atoms.
We denote by $\queryvars{q}$ the set of variables that occur in $q$.
The set $q$ represents the first-order sentence 
$$\exists u_{1}\dotsm\exists u_{k}\normalformula{R_{1}(\underline{{x}_{1}},{y}_{1})\land\dotsm\land R_{n}(\underline{{x}_{n}},{y}_{n})},$$ where $\{u_{1}, \dots, u_{k}\}=\queryvars{q}$.
%We write $\db\models q$ to denote that $q$ is \emph{satisfied} by database instance $\db$, under standard first-order logic semantics.

We say that a Boolean conjunctive query $q$ has a {\em self-join\/} if some relation name occurs more than once in $q$.
A conjunctive query without self-joins is called {\em self-join-free\/}.
%If $q$ has no self-join, then it is called {\em self-join-free\/}.
%By a little abuse of notation, we may confuse atoms with their relation names in a self-join-free Boolean conjunctive query $q$.
%That is, if we use a relation name $R$ at places where an atom is expected,
%then we mean the (unique) $R$-atom of~$q$.

%If $q$ is a Boolean conjunctive query, $\vec{x}=\tuple{x_{1},\dots,x_{\ell}}$ is a sequence of distinct variables that occur in $q$, and $\vec{a}=\tuple{a_{1},\dots,a_{\ell}}$ is a sequence of constants,
%then $\substitute{q}{\vec{x}}{\vec{a}}$ denotes the query obtained from $q$ by replacing all occurrences of $x_{i}$ with $a_{i}$, for all $1\leq i\leq\ell$.

\myparagraph{Consistent query answering}
For every Boolean conjunctive query $q$, the decision problem $\cqa{q}$ takes as input a database instance $\db$, and asks whether $q$ is satisfied by every repair of~$\db$.
%
%It is easy to show the following upper bound on the complexity of $\cqa{q}$.
It is straightforward that for every Boolean conjunctive query~$q$, $\cqa{q}$ is in \coNP.

\myparagraph{Path queries} A {\em path query} is a Boolean conjunctive query without constants of the following form:
\begin{equation*}\label{eq:pq}
q = \{R_{1}(\underline{x_{1}},x_{2}), R_{2}(\underline{x_{2}},x_{3}), \dots, R_{k}(\underline{x_{k}},x_{k+1}) \},
\end{equation*}
where $x_{1}$, $x_{2}$,\dots, $x_{k+1}$ are distinct variables, and $R_{1}$, $R_{2}$,\dots, $R_{k}$ are (not necessarily distinct) relation names.
It will often be convenient to denote this query as a {\em word} $R_{1} R_{2} \dotsm R_{k}$ over the alphabet of relation names.
This ``word'' representation is obviously lossless up to a variable renaming.
Importantly, path queries may have self-joins, i.e., a relation name may occur multiple times. 
% \paris{Don't forget to mention constants if needed}
Path queries containing constants will be discussed in Section~\ref{sec:constant}.
The treatment of constants is significant, because it allows moving from Boolean to non-Boolean queries, by using that free variables behave like constants.
%It will be often convenient  to view a path query equivalently as a {\em word} $R_{1} R_{2} \dots R_{k}$ over the alphabet of relation names.

\section{The Complexity Classification}
\label{sec:classification}
We define syntactic conditions $\cone$, $\ctwo$, and $\cthree$ for path queries~$q$. 
Let $R$ be any relation name in $q$, and let $u, v, w$ be (possibly empty) words over the alphabet of relation names of $q$.
%\jef{Is the word ``consecutive'' really needed in $\ctwo$?}

\begin{description}
\item[$\cone$:] Whenever $q = uRvRw$, $q$ is a prefix of  $uRvRvRw$.
\item[$\ctwo$:] Whenever $q = uRvRw$, $q$ is a factor of  $uRvRvRw$; and whenever $q= u Rv_1 Rv_2 Rw$ for consecutive occurrences of $R$,  $v_1 = v_2$ or $Rw$ is a prefix of $Rv_1$.
\item[$\cthree$:] Whenever $q = uRvRw$, $q$ is a factor of  $uRvRvRw$.
\end{description}
It is instructive to think of these conditions in terms of the rewinding operator introduced in Section~\ref{sec:introduction}:
$\cone$ is tantamount to saying that~$q$ is a prefix of every word to which $q$ rewinds;
and $\cthree$ says that $q$ is a factor of every word to which $q$ rewinds. 
These conditions can be checked in polynomial time in the length of the word $q$.
The following result has an easy proof.

\begin{proposition}
Let $q$ be a path query. 
If $q$ satisfies $\cone$, then $q$ satisfies $\ctwo$; and if $q$ satisfies $\ctwo$, then $q$ satisfies $\cthree$.
\end{proposition}
%\begin{proof}
%If $q$ satisfies $\ctwo$, then $q$ satisfies $\cthree$ by definition.
%%%
%For the other implication, assume that $q$ satisfies $\cone$ and $q=uRv_1Rv_2Rw$ where the occurrences of $R$ are consecutive. 
%Since $q$ satisfies~$\cone$, it is a prefix of $uRv_{1}Rv_{1}Rv_{2}Rw$, and therefore $v_{1}=v_{2}$. Thus $q$ satisfies~$\ctwo$.
%\end{proof}

The main part of this paper comprises a proof of the following theorem, which refines the statement of Theorem~\ref{thm:path-tetrachotomy} by adding syntactic conditions.
The theorem is illustrated by Example~\ref{ex:main}. 

\begin{theorem}
  \label{thm:main}
For every path query~$q$, the following complexity upper bounds obtain:
  \begin{itemize}
    \item if $q$ satisfies $\cone$, then $\cqa{q}$ is in \FO;
    \item if $q$ satisfies $\ctwo$, then $\cqa{q}$ is in \NL; and
    \item if $q$ satisfies $\cthree$, then $\cqa{q}$ is in \PTIME.
  \end{itemize}
Moreover, for every path query $q$, the following complexity lower bounds obtain:
  \begin{itemize}
    \item if $q$ violates $\cone$, then $\cqa{q}$ is \NL-hard;
    \item if $q$ violates $\ctwo$, then $\cqa{q}$ is \PTIME-hard; and
    \item if $q$ violates $\cthree$, then $\cqa{q}$ is \coNP-complete.
  \end{itemize}
\end{theorem}

%The following example illustrates Theorem~\ref{thm:main}.

\begin{example}\label{ex:main}
The query $q_{1}=RXRX$ rewinds to (and only to) $\prefixswipe{RX}{RX}$ and $\swipe{R}{XR}{X}$, which both contain $q_{1}$ as a prefix. It is correct to conclude that $\cqa{q_{1}}$ is in~\FO.

The query $q_{2}=RXRY$ rewinds only to $\prefixswipe{RX}{RY}$, which contains $q_{2}$ as a factor, but not as a prefix. Therefore, $q_{2}$ satisfies $\cthree$, but violates $\cone$.
Since $q_{2}$ vacuously satisfies $\ctwo$ (because no relation name occurs three times in~$q_{2}$), it is correct to conclude that $\cqa{q_{2}}$ is \NL-complete.

The query $q_{3}=RXRYRY$ rewinds to
$\prefixswipe{RX}{RYRY}$, to
$\prefixswipe{RXRY}{RY}$, and to
$\swipe{RX}{RY}{RY}=\swipe{RXR}{YR}{Y}$.
Since these words contain $q_{3}$ as a factor, but not always as a prefix, we have that $q_{3}$ satisfies $\cthree$ but violates $\cone$.
It  can be verified that $q_{3}$ violates $\ctwo$ by writing it as follows:
$$
q_3 = \underbrace{\emptyword}_{u} \underbrace{\boldsymbol{R}X}_{Rv_1} \underbrace{\boldsymbol{R}Y}_{Rv_2} \underbrace{\boldsymbol{R}Y}_{Rw} $$
We have $X=v_1\neq v_2=Y$, but $Rw=RY$ is not a prefix of $Rv_1=RX$. 
Thus,  $\cqa{q_3}$ is \PTIME-complete.

Finally, the path query $q_{4}=RXRXRYRY$ rewinds, among others, to $\swipe{RX}{RXRY}{RY}$, which does not contain $q_{4}$ as a factor.
It is correct to conclude that $\cqa{q_{4}}$ is~\coNP-complete.
\qed
\end{example}

\section{Regexes for $\cone$, $\ctwo$, and $\cthree$}\label{sec:syntax}

In this section, we show that the conditions $\cone$, $\ctwo$, and $\cthree$ can be captured by regular expressions (or regexes) on path queries, which will be used in the proof of Theorem~\ref{thm:main}.  
Since these results are within the field of \emph{combinatorics of words}, we will use the term \emph{word} rather than \emph{path query}. 
%\jef{I think the term ``word'' is appropriate here, to make clear that it is not (only) about (conjunctive) queries.}

\begin{definition}\label{def:bs}
We define four properties $\bone$, $\btwoa$, $\btwob$, $\bthree$ that a word $q$ can possess:
\begin{description}
\item[$\bone$:]
For some integer $k\geq 0$,
there are words $v$, $w$ such that $vw$ is self-join-free and $q$ is a prefix of $w\kleene{v}{k}$.
\item[$\btwoa$:]
For some integers $j,k\geq 0$,
there are words $u$, $v$, $w$ such that $uvw$ is self-join-free and $q$ is a factor of $\kleene{u}{j}w\kleene{v}{k}$.
\item[$\btwob$:] 
For some integer $k\geq 0$,
there are words $u$, $v$, $w$ such that $uvw$ is self-join-free and $q$ is a factor of $\kleene{uv}{k}wv$.
\item[$\bthree$:]
For some integer $k\geq 0$,
there are words $u$, $v$, $w$ such that $uvw$ is self-join-free and $q$ is a factor of $uw\kleene{uv}{k}$.
\qed
\end{description}
%\qed
\end{definition}

%Note that the conditions $\bone$, $\btwoa$, $\btwob$ and $\bthree$ are not mutually exclusive. For example, $\bone$ logically implies $\btwoa$. 
%%%
We can identify each condition among $\cone$, $\ctwo$, $\cthree$, $\bone$, $\btwoa$, $\btwob$,  $\bthree$ with the set of all words satisfying this condition.
Note then that $\bone\subseteq\btwoa\cap\bthree$.
The results in the remainder of this section can be summarized as follows:
\begin{itemize}
\item $\cone=\bone$ (Lemma~\ref{lem:fofo})
\item $\ctwo=\btwoa\cup\btwob$ (Lemma~\ref{lem:notctwo})
\item $\cthree=\btwoa\cup\btwob\cup\bthree$ (Lemma~\ref{lem:forms})
\end{itemize}
Moreover, Lemma~\ref{lem:notctwo} characterizes $\cthree\setminus\ctwo$. 

\begin{lemma}\label{lem:fofo}
For every word $q$, the following are equivalent:
\begin{enumerate}
\item\label{it:focone}
$q$ satisfies $\cone$; and
\item\label{it:fobone}
$q$ satisfies $\bone$.
\end{enumerate}
\end{lemma}

\begin{lemma}\label{lem:forms}
For every word $q$, the following are equivalent:
\begin{enumerate}
\item\label{it:pumping}
$q$ satisfies $\cthree$; and
\item\label{it:forms}
$q$ satisfies $\btwoa$, $\btwob$, or $\bthree$.
\end{enumerate}
\end{lemma}

\begin{definition}[First and last symbol]
For a nonempty  word $u$, we write $\first{u}$ and $\last{u}$ for, respectively, the first and the last symbol of $u$.
\qed
\end{definition}

\begin{lemma}\label{lem:notctwo}
Let $q$ be a word that satisfies $\cthree$.
Then, the following three statements are equivalent:
\begin{enumerate}
\item\label{it:notctwo}
$q$ violates $\ctwo$;
\item\label{it:falsifies}
$q$ violates both~$\btwoa$ and~$\btwob$; and
\item\label{it:factors}
there are words $u$, $v$, $w$ with $u\neq\emptyword$ and $uvw$ self-join-free such that either
\begin{enumerate}
\item\label{it:hercules}
$v\neq\emptyword$ and $\last{u}\cdot wuvu\cdot\first{v}$ is a factor of $q$; or
\item\label{it:tritan}
$v=\emptyword$, $w\neq\emptyword$, and $\last{u}\cdot w\kleene{u}{2}\cdot\first{u}$ is a factor of $q$. 
\end{enumerate}
\end{enumerate}
\end{lemma}

The shortest word of the form~\eqref{it:hercules} in the preceding lemma is $RRSRS$ (let $u=R$, $v=S$, and $w=\emptyword$), and the shortest word of the form~\eqref{it:tritan} is $RSRRR$ (let $u=R$, $v=\emptyword$, and $w=S$).
Note that since each of $\ctwo$, $\btwoa$, and $\btwob$ implies $\cthree$, it is correct to conclude that the equivalence between the first two items in Lemma~\ref{lem:notctwo} does not need the hypothesis that $q$ must satisfy $\cthree$.

\section{Automaton-Based Perspective}
\label{sec:key}

In this section, we prove an important lemma, Lemma~\ref{lem:min-start}, which will be used for proving the complexity upper bounds in Theorem~\ref{thm:main}. 

\subsection{From Path Queries to Finite Automata}

We can view a path query $q$ as a word where the alphabet is the set of relation names. We now associate each path query $q$ with a nondeterministic finite automaton (NFA), denoted $\nfa{q}$.

\begin{definition}[$\nfa{q}$]\label{def:nfa}
Every word $q$ gives rise to a nondeterministic finite automaton (NFA) with $\emptyword$-moves, denoted $\nfa{q}$, as follows.
\begin{description}
\item[States:]
The set of states is the set of prefixes of $q$. The empty word $\emptyword$ is a prefix of~$q$.
%%%
\item[Forward transitions:]
If $u$ and $uR$ are states, then there is a transition with label $R$ from state $u$ to state $uR$.
These transitions are said to be \emph{forward}. 
\item[Backward transitions:]
If $uR$ and $wR$ are states such that $\card{u}<\card{w}$ (and therefore $uR$ is a prefix of $w$),
then there is a transition with label~$\emptyword$ from state $wR$ to state $uR$.
These transitions are said to be \emph{backward}, and capture the operation dubbed rewinding. 
\item[Initial and accepting states:]
The initial state is~$\emptyword$ and the only accepting state is $q$.
\qed
\end{description}
%\qed
\end{definition}

\begin{figure*}[t]\centering
\begin{tikzpicture}[shorten >=1pt,node distance=2cm,on grid,auto] 
   \node[state,initial] (q_0)   {$\emptyword$}; 
   \node[state] (q_1) [right=of q_0] {$R$}; 
   \node[state] (q_2) [right=of q_1] {$RX$}; 
   \node[state] (q_3) [right=of q_2] {$RXR$}; 
   \node[state] (q_4) [right=of q_3] {$RXRR$};
   \node[state,accepting] (q_5) [right=of q_4] {$RXRRR$};
   
    \path[->] 
    (q_0) edge  node {$R$} (q_1)
    (q_1) edge  node {$X$} (q_2)
    (q_2) edge  node {$R$} (q_3)
    (q_3) edge  node {$R$} (q_4)
    (q_4) edge  node {$R$} (q_5)
    (q_3) edge[bend left,in=130,out=50]  node[above] {$\varepsilon$} (q_1)
    (q_4) edge[bend left,in=130,out=50]  node[above] {$\varepsilon$} (q_1)
    (q_4) edge[bend left,in=130,out=50]  node[above] {$\varepsilon$} (q_3)
    (q_5) edge[bend left,in=130,out=50]  node[above] {$\varepsilon$} (q_1)
    (q_5) edge[bend left,in=230,out=310]  node {$\varepsilon$} (q_3)
    (q_5) edge[bend left,in=130,out=50]  node[above] {$\varepsilon$} (q_4);
\end{tikzpicture}
\caption{The $\nfa{q}$ automaton for the path query $q=RXRRR$.}
\label{fig:nfa}
\end{figure*}
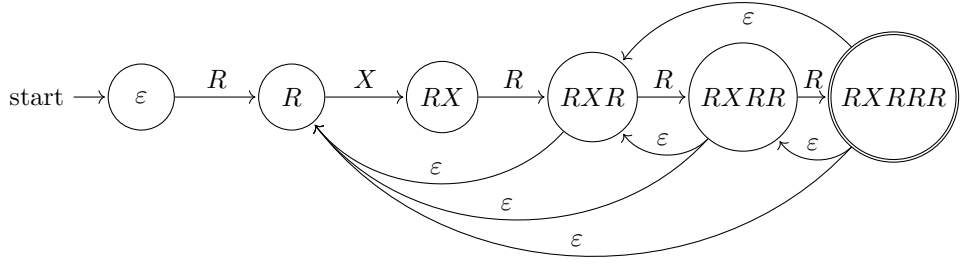

Figure~\ref{fig:nfa} shows the automaton $\nfa{RXRRR}$. Informally, the forward transitions capture the automaton that would accept the word $RXRRR$, while the backward transitions capture the existence of self-joins that allow an application of the rewind operator. 
%%%
We now take an alternative route for defining the language accepted by ~$\nfa{q}$, which straightforwardly results in  Lemma~\ref{lem:pumpingnfa}.
Then, Lemma~\ref{lem:pumpclosure} gives alternative ways for expressing $\cone$ and $\cthree$. 

%\jef{Maybe, we should say again that $\pumpclosure{q}$ is a language over the alphabet of relation names. So $q$ in the following definition uses the ``word'' representation of a path query.}

%\jef{Changed the symbol $\star$, to make it different from Kleene star.}

\begin{definition}\label{def:pumpclosure}
Let $q$ be a path query, represented as a word over the alphabet of relation names.
We define the language $\pumpclosure{q}$ as the smallest set of words such that
\begin{enumerate}[label=(\alph*)]
\item 
$q$ belongs to $\pumpclosure{q}$; and
\item \label{it:pump}
\emph{Rewinding:}
if $uRvRw$ is in $\pumpclosure{q}$ for some relation name $R$ and (possibly empty) words $u, v, w$,
then $uRvRvRw$ is also in $\pumpclosure{q}$.
\qed
\end{enumerate}
%\qed
\end{definition}
That is,  $\pumpclosure{q}$ is the smallest language that contains~$q$ and is closed under rewinding.

\begin{lemma}\label{lem:pumpingnfa}
For every path query $q$, the automaton $\nfa{q}$ accepts the language $\pumpclosure{q}$.
\end{lemma}

\begin{lemma} \label{lem:pumpclosure}
Let $q$ be a path query. Then,
\begin{enumerate}
\item\label{it:lc1} 
$q$ satisfies $\cone$ if and only if $q$ is a prefix of each $p \in \pumpclosure{q}$;
\item\label{it:lc3}
$q$ satisfies $\cthree$ if and only if $q$ is a factor of each $p \in \pumpclosure{q}$.
\end{enumerate}
\end{lemma}
\begin{proof}
\framebox{$\impliedby$ in~\eqref{it:lc1} and~\eqref{it:lc3}}
This direction  is trivial, because whenever $q = u Rv Rw$, we have that $u Rv Rv Rw \in\pumpclosure{q}$. 

We now show the $\implies$ direction in both items.
To this end, we call an application of the rule~\ref{it:pump} in Definition~\ref{def:pumpclosure} a \emph{rewind}. 
By construction, each word in $\pumpclosure{q}$ can be obtained from $q$ by using $k$~rewinds, for some nonnegative integer~$k$.
Let $q_k$ be a word in $\pumpclosure{q}$ that can be obtained from~$q$ by using $k$~rewinds. 

\framebox{$\implies$ in~\eqref{it:lc1}}
We use induction on $k$ to show that $q$ is a prefix of~$q_k$.
For he induction basis, $k=0$, we have that~$q$ is a prefix of $q_0 = q$.
We next show the induction step $k\rightarrow k+1$.
%\begin{description}
%\item[Basis.] We have $q_0 = q$, which contains $q$ as a prefix.
%\item[Induction step.] 
Let $q_{k+1} = u Rv Rv Rw$ where $q_{k} = u Rv Rw$ is a word in $\pumpclosure{q}$ obtained with $k$~rewinds. By the induction hypothesis, we can assume a word $s$ such that $q_k =q\cdot s$.

    \begin{itemize}
      \item If $q$ is a prefix of $u RvR$, then $q_{k+1} = u Rv  Rv Rw$ trivially contains $q$ as a prefix;
      \item otherwise $u Rv R$ is a proper prefix of $q$. Let $q = u  Rv Rt$ where $t$ is nonempty. Since $q$ satisfies $\cone$,  $Rt$ is a prefix of $Rv$. Then $q_{k+1} = u Rv Rv Rw$ contains $q = u \cdot Rv \cdot Rt$ as a prefix. 
    \end{itemize}
%\end{description}

\framebox{$\implies$ in~\eqref{it:lc3}}
We use induction on $k$ to show that $q$ is a factor of~$q_k$.
For the induction basis, $k=0$, we have that $q$ is a prefix of $q_0 = q$.
%\begin{description}
%\item[Basis.] We have $q_0 = q$, which contains $q$ as a factor.
%\item[Induction step.] 
For the induction step, $k\rightarrow k+1$, let $q_{k+1} = u Rv Rv Rw$ where $q_{k} = u Rv Rw$ is a word in $\pumpclosure{q}$ obtained with $k$~rewinds. By the induction hypothesis, $q_k = u  Rv Rw$ contains $q$ as a factor. If $q$ is a factor of either $u Rv R$ or $Rv  Rw$, then $q_{k+1} = u Rv  Rv Rw$ contains $q$ as a factor. Otherwise, we can decompose $q_{k} = u^- q^-  Rv  Rq^+  w^+$ where $q = q^-  Rv  R  q^+$, $u = u^- q^-$ and $w = q^+w^+$. Since $q$ satisfies $\cthree$, the word $q^-  Rv Rv  R  q^+$, which is a factor of~$q_{k+1}$, contains $q$ as a~factor. 
 %\end{description}
\end{proof}

In the technical treatment, it will be convenient to consider the automaton obtained from $\nfa{q}$ by changing its start state, as defined next.

\begin{definition}\label{def:snfa}
If $u$ is a prefix of $q$ (and thus $u$ is a state in $\nfa{q}$),
then $\snfa{q}{u}$ is the automaton obtained from $\nfa{q}$ by letting the initial state be $u$ instead of the empty word. 
%Informally, $\snfa{q}{u}$ accepts suffixes of words accepted by $\nfa{q}$. 
Note that $\snfa{q}{\emptyword}=\nfa{q}$.
It may be helpful to think of the first $\mathsf{S}$ in $\snfa{q}{u}$ as ``$\underline{\mathsf{S}}$tart at $u$.''
\qed
\end{definition}

\subsection{Reification Lemma}

In this subsection, we first define how an automaton executes on a database instance.
We then state an helping lemma which will be used in the proof of Lemma~\ref{lem:min-start}, which constitutes the main result of Section~\ref{sec:key}.
To improve the readability and logical flow of our presentation,  
we postpone the proof of the helping lemma to~Section~\ref{sec:proofkey}.

\begin{definition}[Automata on database instances]
Let $\db$ be a database instance.
A \emph{path (in $\db$)} is defined as a sequence $R_{1}(\underline{c_{1}},c_{2})$, $R_{2}(\underline{c_{2}},c_{3})$, \dots, $R_{n}(\underline{c_{n}},c_{n+1})$  of facts in $\db$. 
Such a path is said to \emph{start in $c_{1}$}.
We call $R_{1}R_{2}\dotsm R_{n}$ the \emph{trace} of this path.
A path is said to be \emph{accepted} by an automaton if its trace is accepted by the automaton.

Let $q$ be a path query and $\rep$ be a consistent database instance.
We define $\starttwo{q}{\rep}$ as the set containing all (and only) constants $c\in\adom{\rep}$ such that there is a path in $\rep$ that starts in~$c$ and is accepted by $\nfa{q}$.
\qed
\end{definition}

\begin{example}\label{ex:trace}
Consider the query $q_{2}=RRX$ and the database instance of Figure~\ref{fig:example-rrx}.
Let $\rep_{1}$ and $\rep_{2}$ be the repairs containing, respectively, $R(\underline{1},2)$ and $R(\underline{1},3)$.
The only path with trace $RRX$ in $\rep_{1}$ starts in~$1$; and
the only path with trace $RRX$ in $\rep_{2}$ starts in~$0$.
%%%
The regular expression for $\pumpclosure{q}$ is $RR\kleene{R}{*}X$.
We have $\starttwo{q}{\rep_{1}}=\{0,1\}$ and $\starttwo{q}{\rep_{2}}=\{0\}$.
\qed
\end{example}

The following lemma tells us that, among all repairs, there is one that is inclusion-minimal with respect to $\starttwo{q}{\cdot}$.
In the preceding example, the repair $\rep_{2}$ minimizes~$\starttwo{q}{\cdot}$.

\begin{lemma}\label{lem:key}
Let $q$ be a path query, and $\db$ a database instance.
There exists a repair $\rep^{*}$ of $\db$ such that for every repair $\rep$ of $\db$,
$\starttwo{q}{\rep^{*}}\subseteq\starttwo{q}{\rep}$.
\end{lemma}

%Example~\ref{ex:trace} illustrates that the above lemma does not hold true if we replace $\nfa{q}$ by the automaton that accepts only~$q$. 
%Lemma~\ref{lem:key} is helpful for solving $\cqa{q}$, as will become apparent in 
Informally, we think of the next Lemma~\ref{lem:min-start} as a \emph{reification lemma}.
The notion of \emph{reifiable variable} was coined in~\cite[Definition~8.5]{10.1145/2188349.2188351}, to refer to a variable~$x$ in a query~$\exists x\formula{\varphi(x)}$ such that whenever that query is true in every repair of a database instance, then there is a constant $c$ such that $\varphi(c)$ is true in every repair.
The following lemma captures a very similar concept.

\begin{lemma}[Reification Lemma for $\cthree$] \label{lem:min-start}
Let $q$ be a path query that satisfies $\cthree$.
Then, for every database instance $\db$, the following are equivalent:
\begin{enumerate}
\item\label{it:rifidb}
 $\db$ is a ``yes''-instance of $\cqa{q}$; and
\item \label{it:rificonstant}
there exists a constant $c$ (which depends on $\db$) such that for every repair $\rep$ of $\db$,
$c\in\starttwo{q}{\rep}$.
\end{enumerate}
\end{lemma}
\begin{proof}
\framebox{\ref{it:rifidb}$\implies$\ref{it:rificonstant}}
Assume~\eqref{it:rifidb}.
By Lemma~\ref{lem:key}, there exists a repair $\rep^{*}$ of $\db$ such that for every repair $\rep$ of $\db$, $\starttwo{q}{\rep^{*}}\subseteq\starttwo{q}{\rep}$.
Since $\rep^{*}$ satisfies $q$, there is a path $R_{1}(\underline{c_{1}},c_{2})$, $R_{2}(\underline{c_{2}},c_{3})$, \dots, $R_{n}(\underline{c_{n}},c_{n+1})$ in $\rep^{*}$ such that $q=R_{1}R_{2}\dotsm R_{n}$.
Since~$q$ is accepted by $\nfa{q}$, we have $c_{1}\in\starttwo{q}{\rep^{*}}$.
It follows that $c_{1}\in\starttwo{q}{\rep}$ for every repair $\rep$ of $\db$.

\framebox{\ref{it:rificonstant}$\implies$\ref{it:rifidb}}
Let $\rep$ be any repair of $\db$. 
By our hypothesis that~\eqref{it:rificonstant} holds true, there is some $c\in\starttwo{q}{\rep}$.
Therefore, there is a path in $\rep$ that starts in $c$ and is accepted by $\nfa{q}$.
Let $p$ be the trace of this path.
By Lemma~\ref{lem:pumpingnfa}, $p\in\pumpclosure{q}$.
%By Lemma~\ref{lem:pumpingnfa}, $q\in\pumpclosure{q}$.
Since $q$ satisfies $\cthree$ by the hypothesis of the current lemma, it follows by Lemma~\ref{lem:pumpclosure} that $q$ is a factor of~$p$.
Consequently, there is a path in $\rep$ whose trace is~$q$.
It follows that $\rep$ satisfies~$q$.
\end{proof}

\subsection{Proof of Lemma~\ref{lem:key}}\label{sec:proofkey}

We will use the following definition.

\begin{definition}[States Set]\label{def:statesset}
This definition is relative to a path query~$q$.
Let $\rep$ be a consistent database instance, and let $f$ be an $R$-fact in $\rep$, for some relation name~$R$. 
The \emph{states set} of $f$ in $\rep$, denoted $\sset{f}{\rep}{q}$, is defined as the smallest set of states satisfying the following property, for all prefixes $u$ of $q$:
\begin{quote}
if $\snfa{q}{u}$ accepts a path in $\rep$ that starts with $f$,
then $uR$ belongs to $\sset{f}{\rep}{q}$.
\end{quote}
Note that if $f$ is an $R$-fact, then all states in $\snfa{q}{\rep}$ have~$R$ as their last relation name.
\qed
\end{definition}

\begin{example}
Let $q=RRX$ and $\rep=\{R(\underline{a},b)$, $R(\underline{b},c)$, $R(\underline{c},d)$, $X(\underline{d},e)$, $R(\underline{d},e)\}$.
Then  $\nfa{q}$ has states $\{\emptyword, R, RR, RRX\}$ and accepts the regular language $RR\kleene{R}{*}X$.
Since $\snfa{q}{\emptyword}$ accepts the path $R(\underline{b},c)$, $R(\underline{c},d)$, $X(\underline{c},d)$, the states set $\sset{R(\underline{b},c)}{\rep}{q}$ contains $\emptyword\cdot R=R$.
Since the latter path is also accepted by  $\snfa{q}{R}$, we also have $R\cdot R\in\sset{R(\underline{b},c)}{\rep}{q}$.
Finally, note that $\sset{R(\underline{d},e)}{\rep}{q}=\emptyset$,
because there is no path that contains $R(\underline{d},e)$ and is accepted by $\nfa{q}$. 
\qed
\end{example}

\begin{lemma}\label{lem:sufclo}
Let $q$ be a path query, and $\rep$ a consistent database instance.
If  $\sset{f}{\rep}{q}$  contains state $uR$, then it contains every state of the form $vR$ with $\card{v}\geq\card{u}$.
\end{lemma}

\begin{proof}
Assume $uR\in\sset{f}{\rep}{q}$. Then $f$ is an $R$-fact and there is a path $f\cdot\pi$ in $\rep$ that is accepted by $\snfa{q}{u}$. 
Let $vR$ be a state with $\card{v}>\card{u}$.
Thus, by construction, $\nfa{q}$ has a backward transition with label $\emptyword$ from state $vR$ to state $uR$.

It suffices to show that $f\cdot\pi$ is accepted by $\snfa{q}{v}$.
Starting in state $v$,  $\snfa{q}{v}$ traverses $f$ (reaching state $vR$) and then uses the backward transition (with label $\emptyword$) to reach the state $uR$.
From there on,  $\snfa{q}{v}$ behaves like $\snfa{q}{u}$.
\end{proof}

From Lemma~\ref{lem:sufclo}, it follows that $\sset{f}{\rep}{q}$ is completely determined by the shortest word in it.

\begin{definition}
Let $q$ be a path query and $\db$ a database instance.
For every fact $f\in\db$, we define:
$$
\csset{f}{\db}{q}\defeq\bigcap\{\sset{f}{\rep}{q}\mid\textnormal{$\rep$ is a repair of $\db$ that contains $f$}\},
$$
where $\bigcap X = \bigcap_{S \in X} S$.
\qed
\end{definition}

It is to be noted here that whenever $\rep_{1}$ and $\rep_{2}$ are repairs containing~$f$, then by Lemma~\ref{lem:sufclo}, $\sset{f}{\rep_{1}}{q}$ and $\sset{f}{\rep_{2}}{q}$ are comparable by set inclusion.
Therefore, informally, $\csset{f}{\db}{q}$ is the $\subseteq$-minimal states set of $f$ over all repairs that contain~$f$.

%However, for theoretical clarity, we keep the entire set.
%Example~\ref{ex:ssempty} shows that this set can be empty.

\begin{definition}[Preorder $\mypreceq{q}$ on repairs]
Let $\db$ be a database instance.
For all repairs $\rep,\sep$ of $\db$, we define $\rep\mypreceq{q}\sep$ if 
for every $f\in\rep$ and $g\in\sep$ such that $f$ and $g$ are key-equal, 
we have $\sset{f}{\rep}{q}\subseteq\sset{g}{\sep}{q}$.
 
Clearly, $\mypreceq{q}$ is a reflexive and transitive binary relation on the set of repairs of $\db$.
We write $\rep\myprec{q}\sep$ if both $\rep\mypreceq{q}\sep$ and for some $f\in\rep$ and $g\in\sep$ such that $f$ and $g$ are key-equal, 
$\sset{f}{\rep}{q}\subsetneq\sset{g}{\sep}{q}$.
\qed
\end{definition}

%\jef{I simplified the proof of the following lemma.}

%We start with a new definition relative to a Boolean path query $q$ and a database instance $\db$.

\begin{lemma}\label{lem:shrinking}
Let $q$ be a path query.
For every database instance $\db$,
there is a repair $\rep^{*}$ of $\db$ such that for every repair $\rep$ of $\db$, 
$\rep^{*} \mypreceq{q} \rep$.
\end{lemma}
\begin{proof}
Construct a repair $\rep^{*}$ as follows.
For every block $\block$ of  $\db$, insert into $\rep^{*}$ a fact $f$ of $\block$ such that $\csset{f}{\db}{q}=\bigcap\{\csset{g}{\db}{q}\mid g\in\block\}$.
More informally, we insert into $\rep^{*}$ a fact $f$ from $\block$ with a states set that is $\subseteq$-minimal over all repairs and all facts of $\block$.
We first show the following claim.

\begin{claim}
\label{claim:smallest}
For every fact $f$ in $\rep^{*}$, we have
$\sset{f}{\rep^{*}}{q}=\csset{f}{\db}{q}$.
\end{claim}
\begin{proof}
Let $f_{1}$ be an arbitrary fact in $\rep^{*}$.
We show $\sset{f_{1}}{\rep^{*}}{q}=\csset{f_{1}}{\db}{q}$.

\framebox{$\supseteq$}
Obvious, because $\rep^{*}$ is itself a repair of $\db$ that contains~$f_{1}$.

\framebox{$\subseteq$}
Let $f_{1}=R_{1}(\underline{c_0},c_{1})$.
Assume by way of a contradiction that there is $p_{1}\in\sset{f_{1}}{\rep^{*}}{q}$ such that $p_{1}\notin\csset{f_{1}}{\db}{q}$.
Then, for some (possibly empty) prefix~$p_{0}$ of $q$, there is a sequence:
\begin{equation}
p_{0}
\step{\emptyword} p_{0}' 
\verylongstep{f_{1}=R_{1}(\underline{c_{0}}, c_{1})}p_{1}
\step{\emptyword} p_{1}'
\verylongstep{f_{2}=R_{2}(\underline{c_{1}}, c_{2})}p_{2} 
\quad
\dotsm
\quad 
p_{n-1}
\step{\emptyword} p_{n-1}'
\verylongstep{f_{n}=R_{n}(\underline{c_{n-1}}, c_{n})}p_{n} = q,
\label{eq:trace}
\end{equation}
where $f_{1},f_{2},\ldots,f_{n}\in\rep^{*}$, for each $i\in\{1,\ldots,n\}$,  $p_{i}=p_{i-1}'R_{i}$, and for each $i\in\{0,\ldots,n-1\}$,  either $p_{i}'=p_{i}$ or $p_{i}'$ is a strict prefix of $p_{i}$ such that $p_{i}'$ and~$p_{i}$ agree on their rightmost relation name. 
%If $p_{i}=p_{i}'$, the $\emptyword$-transition from $p_{i}$ to $p_{i}'$ can be ignored.
Informally, the sequence~\eqref{eq:trace} represents an accepting run of $\snfa{q}{p_{0}}$ in $\rep^{*}$.
%In particular, $p_{1}=p_{0}'R_{1}$.
Since $p_{1}\in\sset{f_{1}}{\rep^{*}}{q}\setminus\csset{f_{1}}{\db}{q}$,
we can assume a largest index $\ell\in\{1,\ldots,n\}$ such that 
$p_{\ell}\in\sset{f_{\ell}}{\rep^{*}}{q}\setminus\csset{f_{\ell}}{\db}{q}$.
By construction of $\rep^{*}$, there is a repair $\sep$ such that $f_{\ell}\in\sep$ and $\sset{f_{\ell}}{\sep}{q}=\csset{f_{\ell}}{\db}{q}$.
Consequently, $p_{\ell}\notin\sset{f_{\ell}}{\sep}{q}$.
We distinguish two cases:
\begin{description}
\item[Case that $\ell=n$.]
Thus, the run~\eqref{eq:trace} ends with
\begin{equation*}
\dotsm
\quad
p_{\ell-1}
\step{\emptyword} p_{\ell-1}'
\verylongstep{f_{\ell}=R_{\ell}(\underline{c_{\ell-1}}, c_{\ell})}p_{\ell}=q.
\end{equation*} 
Thus, the rightmost relation name in $q$ is $R_{\ell}$.
Since $f_{\ell}\in\sep$, it is clear that $p_{\ell}\in\sset{f_{\ell}}{\sep}{q}$, a contradiction.
\item[Case that $\ell<n$.]
Thus, the run~\eqref{eq:trace} includes
\begin{equation*}
\dotsm
\quad
p_{\ell-1}
\step{\emptyword} p_{\ell-1}'
\verylongstep{f_{\ell}=R_{\ell}(\underline{c_{\ell-1}}, c_{\ell})}p_{\ell} 
\step{\emptyword} p_{\ell}'
\verylongstep{f_{\ell+1}=R_{\ell+1}(\underline{c_{\ell}}, c_{\ell+1})}p_{\ell+1}
\quad
\dotsm,
\end{equation*} 
where $\ell+1$ can be equal to~$n$.
Clearly, $p_{\ell+1}\in\sset{f_{\ell+1}}{\rep^{*}}{q}$.
Assume without loss of generality that $\sep$ contains $f_{\ell+1}'\defeq R_{\ell+1}(\underline{c_{\ell}},c_{\ell+1}')$, which is key-equal to $f_{\ell+1}$ (possibly $c_{\ell+1}'=c_{\ell+1}$).
From $p_{\ell}\notin\sset{f_{\ell}}{\sep}{q}$, it follows $p_{\ell+1}\notin\sset{f_{\ell+1}'}{\sep}{q}$.
Consequently, $p_{\ell+1}\notin\csset{f_{\ell+1}'}{\db}{q}$.
By our construction of $r^{*}$, we have $p_{\ell+1}\notin\csset{f_{\ell+1}}{\db}{q}$.
Consequently, $p_{\ell+1}\in\sset{f_{\ell+1}}{\rep^{*}}{q}\setminus\csset{f_{\ell+1}}{\db}{q}$, which contradicts that~$\ell$ was chosen to be the largest such an index possible.
\end{description}
The proof of Claim~\ref{claim:smallest} is now concluded.
\renewcommand\qedsymbol{$\triangleleft$}
\end{proof}
To conclude the proof of the lemma, let $\rep$ be any repair of $\db$, and let $f\in\rep^*$ and $f'\in\rep$ be two key-equal facts in~$\db$. 
By Claim~\ref{claim:smallest} and the construction of $\rep^*$, we have that 
$\sset{f}{\rep^{*}}{q}=\csset{f}{\db}{q} \subseteq \csset{f'}{\db}{q} \subseteq \sset{f'}{\rep}{q},$
as desired.
\end{proof}

We can now give the proof of Lemma~\ref{lem:key}.

\begin{proof}[Proof of Lemma~\ref{lem:key}]
Let $\db$ be a database instance.
Then by Lemma~\ref{lem:shrinking}, there is a repair $\rep^{*}$ of $\db$ such that for every repair $\rep$ of $\db$, $\rep^{*} \mypreceq{q} \rep$.
It suffices to show that for every repair $\rep$ of $\db$,
$\starttwo{q}{\rep^{*}}\subseteq\starttwo{q}{\rep}$.
%%%
To this end, consider any repair $\rep$ and $c\in\starttwo{q}{\rep^{*}}$. 
Let $R$ be the first relation name of $q$.
Since $c\in\starttwo{q}{\rep^{*}}$, there is $d\in\adom{\rep^{*}}$ such that $R\in\sset{R(\underline{c},d)}{\rep^{*}}{q}$. 
Then, there is a unique $d'\in\adom{\rep}$ such that $R(\underline{c},d')\in\rep$, where it is possible that $d'=d$.
From $\rep^{*} \mypreceq{q} \rep$,
it follows $\sset{R(\underline{c},d)}{\rep^{*}}{q}\subseteq\sset{R(\underline{c},d')}{\rep}{q}$.
Consequently, $R\in\sset{R(\underline{c},d')}{\rep}{q}$,
which implies $c\in\starttwo{q}{\rep}$.
This conclude the proof.
\end{proof}

	\section{Complexity Upper Bounds}
\label{sec:algorithm}

We now show the complexity upper bounds of Theorem~\ref{thm:main}.

\subsection{A \PTIME\ Algorithm for $\cthree$}

We now specify a polynomial-time algorithm for $\cqa{q}$, for path queries~$q$ that satisfy condition $\cthree$.
%By Lemma~\ref{lem:pumpclosure}, $q$ satisfies $\cthree$ if and only if $q$ is a factor of every word that is accepted by $\nfa{q}$.
The algorithm is based on the automata defined in Definition~\ref{def:snfa}, and uses the concept defined next.

\begin{definition}[Relation $\proves{q}$]
Let $q$ be a path query and $\db$ a database instance.
For every $c\in\adom{q}$ and every prefix~$u$ of~$q$, 
we write $\db\proves{q}\pair{c}{u}$
if every repair of $\db$ has a path that starts in~$c$ and is accepted by $\snfa{q}{u}$.
\qed
\end{definition}

An algorithm that decides the relation~$\proves{q}$ can be used to solve $\cqa{q}$  for path queries satisfying~$\cthree$.
Indeed, by Lemma~\ref{lem:min-start}, for path queries satisfying~$\cthree$, $\db$ is a ``yes''-instance for the problem $\cqa{q}$ if and only if there is a constant $c\in\adom{\db}$ such that $\db\proves{q}\pair{c}{u}$ with $u=\emptyword$.

\begin{figure*}[t]
\begin{tabular}{rl}
\textbf{Initialization Step:} & $N \leftarrow \{\pair{c}{q}\mid c\in\adom{\db}\}.$\\
\textbf{Iterative Rule:} & 
\begin{tabular}[t]{rl}
\textbf{if}
   & $uR$ is a prefix of $q$, and\\
   & $R(\underline{c},*)$ is a nonempty block in $\db$ s.t.\ for every $R(\underline{c},y)\in\db$, $\pair{y}{uR}\in N$\\
\textbf{then}\\
 & 
$N \leftarrow N
\cup
\underbrace{\{\pair{c}{u}\}}_{\mathrm{forward}}
\cup
\underbrace{
\{\pair{c}{w}\mid\mbox{$\nfa{q}$ has a backward transition from $w$ to $u$}\}}_{\mathrm{backward}}.
$   
\end{tabular}  
\end{tabular}
\caption{Polynomial-time algorithm for computing $\{\pair{c}{u}\mid\db\proves{q}\pair{c}{u}\}$, for a fixed path query~$q$ satisfying $\cthree$.}
\label{fig:algo}
\end{figure*}

\begin{figure*}\small
\begin{tabular}{cc}
      \begin{tabular}{c | c}
		  Iteration & Tuples added to $N$ \\
		  \hline
		  init. & \texttt{<0, RRX>, <1, RRX>, <2, RRX>, <3, RRX>, <4, RRX>, <5, RRX>} \\
		  1 & \texttt{<4, RR>} \\
		  2 & \texttt{<3, R>, <3, RR>} \\
		  3 & \texttt{<2, R>, <2, RR>} \\
		  4 & \texttt{<1, R>, <1, RR>} \\
		  5 & \texttt{<0, R>, <0, RR>, <0, $\varepsilon$>} \\
	  \end{tabular}
&
\begin{minipage}{0.25\textwidth}
      \begin{tikzpicture}[->,>=stealth,auto=left, scale=1.3,vnode/.style={circle,black,inner sep=1pt,scale=1},el/.style = {inner sep=3/2pt}]
		  \node[vnode] (s0) at (0, 0) {$0$};
		  \node[vnode] (s1) at (1, 0) {$1$};
		  \node[vnode] (s2) at (2, 0) {$2$};
		  \node[vnode] (s3) at (3, 0) {$3$};

		  \node[vnode] (s4) at (2, -1) {$4$};
		  \node[vnode] (s5) at (3, -1) {$5$};
		  \path[->] (s0) edge node[el] {$R$} (s1);
		  
		  \path[->] (s1) edge node[el] {$R$} (s2);
		  \path[->] (s2) edge node[el] {$R$} (s3);
		  \path[->] (s1) edge node[el] {$R$} (s2);
		  \path[->] (s1) edge node[el] {$R$} (s4);
		  
		  \path[->] (s2) edge node[el] {$R$} (s4);
		  \path[->] (s3) edge node[el] {$R$} (s4);
		  \path[->] (s4) edge node[el,below] {$X$} (s5);
		\end{tikzpicture}
\end{minipage}
\end{tabular}	  
      \caption{Example run of our algorithm for $q=RRX$, on the database instance $\db$ shown at the right.}
      \label{fig:example-rrx-algo-table}
\end{figure*}
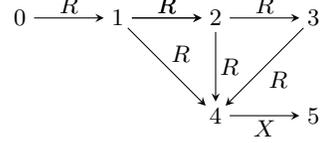

%Figure~\ref{fig:algo} shows an algorithm that computes $\{\pair{c}{u}\mid\db\proves{q}\pair{c}{u}\}$ as the fixed point of a binary relation $N$. 
%The \emph{Initialization Step} inserts into $N$ all pairs $\pair{c}{q}$, which is correct because $\db\proves{q}\pair{c}{q}$ holds vacuously, as $q$ is the accepting state of~$\snfa{q}{q}$.
%Then, the \emph{Iterative Rule} is executed until $N$ remains unchanged.
%This rule makes sure that, whenever a pair $\pair{c}{vS}$ is added, all pairs  $\pair{c}{wS}$ with $\card{w}>\card{v}$ are also added. 
%This is correct because $\db\proves{q}\pair{c}{vS}$ and $\card{w}>\card{v}$ implies $\db\proves{q}\pair{c}{wS}$, as $\snfa{q}{wS}$ has an $\emptyword$-transition from state $wS$ to $uS$.
%Figure~\ref{fig:example-rrx-algo-table} shows an example run of the algorithm in Figure~\ref{fig:algo}.
%The next lemma states the correctness of the algorithm.

Figure~\ref{fig:algo} shows an algorithm that computes
$\{\pair{c}{u}\mid\db\proves{q}\pair{c}{u}\}$ 
as the fixed point of a binary relation $N$. 
The \emph{Initialization Step} inserts into $N$ all pairs $\pair{c}{q}$, which is correct because $\db\proves{q}\pair{c}{q}$ holds vacuously, as $q$ is the accepting state of~$\snfa{q}{q}$.
Then, the \emph{Iterative Rule} is executed until $N$ remains unchanged; it intrinsically reflects the constructive proof of Lemma~\ref{lem:shrinking}:
$\db \proves{q} \pair{c}{u}$ if and only if  for every fact $f=R(\underline{c}, d)\in \db$, we have $uR \in \csset{f}{\db}{q}$.
Figure~\ref{fig:example-rrx-algo-table} shows an example run of the algorithm in Figure~\ref{fig:algo}.
The next lemma states the correctness of the algorithm.

 \begin{lemma}
 \label{lem:correctness}
 Let $q$ be a path query. 
 Let $\db$ be a database instance.
 Let $N$ be the output relation returned by the algorithm in~Figure~\ref{fig:algo} on input~$\db$. 
 Then, for every $c \in \adom{\db}$ and every prefix~$u$ of~$q$, 
\begin{center} 
$\pair{c}{u} \in N$ if and only if $\db \proves{q} \pair{c}{u}$.
\end{center}
 \end{lemma}
 
\begin{proof}
\framebox{$\impliedby$}
Proof by contraposition.
Assume $\pair{c}{u}\notin N$. 
The proof shows the construction of a repair $\rep$ of $\db$ such that $\rep$ has no path that starts in~$c$ and is accepted by $\snfa{q}{u}$.
Such a repair shows $\db\nproves{q}\pair{c}{u}$.

We explain which fact of an arbitrary block $R(\underline{a},*)$ of $\db$ will be inserted in~$\rep$. 
Among all prefixes of $q$ that end with $R$,
let $u_{0}R$ be the longest prefix such that $\pair{a}{u_{0}}\notin N$.
%Consider all states of $\nfa{q}$ of the form $uR$.
%Let $u_{0}$ be the smallest state such that $\tuple{a,u_{0}} \notin N$. 
If such $u_{0}R$ does not exist,
then an arbitrarily picked fact of the block $R(\underline{a},*)$ is inserted in $\rep$.
Otherwise, the \emph{Iterative Rule} in Figure~\ref{fig:algo} entails the existence of a fact $R(\underline{a},b)$ such that $\pair{b}{u_{0}R} \notin N$. 
Then, $R(\underline{a},b)$ is inserted in~$\rep$. 
We remark that this repair $\rep$ is constructed in exactly the same way as the repair~$\rep^*$ built in the proof of Lemma~\ref{lem:shrinking}. 

Assume for the sake of contradiction that there is a path~$\pi$ in~$\rep$ that starts in~$c$ and is accepted by~$\snfa{q}{u}$.
Let 
$\pi\defeq R_1(\underline{c_{0}},c_{1})$, $R_2(\underline{c_{1}},c_{2})$, \dots, $R_n(\underline{c_{n-1}}, c_{n})$ where $c_{0}=c$.
Since $\pair{c_{0}}{u}\not\in N$ and $\pair{c_{n}}{q}\in N$,
there is a longest prefix $u_{0}$ of $q$, where $\card{u_{0}}\geq\card{u}$, and $i\in\{1,\dots,n\}$ such that
$\pair{c_{i-1}}{u_{0}}\not\in N$ and
$\pair{c_{i}}{u_{0}R_{i}}\in N$.
From $\pair{c_{i-1}}{u_{0}}\not\in N$, it follows that $\db$ contains a fact $R_{i}(\underline{c_{i-1}},d)$ such that $\pair{d}{u_{0}R_{i}}\not\in N$.
Then $R_{i}(\underline{c_{i-1}},c_{i})$ would not be chosen in a repair, contradicting $R_i(\underline{c_{i-1}},c_{i})\in\rep$.

\framebox{$\implies$} 
Assume that $\pair{c}{u}\in N$.
Let $\ell$ be the number of executions of the \emph{Iterative Rule} that were used to insert $\pair{c}{u}$ in $N$.
We show $\db\proves{q}\pair{c}{u}$ by induction on~$\ell$.

The basis of the induction, $\ell=0$, holds because the \emph{Initialization Step} is obviously correct.
Indeed, since $q$ is an accepting state of $\snfa{q}{q}$, we have $\db\proves{q}\pair{c}{q}$.
For the inductive step, $\ell\rightarrow\ell+1$, we distinguish two cases.
\paragraph{
Case that $\pair{c}{u}$ is added to $N$ by the \emph{forward} part of the \emph{Iterative Rule}.}
That is, $\pair{c}{u}$ is added because $\db$ has a block $\{R(\underline{c},d_{1})$, \dots, $R(\underline{c},d_{k})\}$ with $k\geq 1$ and for every $i\in\{1,\dots,k\}$, 
we have that $\pair{d_i}{uR}$ was added to $N$ by a previous execution of the \emph{Iterative Rule}.
Let $\rep$ be an arbitrary repair of $\db$.
Since every repair contains exactly one fact from each block,
we can assume $i\in\{1,\dots,k\}$ such that $R(\underline{c},d_{i})\in\rep$.
By the induction hypothesis, $\db \proves{q} \pair{d_i}{uR}$ and thus $\rep$ has a path that starts in~$d_{i}$ and is accepted by $\snfa{q}{uR}$.
Clearly, this path can be left extended with $R(\underline{c},d_{i})$, and this left extended path is accepted by $\snfa{q}{u}$.
Note incidentally that the path in $\rep$ may already use $R(\underline{c},d_{i})$, in which case the path is cyclic.
Since $\rep$ is an arbitrary repair, it is correct to conclude $\db\proves{q}\pair{c}{u}$.
\paragraph{
Case that $\pair{c}{u}$ is added to $N$ by the \emph{backward} part of the \emph{Iterative Rule}.}
Then, there exists a relation name $S$ and words $v, w$ such that $u=vSwS$, and $\pair{c}{u}$ is added because $\pair{c}{vS}$ was added in the same iteration.
Then, $\snfa{q}{u}$ has an $\emptyword$-transition from state~$u$ to~$vS$.
Let $\rep$ be an arbitrary repair of $\db$.
By the reasoning in the previous case, $\rep$ has a path that starts in $c$ and is accepted by $\snfa{q}{vS}$.
We claim that $\rep$ has a path that starts in~$c$ and is accepted by $\snfa{q}{u}$.
Indeed, $\snfa{q}{u}$ can use the $\emptyword$-transition to reach the state $vS$, and then behave like~$\snfa{q}{vS}$.
This concludes the proof.
\end{proof}

The following corollary is now immediate.

\begin{corollary}\label{cor:induction}
Let $q$ be a path query.
Let $\db$ be a database instance, and $c\in\adom{\db}$.
Then, the following are equivalent:
\begin{enumerate}
\item 
$c\in\starttwo{q}{\rep}$ for every repair $\rep$ of $\db$; and
%\deleted{$c\in\bigcap\{\starttwo{q}{\rep}\mid\mbox{$\rep$ is a repair of $\db$}\}$; and}
\item
$\pair{c}{\epsilon} \in N$, where $N$ is the output of the algorithm in~Figure~\ref{fig:algo}. 
\end{enumerate}
\end{corollary}

Finally, we obtain the following tractability result.
\begin{figure*}
\centering
\begin{minipage}{\textwidth}
\begin{equation*}
\varphi_{q}(N,x,z):=
\left(
\begin{array}{ll}
& \formula{\alpha(x)\land z=\qconstant{q}}\\[1ex]
\lor & \formula{\bigvee_{\constant{uR}\leq\constant{q}}\formula{\formula{z=\qconstant{u}}\land\exists y R(\underline{x},y)\land\forall y\formula{R(\underline{x},y)\rightarrow N(y,\qconstant{uR})}}}\\[1ex]
\lor &
\formula{\bigvee_{
\substack{
\emptyword<\constant{u}<\constant{uv}\leq\constant{q}\\
\last{\constant{u}}=\last{\constant{v}}
}
}
\formula{ 
N(x,\qconstant{u})\land z=\qconstant{uv}}}
\end{array}
\right)
\end{equation*}
\end{minipage}
\caption{Definition of $\varphi_q(N, x, z)$.
The predicate $\alpha(x)$ states that $x$ is in the active domain,
and $<$ is shorthand for \emph{``is a strict prefix of''.}}
\label{fig:lfp}
\end{figure*}
\begin{lemma}\label{lem:lfp}
For each path query $q$ satisfying $\cthree$, $\cqa{q}$ is expressible in Least Fixpoint~Logic, and hence is in \PTIME.
\end{lemma}

\begin{proof}
%It remains to be argued that the algorithm in Figure~\ref{fig:algo} runs in polynomial time. 
%Indeed, for every block $R(\underline{c},*)$, 
%there are at most $\card{q}+1$ pairs of the form $\pair{c}{u}$.
%So the number of iterations is $\bigo{\card{\db}\cdot\formula{\card{q}+1}}$.
%Moreover, given $c \in \adom{\db}$, we can test in polynomial time whether $\pair{c}{\epsilon}\in N$.

For a path query~$q$,
define the following formula in LFP~\cite{DBLP:books/sp/Libkin04}:
\begin{equation}\label{eq:lfp}
\psi_{q}(s,t)\defeq\left[\mathbf{lfp}_{N,x,z}\varphi_{q}(N,x,z)\right](s,t),
\end{equation}
where $\varphi_{q}(N,x,z)$ is given in Figure~\ref{fig:lfp}.
Herein, $\alpha(x)$ denotes a first-order query that computes the active domain.
That is, for every database instance $\db$ and constant $c$,
$\db\models\alpha(c)$ if and only if $c\in\adom{\db}$.
Further, $u\leq v$ means that $u$ is a prefix of $v$;
and $u<v$ means that $u$ is a proper prefix of $v$.
Thus, $u<v$ if and only if $u\leq v$ and $u\neq v$.
The formula $\varphi_{q}(N,x,z)$  is positive in $N$, which is a $2$-ary predicate symbol.
It is understood that the middle disjunction ranges over all nonempty prefixes $uR$ of $q$ (possibly $u=\emptyword$).
The last disjunction ranges over all pairs $(u,uv)$ of distinct nonempty prefixes of $q$ that agree on their last symbol.
We used a different typesetting to distinguish the constant words $\constant{q}$, $\constant{uR}$, $\constant{uv}$ from first-order variables $x$, $z$.
It is easily verified that the LFP query~\eqref{eq:lfp} expresses the algorithm of Figure~\ref{fig:algo}. 
\end{proof}

Since the formula~\eqref{eq:lfp} in the proof of~Lemma~\ref{lem:lfp} uses universal quantification, it is not in Existential Least Fixpoint Logic, which is equal to~$\mathsf{DATALOG}_{\neg}$~\cite[Theorem~10.18]{DBLP:books/sp/Libkin04}.
% It can be shown that there exists a path query $q$ such that $\cqa{q}$ is in \PTIME\ but not expressible in stratified Datalog.

\subsection{\FO-Rewritability for $\cone$}

We now show that if a path query $q$ satisfies $\cone$, then $\cqa{q}$ is in \FO, and a first-order rewriting for $q$ can be effectively constructed.

%\jef{Note that we define a  ``first-order rewriting for~$q$'' (thus not ``for $\cqa{q}$'').}

\begin{definition}[First-order rewriting]
If $q$ is a Boolean query such that $\cqa{q}$ is in \FO,
then a \emph{(consistent) first-order rewriting} for~$q$ is a first-order sentence~$\psi$ such that for every database instance $\db$, the following are equivalent:
\begin{enumerate}
\item $\db$ is a ``yes''-instance of $\cqa{q}$; and
\item $\db$ satisfies $\psi$.
\qed
\end{enumerate} 
\end{definition}

\begin{definition}\label{def:fixedhead}
If 
$q = \{R_{1}(\underline{x_{1}},x_{2})$, $R_{2}(\underline{x_{2}},x_{3})$, \dots, $R_{k}(\underline{x_{k}},x_{k+1})\}$, $k\geq 1$,
and $c$ is a constant,
then $\fixedhead{q}{c}$ is the Boolean conjunctive query
$
\fixedhead{q}{c} \defeq \{R_{1}(\underline{c},x_{2}), R_{2}(\underline{x_{2}},x_{3}), \dots, R_{k}(\underline{x_{k}},x_{k+1}) \}
$.
\qed
\end{definition}

\begin{lemma} \label{lemma:constant-head-fo}
For every nonempty path query $q$ and constant $c$,
the problem $\cqa{\fixedhead{q}{c}}$ is in \FO.
Moreover, it is possible to construct a first-order formula $\psi(x)$, with free variable $x$, such that for every constant~$c$, the sentence $\exists x\formula{\psi(x)\land x=c}$ is a first-order rewriting for $\fixedhead{q}{c}$.
\end{lemma}
\begin{proof}
The proof inductively constructs a first-order rewriting for $\fixedhead{q}{c}$, where the induction is on the number~$n$ of atoms in~$q$.
For the basis of the induction, $n=1$, 
we have $\fixedhead{q}{c} = R(\underline{c},y)$. 
Then, the first-order formula 
$\psi(x)=\exists y R(\underline{x},y)$ 
obviously satisfies the statement of the lemma.

We next show the induction step, $n \rightarrow n+1$. 
Let $R(\underline{x_{1}},x_{2})$ be the left-most atom of~$q$,
and assume that $p\defeq q\setminus\{R(\underline{x_{1}},x_{2})\}$ is a path query with $n\geq 1$~atoms.
%Suppose that $q = R \cdot p$ where $p$ contains $n$ relations. 
By the induction hypothesis, 
it is possible to construct a first-order formula $\varphi(z)$, with free variable $z$, such that for every constant~$d$, 
\begin{equation}\label{eq:induction}
\begin{minipage}{0.8\columnwidth}
$\exists z\formula{\varphi(z)\land z=d}$ is a first-order rewriting for $\fixedhead{p}{d}$.
\end{minipage}
\end{equation}
We now  define $\psi(x)$ as follows:
\begin{equation}\label{eq:full}
\psi(x)=\exists y\formula{R(\underline{x},y)}\land\forall z\formula{R(\underline{x},z)\rightarrow\varphi(z)}.
\end{equation}
We will show that for every constant $c$, $\exists x\formula{\psi(x)\land x=c}$ is a first-order rewriting for~$\fixedhead{q}{c}$.
To this end, let $\db$ be a database instance.
It remains to be shown that $\db$ is a ``yes''-instance of $\cqa{\fixedhead{q}{c}}$ if and only if $\db$ satisfies~$\exists x\formula{\psi(x)\land x=c}$. 

\framebox{$\impliedby$}
Assume $\db$ satisfies $\exists x\formula{\psi(x)\land x=c}$.
Because of the conjunct $\exists y\formula{R(\underline{x},y)}$ in~\eqref{eq:full},
we have that $\db$ includes a block~$R(\underline{c},*)$.
Let $\rep$ be a repair of $\db$. 
We need to show that $\rep$ satisfies~$\fixedhead{q}{c}$.
Clearly, $\rep$ contains $R(\underline{c},d)$ for some constant~$d$.
Since $\db$ satisfies $\exists z\formula{\varphi(z)\land z=d}$,
the induction hypothesis~\eqref{eq:induction} tells us that $\rep$ satisfies $\fixedhead{p}{d}$.
It is then obvious that $\rep$ satisfies $\fixedhead{q}{c}$.

\framebox{$\implies$}
Assume $\db$ is a ``yes''-instance for $\cqa{\fixedhead{q}{c}}$. 
Then $\db$ must obviously satisfy $\exists y\formula{R(\underline{c},y)}$.
Therefore, $\db$ includes a block~$R(\underline{c},*)$.
Let $\rep$ be an arbitrary repair of $\db$.
There exists $d$ such that $R(\underline{c},d)\in\rep$.
Since $\rep$ satisfies $\fixedhead{q}{c}$, it follows that $\rep$ satisfies $\fixedhead{p}{d}$.
Since $\rep$ is an arbitrary repair,
the induction hypothesis~\eqref{eq:induction} tells us that $\db$ satisfies $\exists z\formula{\varphi(z)\land z=d}$.
It is then clear that $\db$ satisfies $\exists x\formula{\psi(x)\land x=c}$.
%The proof implies a recursive construction of a first-order rewriting.
\end{proof}

\begin{lemma} \label{lem:fo-algo}
For every path query $q$ that satisfies $\cone$, the problem $\cqa{q}$ is in \FO, and a first-order rewriting for $q$ can be effectively constructed.
\end{lemma}
\begin{proof}
By Lemmas~\ref{lem:pumpclosure} and~\ref{lem:min-start}, a database instance $\db$ is a ``yes''-instance for $\cqa{q}$ if and only if there is a constant $c$ (which depends on $\db$) such that  $\db$ is a ``yes''-instance for $\cqa{\fixedhead{q}{c}}$. 
By Lemma~\ref{lemma:constant-head-fo}, it is possible to construct a first-order rewriting $\exists x\formula{\psi(x)\land x=c}$ for $\fixedhead{q}{c}$.
It is then clear that  $\exists x\formula{\psi(x)}$ is a first-order rewriting for~$q$.
\end{proof}

\subsection{An \NL\ Algorithm for $\ctwo$}\label{sec:datalog}
% \xiating{Should we move Section~\ref{sec:syntax} here? Since it is only used to the \NL~proof. \jef{I don't think so, because Section~\ref{sec:syntax} also gives interesting insights in $\cone$ and $\cthree$.}}

We show that $\cqa{q}$ is in \NL\ if $q$ satisfies $\ctwo$ by expressing it in linear Datalog with stratified negation. 
The proof will use the syntactic characterization of $\ctwo$ established in Lemma~\ref{lem:notctwo}.

\begin{lemma} \label{lem:nl-algo}
For every path query $q$ that satisfies $\ctwo$, the problem $\cqa{q}$ is in linear Datalog with stratified negation (and hence in \NL).
\end{lemma}

In the remainder of this section, we develop the proof of Lemma~\ref{lem:nl-algo}.

\begin{definition}
Let $q$ be a path query. 
We define $\nfashortest{q}$ as the automaton that accepts $w$ if $w$ is accepted by $\nfa{q}$ and no proper prefix of~$w$ is accepted by $\nfa{q}$.
\qed
\end{definition}

It is well-known that such an automaton $\nfashortest{q}$ exists.

\begin{example}
Let $q=RXRYR$. 
Then,  $RXRYRYR$ is accepted by $\nfa{q}$, but not by $\nfashortest{q}$, because the proper prefix $RXRYR$ is also accepted by $\nfa{q}$.
\qed	
\end{example}

\begin{definition}
Let $q$ be a path query and $\rep$ be a consistent database instance.
We define $\startshortest{q}{\rep}$ as the set containing all (and only) constants $c\in\adom{\rep}$ such that there is a path in $\rep$ that starts in~$c$ and is accepted by $\nfashortest{q}$.
\qed
\end{definition}

\begin{lemma}\label{lem:minnfa}
Let $q$ be a path query.
For every consistent database instance $\rep$, we have that $\starttwo{q}{\rep}=\startshortest{q}{\rep}$.
\end{lemma}
\begin{proof}
By construction, $\startshortest{q}{\rep}\subseteq\starttwo{q}{\rep}$.
%\jef{I use Greek letter $\pi$ for a sequence of database facts; I prefer to reserve lowercase letters $u,v,w$ for words (over the alphabet of relation names).}
Next assume that $c\in\starttwo{q}{\rep}$ and let $\pi$ be the path that starts in~$c$ and is accepted by $\nfa{q}$. 
Let $\pi^{-}$ be the shortest prefix of $\pi$ that is accepted by $\nfa{q}$. 
Since $\pi^{-}$ starts in~$c$ and is accepted by $\nfashortest{q}$, it follows $c\in\startshortest{q}{\rep}$.
\end{proof}

\begin{lemma}\label{lem:stutterb2}
Let $u\cdot v\cdot w$ be a self-join-free word over the alphabet of relation names.
Let $s$ be a suffix of $uv$ that is distinct from $uv$.
For every integer $k\geq 0$, $\nfashortest{s\kleene{uv}{k}wv}$ 
accepts the language of the regular expression $s\kleene{uv}{k}\kleene{uv}{*}wv$.
\end{lemma}
\begin{proof}
Let $q=s\kleene{uv}{k}wv$. Since $u\cdot v\cdot w$ is self-join-free,
applying the rewinding operation, zero, one, or more times, in the part of $q$ that precedes~$w$ will repeat the factor $uv$. 
This gives words of the form
$s\kleene{uv}{\ell}wv$ with $\ell\geq k$.
The difficult case is where we rewind a factor of~$q$ that itself contains~$w$ as a factor. 
In this case, the rewinding operation will repeat a factor of the form
$v_{2}\kleene{uv}{\ell}wv_{1}$ such that $v=v_{1}v_{2}$ and $v_{2}\neq\emptyword$,
which results in words of one of the following forms ($s=s_{1}\cdot v_{2}$):
$$
\begin{array}{l}
\swipe
{\formula{s\kleene{uv}{\ell_{1}}uv_{1}}}
{v_{2}\kleene{uv}{\ell_{2}}wv_{1}}
{\formula{v_{2}}}; \mbox{\ or}\\
\swipe
{\formula{s_{1}}}
{v_{2}\kleene{uv}{\ell}wv_{1}}
{\formula{v_{2}}}.
\end{array}
$$
These words have a prefix belonging to the language of the regular expression $s\kleene{uv}{k}\kleene{uv}{*}wv$.
\end{proof}

%\jef{Note sure whether the distinction between $\step{q}$ and $\cqastep{q}$ is in the next definition is useful later on.}
%\xiating{I like the distinction since it shows that ``WLOG'', we can just use $\step{q}$ to obtain $\cqastep{q}$ by shortcutting the paths. }

\begin{definition}\label{def:terminal}
Let $\db$ be a database instance, and $q$ a path query.

For $a,b\in\adom{\db}$, we write $\db\models a\step{q}b$ if there exists a path in~$\db$ from~$a$ to~$b$ with trace~$q$.
Even more formally, $\db\models a\step{q}b$  if $\db$ contains facts $R_{1}(\underline{a_{1}},a_{2}), R_{2}(\underline{a_{2}},a_{3}), \ldots,  R_{\card{q}}(\underline{a_{\card{q}}},a_{\card{q}+1})$ such that $R_{1}R_{2}\dotsm R_{\card{q}}=q$.
We write $\db\models a\step{q_{1}}b\step{q_{2}}c$ as a shorthand for $\db\models a\step{q_{1}}b$ and $\db\models b\step{q_{2}}c$.

We write $\db\models a\cqastep{q}b$ if there exists a \emph{consistent path} in~$\db$ from~$a$ to~$b$ with trace~$q$, where a path is called consistent if it does not contain two distinct key-equal facts.

A constant $c\in\adom{\db}$ is called \emph{terminal for~$q$ in~$\db$} if for some (possibly empty) proper prefix~$p$ of~$q$, there is a consistent path in $\db$ with trace~$p$ that cannot be right extended to a consistent path in $\db$ with trace~$q$. 
\qed
\end{definition}

Note that for every $c\in\adom{\db}$, we have $c\cqastep{\emptyword}c$.
Clearly, if $q$ is self-join-free, then $c\step{q}d$ implies $c\cqastep{q}d$ (the converse implication holds vacuously true).

\begin{example}
Let $\db=\{R(\underline{c},d), S(\underline{d},c), R(\underline{c},e), T(\underline{e},f)\}$.
Then, $c$ is terminal for $RSRT$ in~$\db$ because the path $R(\underline{c},d), S(\underline{d},c)$ cannot be right extended to a consistent path with trace~$RSRT$, because $d$ has no outgoing $T$-edge. 
Note incidentally that $\db\models c\cqastep{RS}c\cqastep{RT}f$, but  $\db\not\models c\cqastep{RSRT}f$.
\qed
\end{example}

\begin{lemma}
\label{prop:check-terminal}
Let $\db$ be a database instance, and $c\in\adom{\db}$.
Let $q$ be a path query.
Then, $c$ is terminal for~$q$ in~$\db$ if and only if $\db$ is a ``no''-instance of $\cqa{\fixedhead{q}{c}}$, with $\fixedhead{q}{c}$ as defined by Definition~\ref{def:fixedhead}.
\end{lemma}
\begin{proof}
\framebox{$\implies$}
Straightforward.
\framebox{$\impliedby$}
Assume $\db$ is a ``no''-instance of $\cqa{\fixedhead{q}{c}}$.
Then, there is a repair~$\rep$ of $\db$ such that $\rep\not\models\fixedhead{q}{c}$.
The empty path is a path in~$\rep$ that starts in~$c$ and has trace $\emptyword$, which is a prefix of~$q$. 
We can therefore assume a longest prefix~$p$ of~$q$ such there exists a path~$\pi$ in~$\rep$ that starts in~$c$ and has trace~$p$.
Since $\rep$ is consistent, $\pi$ is consistent.
From $\rep\not\models\fixedhead{q}{c}$, it follows that $p$ is a proper prefix of~$q$.
By Definition~\ref{def:terminal}, $c$ is terminal for~$q$ in~$\db$.
\end{proof}

%\jef{
%\begin{proposition}
%Let $\db$ be a database instance, and $c\in\adom{\db}$.
%Let $p,q$ be path queries such that~$p$ is a prefix of~$q$.
%If $c$ is terminal for~$p$ in~$\db$,
%then $c$ is terminal for~$q$ in~$\db$.
%\end{proposition}
%\begin{proof}
%Straightforward.
%\end{proof}
%}

We can now give the proof of Lemma~\ref{lem:nl-algo}.

%\begin{lemma}\label{lem:nlproof}
%If $q$ satisfies $\ctwo$, then $\cqa{q}$ is in \NL.
%\end{lemma}
\begin{proof}[Proof of Lemma~\ref{lem:nl-algo}]
Assume $q$ satisfies $\ctwo$.
By Lemma~\ref{lem:notctwo}, $q$ satisfies $\btwoa$ or $\btwob$.
We treat the case that $q$ satisfies~$\btwob$ (the case that $q$ satisfies $\btwoa$ is even easier).
We have that $q$ is a factor of $\kleene{uv}{k}wv$, where $k$ is chosen as small as possible, and $uvw$ is self-join-free.
The proof is straightforward if $k=0$; we assume $k\geq 1$ from here on. 
To simplify notation, we will show the case where $q$ is a suffix of $\kleene{uv}{k}wv$; our proof can be easily extended to the case where~$q$ is not a suffix, at the price of some extra notation.
There is a suffix $s$ of $uv$ such that $q=s\kleene{uv}{k-1}wv$. 

We first define a unary predicate~$P$ (which depends on~$q$) such that $\db\models P(d)$ if for some $\ell\geq 0$, there are constants $d_{0},d_{1},\ldots,d_{\ell}\in\adom{\db}$ with $d_{0}=d$ such that:
\begin{enumerate}[label={(\roman*)}] 
\item\label{it:dodl}
$\db\models d_{0}\step{uv}d_{1}\step{uv}d_{2}\step{uv}\dotsm\step{uv}d_{\ell}$;
\item\label{it:dod2}
for every $i\in\{0,1,\ldots,\ell\}$, $d_{i}$ is terminal for~$wv$ in~$\db$; and
\item\label{it:dod3}
either $d_{\ell}$ is terminal for~$uv$ in~$\db$, or $d_{\ell}\in\{d_{0},\ldots,d_{\ell-1}\}$.
\end{enumerate}

\begin{claim}\label{cla:consistent}
The definition of the predicate~$P$ does not change if we replace item~\ref{it:dodl} by the stronger requirement that for every $i\in\{0,1,\ldots,\ell-1\}$, there exists a path $\pi_{i}$ from $d_{i}$ to $d_{i+1}$ with trace~$uv$ such that the composed path $\pi_{0}\cdot\pi_{1}\cdots\pi_{\ell-1}$ is consistent. 
\end{claim}

\begin{proof}
% [\xiating{Attempted inductive proof.}] \jef{I am fine with your proof. Now it seems that the proof is essentially the same as a proof of the following: if there is a path in a directed graph from $s$ to $t$, then there is a path in which every vertex occurs at most once. So I think the proof is really ``straightforward'' (but it is fine having it).}
% \xiating{Yes indeed - the essential arguments are really simple. It is perhaps important to mention that the ``consistent'' path may visit different constants than the ``original path''. \jef{Why is that? The inductive step in your proof replaces $\pi_{1}\cdot R(\underline{a},b_1)\cdot \pi_{2}\cdot R(\underline{a},b_2)\cdot \pi_{3}$ with $\pi_{1}\cdot R(\underline{a},b_2)\cdot \pi_{3}$; all edges on the new path were already on the old path. Am I missing something?}}
% \xiating{You are right, and my previous comment was imprecise about the word ``different''. I meant the following: we can always construct a ``consistent path'' that visits a \emph{subset} of constants in the ``original path''.}
It suffices to show the following statement by induction on increasing~$l$:
\begin{quote}
whenever there exist $l \geq 1$ and constants $d_0, d_1, \dots, d_l$ with $d_0 = d$ such that conditions~\ref{it:dodl}, \ref{it:dod2}, and~\ref{it:dod3} hold, there exist another constant $k \geq 1$ and constants $c_0, c_1, \dots, c_k$ with $c_0 = d$ such that conditions~\ref{it:dodl}, \ref{it:dod2}, and~\ref{it:dod3} hold, and, moreover, for each $i \in \{0, 1, \dots, k-1\}$, there exists a path $\pi_i$ from $c_i$ to $c_{i+1}$ such that the composed path $\pi_0 \cdot \pi_1 \cdots \pi_{k-1}$ is consistent.
\end{quote}

\begin{description}
\item[Basis $l = 1$.] Then we have $\db \models d_0 \step{uv} d_1$, witnessed by a path $\pi_0$. Since $uv$ is self-join-free, the path $\pi_0$ is consistent. The claim thus follows with $k = l = 1$, $c_0 = d_0$ and $c_1 = d_1$.
\item[Inductive step $l \rightarrow l+1$.] 
Assume that the statement holds for any integer in $\{1, 2, \ldots, l\}$. Suppose that there exist $l\geq 2$ and constants $d_0, d_1, \dots, d_{l+1}$ with $d_0 = d$ such that  conditions~\ref{it:dodl}, \ref{it:dod2}, and~\ref{it:dod3} hold.

For $i\in\{0,\ldots,l\}$, let $\pi_i$ be a path with trace~$uv$ from $d_i$ to $d_{i+1}$ in $\db$. The claim holds if the composed path $\pi_0 \cdot \pi_1 \cdots \pi_{l}$ is consistent, with $k = l+1$ and $c_i = d_i$ for $i \in \{0,1,\dots,l+1\}$.

Now, assume that for some $i<j$, the paths that show $\db\models d_{i}\step{uv}d_{i+1}$ and $\db\models d_{j}\step{uv}d_{j+1}$ contain, respectively, $R(\underline{a},b_{1})$ and $R(\underline{a},b_{2})$ with $b_{1}\neq b_{2}$.
It is easily verified that 
$$\db\models d_{0}\step{uv}d_{1}\step{uv}d_{2}\step{uv}\dotsm\step{uv}d_{i}\step{uv}d_{j+1} \step{uv} \dotsm \step{uv} d_{l+1},$$
where the number of $uv$-steps is strictly less than $l+1$.
Informally, we follow the original path until we reach $R(\underline{a},b_{1})$, but then follow $R(\underline{a},b_{2})$ instead of $R(\underline{a},b_{1})$, and continue on the path that proves $\db\models d_{j}\step{uv}d_{j+1}$.
Then the claim holds by applying the inductive hypothesis on constants $d_0, d_1, \ldots, d_i, d_{j+1}, \ldots, d_{l+1}$.
\end{description}
The proof is now complete.
\end{proof}

Since we care about the expressibility of the predicate~$P$ in Datalog, Claim~\ref{cla:consistent} is not cooked into the definition of~$P$. The idea is the same as in an \NL-algorithm for reachability: if there exists a directed path from~$s$ to~$t$, then there is such a path without repeated vertices; but we do not care for repeated vertices when computing reachability.

\begin{claim}\label{cla:terminal}
The definition of predicate~$P$ does not change if we require that for $i\in\{0,1,\ldots,\ell-1\}$, $d_{i}$ is not terminal for~$uv$ in~$\db$. 
\end{claim}
\begin{proof}
Assume that for some $0\leq i<\ell$, $d_{i}$ is terminal for~$uv$ in~$\db$.
Then, all conditions in the definition are satisfied by choosing $\ell$ equal to~$j$.
\end{proof}

Claim~\ref{cla:terminal} is not cooked into the definition of~$P$ to simplify the the encoding of~$P$ in Datalog.

Next, we define a unary predicate $O$ such that $\db\models O(c)$ for a constant~$c$ if $c\in\adom{\db}$ and one of the following holds true:
\begin{enumerate}
\item
$c$ is terminal for~$s\kleene{uv}{k-1}$ in~$\db$; or
\item
there is a constant $d\in\adom{\db}$ such that both $\db\models c\longcqastep{s\kleene{uv}{k-1}}d$ and $\db\models P(d)$.
%$\db$ contains consistent path with trace $s\kleene{uv}{k-1}$ from~$c$ to some constant in $\{d\in\adom{\db}\mid\db\models P(d)\}$.
\end{enumerate}

\begin{claim}\label{cla:oc}
Let $c\in\adom{\db}$.
The following are equivalent:
\begin{enumerate}[label={(\Roman*)}] 
\item\label{it:to}
there is a repair $\rep$ of $\db$ that contains no path that starts in~$c$ and whose trace is in the language of the regular expression $s\kleene{uv}{k-1}\kleene{uv}{*}wv$; and
\item\label{it:ot}
$\db\models O(c)$.
\end{enumerate}
\end{claim}
\begin{proof}
Let $wv=S_{0}S_{1}\dotsm S_{m-1}$ and $uv=R_{0}R_{1}\dotsm R_{n-1}$.

\framebox{\ref{it:to}$\implies$\ref{it:ot}}
Assume that item~\ref{it:to} holds true.
Let the first relation name of~$s$ be $R_{i}$.
Starting from~$c$, let $\pi$ be a maximal (possibly infinite) path in~$\rep$ that starts in~$c$ and has trace $R_{i}R_{i+1}R_{i+2}\dotsm$, where addition is modulo~$n$.  
Since $\rep$ is consistent, $\pi$ is deterministic.
Since $\rep$ is finite, $\pi$ contains only finitely many distinct edges.
Therefore, $\pi$ ends either in a loop or in an edge $R_{j}(\underline{d},e)$ such that $\db\models\neg\exists y R_{j+1}(\underline{e},y)$ (recall that $\rep$ contains a fact from every block of $\db$).
Assume that $\pi$ has a prefix $\pi'$ with trace $s\kleene{uv}{k-1}$;
if~$e$ occurs at the non-primary key position of the last $R_{n-1}$-fact of $\pi'$ or of any $R_{n-1}$-fact occurring afterwards in~$\pi$, then it follows from item~\ref{it:to} that there exist a (possibly empty) prefix~$pS_{j}$ of~$wv$ and a constant~$f\in\adom{\rep}$ such that $\rep\models e\step{p}f$ and $\db\models\neg\exists y S_{j}(\underline{f},y)$.   
It is now easily verified that $\db\models O(c)$.

\framebox{\ref{it:ot}$\implies$\ref{it:to}}
Assume $\db\models O(c)$.
It is easily verified that the desired result holds true if $c$ is terminal for~$s\kleene{uv}{k-1}$ in~$\db$.
Assume from here on that $c$ is not terminal for~$s\kleene{uv}{k-1}$ in~$\db$.
That is, for every repair~$\rep$ of~$\db$, there is a constant~$d$ such that $\rep\models c\longstep{s\kleene{uv}{k-1}}d$.
Then, there is a consistent path~$\alpha$ with trace~$s\kleene{uv}{k-1}$ from~$c$ to some constant $d\in\adom{\db}$ such that $\db\models P(d)$, using the stronger definition of~$P$ implied by Claims~\ref{cla:consistent} and~\ref{cla:terminal}.
Let $d_{0},\ldots,d_{\ell}$ be as in our (stronger) definition of~$P(d)$, that is, first, $d_{1},\ldots,d_{\ell-1}$ are not terminal for~$uv$ in~$\db$ (cf.~Claim~\ref{cla:terminal}), and second, there is a $\subseteq$-minimal consistent subset $\pi$ of $\db$ such that $\pi\models d_{0}\step{uv}d_{1}\step{uv}d_{2}\step{uv}\dotsm\step{uv}d_{\ell}$ (cf.~Claim~\ref{cla:consistent}).
We construct a repair~$\rep$ as follows:
\begin{enumerate}
\item\label{it:pipi}
insert into $\rep$ all facts of~$\pi$; %\xiating{what is $\pi$ in this direction?\jef{Added its definition in the preceding paragraph.}}
\item
for every $i\in\{0,\ldots,\ell\}$, $d_{i}$ is terminal for $wv$ in~$\db$.
We ensure that $\rep\models d_{i}\longstep{S_{0}S_{1}\dotsm S_{j_{i}}}e_{i}$ for some $j_{i}\in\{0,\ldots,m-2\}$ and some constant~$e_{i}$ such that $\db\models\neg\exists y S_{j_{i}+1}(\underline{e_{i}},y)$;
\item
if $d_{\ell}$ is terminal for $uv$ in~$\db$, then 
we ensure that $\rep\models d_{\ell}\longstep{R_{0}R_{1}\dotsm R_{j}}e$ for some $j\in\{0,\ldots,n-2\}$ and some constant~$e$ such that $\db\models\neg\exists y S_{j+1}(\underline{e},y)$; 
\item
insert into $\rep$ the facts of~$\alpha$ that are not key-equal to a fact already in~$\rep$; and
\item
complete $\rep$ into a $\subseteq$-maximal consistent subset of~$\db$.
\end{enumerate}
%\jef{We may argue here in more depth that the insertion of ``dangling'' paths in the second and third steps is possible.}
%By our hypothesis that $c$ is not terminal for~$s\kleene{uv}{k-1}$ in~$\db$,
Since $\rep$ is a repair of~$\db$,
there exists a path~$\delta$ with trace $s\kleene{uv}{k-1}$ in~$\rep$ that starts from~$c$.
If $\delta\neq\alpha$, then $\delta$ must contain a fact of~$\pi$ that was inserted in step~\ref{it:pipi}. 
Consequently, no matter whether $\delta=\alpha$ or $\delta\neq\alpha$, the endpoint of~$\delta$ belongs to $\{d_{0},\ldots,d_{\ell}\}$. 
It follows that there is a (possibly empty) path from $\delta$'s endpoint to~$d_{\ell}$ whose trace is of the  form $\kleene{uv}{*}$.
Two cases can occur:
\begin{itemize}
\item
$d_{\ell}$ is terminal for $uv$ in~$\db$.
\item
$d_{\ell}$ is not terminal for $uv$ in~$\db$.
Then there is $j\in\{0,\ldots,\ell-1\}$ such that $d_{j}=d_{\ell}$.
Then, there is a path of the form~$\kleene{uv}{*}$ that starts from $\delta$'s endpoint and eventually loops.
\end{itemize}
Since, by construction, each $d_{i}$ is terminal for~$wv$ in $\rep$, it will be the case that $\delta$ cannot be extended to a path in~$\rep$ whose trace is of the form $s\kleene{uv}{k}\kleene{uv}{*}wv$.
\end{proof}

\begin{claim}\label{cla:datalog}
The unary predicate $O$ is expressible in linear Datalog with stratified negation.
\end{claim}
\begin{proof}
% Easy.
The construction of the linear Datalog program is straightforward.
Concerning the computation of predicates~$P$ and~$O$, note that it can be checked in \FO\ whether or not a constant $c$ is terminal for some path query $q$, by Lemmas~\ref{lemma:constant-head-fo} and~\ref{prop:check-terminal}.
The only need for recursion comes from condition~\ref{it:dodl} in the definition of the predicate~$P$, which searches for a directed path of a particular form.
We give a program for $q=UVUVWV$, where $\mathtt{c(X)}$ states that~$\mathtt{X}$ is a constant, and $\mathtt{ukey(X)}$ states that $\mathtt{X}$ is the primary key of some $U$-fact. 
$\mathtt{consistent(X1,X2,X3,X4)}$ is true if either $\mathtt{X1\neq X3}$ or $\mathtt{X2=X4}$ (or both).
%\xiating{$\mathtt{X2=X4}$?}
\begin{quote}\footnotesize
\begin{verbatim}
uvterminal(X) :- c(X), not ukey(X).
uvterminal(X) :- u(X,Y), not vkey(Y).
wvterminal(X) :- c(X), not wkey(X).
wvterminal(X) :- w(X,Y), not vkey(Y).

uv2terminal(X) :- uvterminal(X).
uv2terminal(X1) :- u(X1,X2), v(X2,X3), uvterminal(X3).

uvpath(X1,X3) :- u(X1,X2), v(X2,X3), wvterminal(X1), wvterminal(X2), wvterminal(X3).
uvpath(X1,X4) :- uvpath(X1,X2), u(X2,X3), v(X3,X4), wvterminal(X3), wvterminal(X4).

p(X) :- uvterminal(X), wvterminal(X). %%% the empty path.
p(X) :- uvpath(X,Y), uvterminal(Y).
p(X) :- uvpath(X,Y), uvpath(Y,Y). %%% p and uvpath are not mutually recursive.

o(X) :- uv2terminal(X).
o(X1) :- u(X1,X2), v(X2,X3), u(X3,X4), v(X4,X5), consistent(X1,X2,X3,X4), consistent(X2,X3,X4,X5), p(X5).
\end{verbatim}
\end{quote}
The above program is in linear Datalog with stratified negation.
It is easily seen that any path query satisfying~$\btwob$ admits such a program for the predicate~$O$.
%\jef{It can be shown that the program remains correct if we omit the $\mathtt{consistent}$-predicates; however, such a shorter program may not be more efficient as it may consider more possibilities.}
\end{proof}
By Lemmas~\ref{lem:min-start}, \ref{lem:minnfa}, and~\ref{lem:stutterb2}, the following are equivalent:
\begin{enumerate}[label=(\alph*)]
\item
$\db$ is a ``no''-instance of $\cqa{q}$; and
\item\label{it:nltrace}
for every constant $c_{i}\in\adom{q}$,
there is a repair $\rep$ of $\db$ that contains no path that starts in~$c_{i}$ and whose trace is in the language of the regular expression $s\kleene{uv}{k-1}\kleene{uv}{*}wv$.
%\jef{
%We can equivalently write: for every constant $c_{i}\in\adom{q}$,
%there is a repair $\rep$ of $\db$ such that for every constant $d\in\adom{\rep}$, for every $k'\geq k$, we have $\rep\not\models c_{i}\longcqastep{s\kleene{uv}{k'}wv}d$.
%}
\end{enumerate}
By Claim~\ref{cla:oc}, item~\ref{it:nltrace} holds true if and only if for every $c\in\adom{\db}$, $\db\models\neg O(c)$.
It follows from Claim~\ref{cla:datalog} that the latter test is in linear Datalog with stratified negation, which concludes the proof of Lemma~\ref{lem:nl-algo}.
\end{proof}

	\section{Complexity Lower Bounds}
\label{sec:hardness}

In this section, we show the complexity lower bounds of Theorem~\ref{thm:main}.
For a path query $q=\{R_1(x_1, x_2)$, \dots, $R_k(x_k, x_{k+1})\}$ and constants $a,b$, we define the following database instances:
\begin{eqnarray*}
\phi_a^b[q] & \defeq & \{R_1(a, \Box_2), R_2(\Box_2, \Box_3), \dots, R_k(\Box_k,b) \}\\
\phi_a^\bot[q] & \defeq & \{R_1(a, \Box_2), R_2(\Box_2, \Box_3), \dots, R_k(\Box_k,\Box_{k+1}) \} \\
\phi_\bot^b[q] & \defeq & \{R_1(\Box_1, \Box_2), R_2(\Box_2, \Box_3), \dots, R_k(\Box_k,b) \} 
\end{eqnarray*}
where the symbols $\Box_i$ denoted fresh constants not occurring elsewhere.
Significantly, two occurrences of $\Box_{i}$ will represent different constants.

\subsection{\NL-Hardness}\label{sec:nlhard}

We first show that if a path query violates $\cone$, then $\cqa{q}$ is \NL-hard, and therefore not in \FO.

\begin{lemma}
\label{lemma:nl-hard}
If a path query $q$ violates $\cone$, then $\cqa{q}$ is \NL-hard.
\end{lemma}
\begin{proof}
Assume that $q$ does not satisfy $\cone$. 
Then, there exists a relation name~$R$ such that $q=uRvRw$ and $q$ is not a prefix of $uRvRvRw$.
It follows that $Rw$ is not a prefix of $RvRw$.
Since $Rv\neq\emptyword$, there exists no (conjunctive query) homomorphism from $q$ to $uRw$.

The problem $\problem{REACHABILITY}$ takes as input a directed graph $G(V,E)$ and two vertices $s,t\in V$, and asks whether $G$ has a directed path from~$s$ to~$t$. This problem is \NL-complete and remains \NL-complete when the inputs are acyclic graphs.
Recall that \NL\ is closed under complement.
We present a first-order reduction from $\problem{REACHABILITY}$ to the complement of $\cqa{q}$, for acyclic directed graphs.

Let $G =(V,E)$ be an acyclic directed graph and $s,t\in V$. 
Let $G'=(V \cup \{s',t'\}, E \cup \{(s',s),(t,t')\})$, where $s',t'$ are fresh vertices. 
We construct an input instance $\db$ for $\cqa{q}$ as follows:
\begin{itemize}
\item for each vertex $x \in V \cup \{s'\}$, we add $\phi_\bot^x[u]$; 
\item for each edge $(x, y) \in E \cup \{(s',s),(t,t')\}$, we add $\phi_x^y[Rv]$; and
\item for each vertex $x \in V$, we add $\phi^\bot_x[Rw]$.
\end{itemize}
This construction can be  executed in \FO.
Figure~\ref{fig:NL} shows an example of the above construction.
Observe that the only conflicts in $\db$ occur in $R$-facts outgoing from a same vertex.

\begin{figure}[h]\centering
    \begin{tikzpicture}[->,>=stealth,auto=left, scale=1.8,vnode/.style={circle,black,inner sep=1pt,scale=1},el/.style = {inner sep=3/2pt}]
      \node[vnode] (s') at (0,0) {$s'$};
      \node[vnode] (s) at (1,0) {$s$};
      \node[vnode] (a) at (2,0) {$a$};
      \node[vnode] (t) at (3,0) {$t$};
      \node[vnode] (t') at (4,0) {$t'$};
      
      \node[vnode] (ps') at (0,2/3) {};
      \node[vnode] (ps) at (1,2/3) {};
      \node[vnode] (pa) at (2,2/3) {};
      \node[vnode] (pt) at (3,2/3) {};
      
       \node[vnode] (ts) at (1,-2/3) {};
      \node[vnode] (ta) at (2,-2/3) {};
      \node[vnode] (tt) at (3,-2/3) {};
      
       \path[->] (ps') edge node[el] {$u$} (s'); 
       \path[->] (ps) edge node[el] {$u$} (s); 
       \path[->] (pa) edge node[el] {$u$} (a); 
       \path[->] (pt) edge node[el] {$u$} (t); 
       
       \path[->] (s) edge node[el] {$Rw$} (ts); 
       \path[->] (a) edge node[el] {$Rw$} (ta); 
       \path[->] (t) edge node[el] {$Rw$} (tt);

       \path[->] (s') edge node[el] {$Rv$} (s); 
       \path[->] (s) edge node[el] {$Rv$} (a); 
       \path[->] (a) edge node[el] {$Rv$} (t); 
       \path[->] (t) edge node[el] {$Rv$} (t'); 
    \end{tikzpicture}
    \caption{Database instance for the \NL-hardness reduction from the graph $G$ with  $V=\{s,a,t\}$ and $E=\{(s,a),(a,t)\}$.}
    \label{fig:NL}
\end{figure}
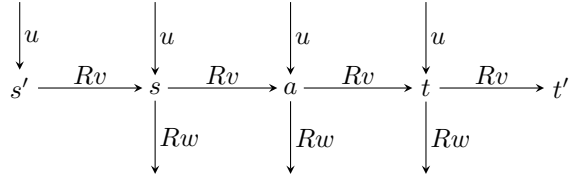

We now show that there exists a directed path from $s$ to $t$ in $G$ if and only if there exists a repair of $\db$ that does not satisfy $q$. 

\framebox{$\implies$}
Suppose that there is a directed path from $s$ to $t$ in $G$. 
Then, $G'$ has a directed path $P = s, x_0, x_1, \dots,t, t'$. 
Then, consider the repair $\rep$ that chooses the first $R$-fact from 
$\phi_x^y[Rv]$ for each edge $(x, y)$ on the path $P$, and the first $R$-fact from $\phi_y^\bot[Rw]$ for each $y$ not on the path $P$. 
We show that $\rep$ falsifies~$q$. 
Assume for the sake of contradiction that $\rep$ satisfies~$q$.
Then, there exists a valuation~$\theta$ for the variables in~$q$ such that $\theta(q)\subseteq\rep$.
Since, as argued in the beginning of this proof, there exists no (conjunctive query) homomorphism from $q$ to $uRw$,
it must be that all facts in $\theta(q)$ belong to a path in~$\rep$ with trace $u\kleene{Rv}{k}$, for some $k\geq 0$.
Since, by construction, no constants are repeated on such paths,
there exists a (conjunctive query) homomorphism from $q$ to $u\kleene{Rv}{k}$, which implies that $Rw$ is a prefix of $RvRw$, a contradiction.
We conclude by contradiction that $\rep$ falsifies~$q$.

\framebox{$\impliedby$}
Proof by contradiction.
Suppose that there is no directed path from $s$ to~$t$ in~$G$. Let $\rep$ be any repair of $\db$; we will show that $\rep$ satisfies $q$. Indeed, there exists a maximal path $P = x_0, x_1, \dots, x_n$ such that $x_0 = s'$, $x_1 = s$, and $\phi_{x_i}^{x_{i+1}}[Rv] \subseteq \rep$. By construction, $s'$ cannot reach $t'$ in $G'$, and thus $x_n \neq t'$. Since $P$ is maximal, we must have $\phi_{x_n}^\bot[Rw] \subseteq \rep$. 
Then $\phi_\bot^{x_{n-1}}[u] \cup \phi_{x_{n-1}}^{x_{n}}[Rv] \cup  \phi_{x_n}^\bot[Rw] $ satisfies~$q$.
\end{proof}

\subsection{\coNP-Hardness}~\label{sec:conphard}
Next, we show the \coNP-hard lower bound.

\begin{lemma}
\label{lemma:conp-hard}
	If a path query $q$ violates $\cthree$, then $\cqa{q}$ is \coNP-hard.
\end{lemma}

\begin{proof}
	If $q$ does not satisfy $\cthree$, then there exists a relation $R$ such that $q = u Rv  Rw$ and $q$ is not a factor of $u  Rv  Rv Rw$. Note that this means that there is no homomorphism from $q$ to $u  Rv  Rv Rw$. Also, $u$ must be nonempty (otherwise, $q = RvRw$ is trivially a suffix of $RvRvRw$). Let $S$ be the first relation of $u$.

	The proof is a first-order reduction from $\problem{SAT}$ to the complement of $\cqa{q}$. 
	The problem $\problem{SAT}$ asks whether a given propositional formula in CNF has a satisfying truth assignment.

	Given any formula $\psi$ for $\problem{SAT}$, we construct an input instance $\db$ for $\cqa{q}$ as follows:
	\begin{itemize}
	\item for each variable $z$, we add $\phi_{z}^\bot[Rw]$ and $\phi_z^\bot[RvRw]$; 
	\item for each clause $C$ and positive literal $z$ of $C$, we add $\phi_{C}^z[u]$;
	\item for each clause $C$ and variable $z$ that occurs in a negative literal of $C$, we add $\phi_{C}^z[uRv]$.
	\end{itemize}
This construction can be  executed in \FO.
Figure~\ref{fig:coNP} depicts an example of the above construction. Intuitively, $\phi_{z}^\bot[Rw]$ corresponds to setting the variable $z$ to true, and $\phi_z^\bot[RvRw]$ to false. 
There are two types of conflicts that occur in~$\db$. First, we have conflicting facts of the form $S(\underline{C},*)$; resolving this conflict corresponds to the clause $C$ choosing one of its literals. Moreover, for each variable~$z$, we have conflicting facts of the form $R(\underline{z},*)$; resolving this conflict corresponds to the variable $z$ choosing a truth assignment.

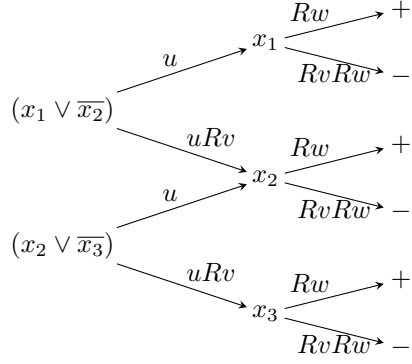
\begin{figure}[h]\centering
    \begin{tikzpicture}[->,>=stealth,auto=left, scale=1.8,vnode/.style={circle,black,inner sep=1pt,scale=1},el/.style = {inner sep=3/2pt}]
      \node[vnode] (c2) at (-1/2,0) {$(x_2 \vee \overline{x_3})$};
       \node[vnode] (c1) at (-1/2,1) {$(x_1 \vee \overline{x_2})$};
      
      \node[vnode] (x3) at (1,-1/2) {$x_3$};
      \node[vnode] (x2) at (1,1/2) {$x_2$};
      \node[vnode] (x1) at (1,3/2) {$x_1$};
      
      \node[vnode] (x3f) at (2,-1+1/4) {$-$};
      \node[vnode] (x3t) at (2,-1/2+1/4) {$+$};
      \node[vnode] (x2f) at (2,1/4) {$-$};
      \node[vnode] (x2t) at (2,1/2+1/4) {$+$};
      \node[vnode] (x1f) at (2,1+1/4) {$-$};
      \node[vnode] (x1t) at (2,3/2+1/4) {$+$};

      \path[->] (c1) edge node[el] {$u$} (x1); 
      \path[->] (c1) edge node[el] {$uRv$} (x2); 
      \path[->] (c2) edge node[el] {$u$} (x2); 
      \path[->] (c2) edge node[el] {$uRv$} (x3); 
      
      \path[->] (x1) edge node[el] {$Rw$} (x1t); 
       \path[->] (x1) edge node[el,below] {$RvRw$} (x1f); 
       \path[->] (x2) edge node[el] {$Rw$} (x2t); 
       \path[->] (x2) edge node[el,below] {$RvRw$} (x2f); 
       \path[->] (x3) edge node[el] {$Rw$} (x3t); 
       \path[->] (x3) edge node[el,below] {$RvRw$} (x3f); 
    \end{tikzpicture}
    \caption{Database instance for the \coNP-hardness reduction from the formula $\psi = (x_1 \vee \overline{x_2}) \wedge (x_2 \vee \overline{x_3})$.}
    \label{fig:coNP}
\end{figure}

	We show now that $\psi$ has a satisfying truth assignment if and only if there exists a repair of $\db$ that does not satisfy $q$.

\framebox{$\implies$}
	Assume that there exists a satisfying truth assignment $\sigma$ for~$\psi$. Then for any clause $C$, there exists a variable $z_C \in C$ whose corresponding literal is true in $C$ under $\sigma$. 
	Consider the repair $\rep$~that:
	\begin{itemize}
	\item for each variable $z$, it chooses the first $R$-fact of $\phi_{z}^\bot[Rw]$ if $\sigma(z)$ is true, otherwise the first $R$-fact of $\phi_z^\bot[RvRw]$;
	\item for each clause $C$, it chooses the first $S$-fact of  $\phi_{C}^z[u]$ if $z_C$ is positive in $C$, or  the first $S$-fact of $\phi_{C}^z[uRv]$ if $z_C$ is negative in $C$.
	\end{itemize} 
Assume for the sake of contradiction that $\rep$ satisfies~$q$. 
Then we must have a homomorphism from $q$ to either $uRw$ or $uRvRvRw$. But the former is not possible, while the latter contradicts $\cthree$.
We conclude by contradiction that $\rep$ falsifies~$q$.
	
\framebox{$\impliedby$}
Suppose that there exists a repair $\rep$ of $\db$ that falsifies~$q$. Consider the assignment $\sigma$:
$$
\sigma(z) = 
\begin{cases} 
\mbox{true} & \mbox{if $\phi_{z}^\bot[Rw] \subseteq \rep$}\\ 
\mbox{false} & \mbox{if $\phi_{z}^\bot[RvRw] \subseteq \rep$}
\end{cases}
$$ 
We claim that $\sigma$ is a satisfying truth assignment for $\psi$. 
Indeed, for each clause $C$, the repair must have chosen a variable $z$ in $C$. If $z$ appears as a positive literal in $C$, then 
	$\phi_C^z[u] \subseteq\rep$.
	% $\psi_{z,C}^+(p^+_{cl}) \subseteq r$. 
	Since $\rep$ falsifies $q$, we must have $\phi_z^\bot[Rw] \subseteq \rep$. Thus, $\sigma(z)$ is true and $C$ is satisfied.
	 If~$z$ appears in a negative literal, then $\phi_C^z[uRv] \subseteq \rep$. 
	 Since $\rep$ falsifies~$q$, we must have $\phi_z^\bot[RvRw] \subseteq \rep$. Thus, $\sigma(z)$ is false and $C$ is again~satisfied.
\end{proof}

\subsection{\PTIME-Hardness}\label{sec:phard}

Finally, we show the \PTIME-hard lower bound.

\begin{lemma}
\label{lemma:p-hard}
If a path query $q$ violates $\ctwo$, then $\cqa{p}$ is \PTIME-hard.
\end{lemma}

\begin{proof}
	Suppose $q$ violates $\ctwo$. If $q$ also violates  $\cthree$ , then the problem $\cqa{q}$ is \PTIME-hard since it is \coNP-hard by Lemma~\ref{lemma:conp-hard}. 
	Otherwise, it is possible to write $q = u  Rv_1  Rv_2 Rw$, with three consecutive occurrences of $R$ such that $v_1 \neq v_2$ and $Rw$ is not a prefix of $Rv_1$. 
	Let $v$ be the maximal path query such that $v_1 = v v_1^+$ and $v_2 = v v_2^+$. Thus $v_1^+ \neq v_2^+$ and the first relation names of $v_1^+$ and $v_2^+$ are different.

Our proof is a reduction from the \textsf{Monotone Circuit Value Problem (MCVP)} known to be \PTIME-complete \cite{10.1145/1008354.1008356}:

\begin{description}
\item[Problem:] $\problem{MCVP}$
	\item[Input:] A monotone Boolean circuit $C$ on inputs $x_1$, $x_2$, $\dots$, $x_n$ and output gate $o$; an assignment $\sigma: \{x_i \mid 1 \leq i \leq n\} \rightarrow \{0, 1\}$.
	\item[Question:] What is the value of the output $o$ under $\sigma$?
\end{description}
We construct an instance $\db$ for $\cqa{q}$ as follows:
\begin{itemize}
\item for the output gate $o$, we add $\phi_\bot^{o}[uRv_1]$;
% \item for each input variable $x$ with $\sigma(x)=0$, we add $\phi_\bot^{x}[u]$;
\item for each input variable $x$ with $\sigma(x)=1$, we add % $\phi_\bot^{x}[u]$ and 
$\phi_x^\bot[Rv_2Rw]$;
\item for each gate $g$, we add $\phi_\bot^g[u] $ and $\phi_g^\bot[Rv_2Rw]$;
\item for each \textsf{AND} gate $g = g_1 \land g_2$, we add 
 	$$\phi_{g}^{g_1}[Rv_1] \cup \phi_{g}^{g_2}[Rv_1].$$
Here, $g_{1}$ and $g_{2}$ can be gates or input variables; and	
\item for each \textsf{OR} gate $g = g_1 \lor g_2$, we add 
$$\setlength{\arraycolsep}{2pt}
 	\begin{array}{*{6}{l}}
 	  & \phi_g^{c_1}[Rv] & \cup & \phi_{c_1}^{g_1}[v_1^+] & \cup & \phi_{c_1}^{c_2}[v_2^+] \\[0.5ex]
	  \cup & \phi_\bot^{c_2}[u]& \cup & \phi_{c_2}^{g_2}[Rv_1] &\cup & \phi_{c_2}^\bot[Rw]
 	\end{array} 
$$ 	
	where $c_1,c_2$ are fresh constants.
\end{itemize}
This construction can be  executed in \FO.
An example of the gadget constructions is shown in Figure~\ref{fig:gadgets-reduction-generic}.
We next show that the output gate $o$ is evaluated to~$1$ under $\sigma$ if and only if each repair of $\db$ satisfies $q$. 

	\begin{figure}[!ht]
	    \centering
	    \subfloat[\textsf{AND} gate]{ 
    \begin{tikzpicture}[->,>=stealth,auto=left, scale=1.8,vnode/.style={black,inner sep=2pt,scale=1},el/.style = {inner sep=3/2pt}]
      \node[vnode] (v1) at (1/2,0) {};
      \node[vnode,draw] (g) at (1,0) {$g$};
      \node[vnode,draw] (g1) at (1,-2/3) {$g_2$};
      \node[vnode,draw] (g2) at (1,2/3) {$g_1$};
      \node[vnode] (e) at (7/4,0) {};

       \path[->] (v1) edge node[el] {$u$} (g); 
       \path[->] (g) edge node[el] {$Rv_1$} (g1); 
       \path[->] (g) edge node[el] {$Rv_1$} (g2);
       \path[->] (g) edge node[el] {$Rv_2Rw$} (e); 
 
    \end{tikzpicture}
	    }%
	    \qquad
	    \subfloat[\textsf{OR} gate]{
    \begin{tikzpicture}[->,>=stealth,auto=left, scale=1.8,vnode/.style={black,inner sep=2pt,scale=1},el/.style = {inner sep=3/2pt}]
      \node[vnode] (v1) at (1/2,0) {};
      \node[vnode] (v2) at (2-1/2,-2/3) {};
      \node[vnode] (e2) at (2+1/2,-2/3) {};
      \node[vnode,draw] (g) at (1,0) {$g$};
      \node[vnode] (e) at (1,-2/3) {};
      \node[vnode,draw] (g2) at (2,1/2) {$g_1$};
      \node[vnode,draw] (g1) at (2,-1-1/4) {$g_2$};
      \node[vnode] (c1) at (2,0) {$c_1$};
       \node[vnode] (c2) at (2,-2/3) {$c_2$};

       \path[->] (v1) edge node[el] {$u$} (g); 
       \path[->] (g) edge node[el,left] {$Rv_2Rw$} (e); 
       \path[->] (c1) edge node[el] {$v_1^+$} (g2);
       \path[->] (g) edge node[el] {$Rv$} (c1); 
       \path[->] (c1) edge node[el] {$v_2^+$} (c2);  
       \path[->] (v2) edge node[el] {$u$} (c2); 
       \path[->] (c2) edge node[el] {$Rw$} (e2); 
       \path[->] (c2) edge node[el] {$Rv_1$} (g1); 

    \end{tikzpicture}	    }%
	    \caption{Gadgets for the \PTIME-hardness reduction.}
	    \label{fig:gadgets-reduction-generic}
	\end{figure}
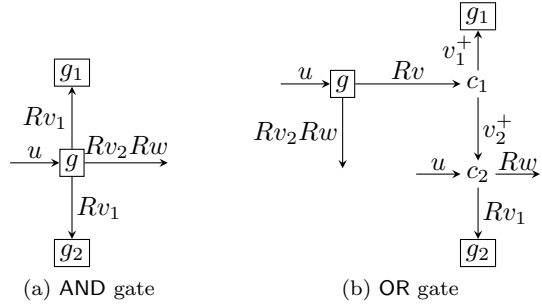

\framebox{$\implies$}
 Suppose the output gate $o$ is evaluated to $1$ under $\sigma$. Consider any repair $\rep$. We construct a sequence of gates starting from $o$, with the invariant that every gate $g$  evaluates to~$1$, and there is a path of the form $uRv_1$ in~$\rep$ that ends in $g$. The output gate $o$ evaluates to~$1$, and also we have that $\phi_\bot^o[uRv_1] \subseteq \rep$ by construction. Suppose that we are at gate $g$. If there is a $Rv_2Rw$ path in $\rep$ that starts in~$g$, the sequence ends and the query $q$ is satisfied. Otherwise, we distinguish two cases:
  	\begin{enumerate}
 		\item $g=g_1 \land g_2$. Then, we choose the gate with $\phi_{g}^{g_{i}}[Rv_1] \subseteq \rep$. Since both gates evaluate to 1 and $ \phi_\bot^{g}[u] \subseteq \rep$, the invariant holds for the chosen gate. 
		\item $g  = g_1 \lor g_2$.  If $g_1$ evaluates to~$1$, we choose $g_1$. Observe that  $ \phi_\bot^g[u]  \cup \phi_g^{c_1}[Rv] \cup \phi_{c_1}^{g_1}[v_1^+] $ creates the desired $uRv_1$ path.
Otherwise $g_2$ evaluates to~$1$. 
If $\phi_{c_2}^\bot[Rw]\subseteq\rep$, then there is a path with trace $uRv_{1}$ ending in~$g$, and a path with trace $Rv_{2}Rw$ starting in~$g$, and therefore $\rep$ satisfies~$q$.
If $\phi_{c_2}^\bot[Rw]\nsubseteq \rep$, we choose $g_2$ and the invariant holds.  
 	\end{enumerate}		
If the query is not satisfied at any point in the sequence, we will reach an input variable~$x$ evaluated at~$1$. But then there is an outgoing $Rv_2Rw$ path from $x$, which means that~$q$ must be satisfied.
	
\framebox{$\impliedby$}
Proof by contraposition.
Assume that $o$ is evaluated to $0$ under $\sigma$. We construct a repair $\rep$ as follows, for each gate $g$:
 \begin{itemize}
 \item if $g$ is evaluated to 1, we choose the first $R$-fact in $\phi_g^\bot[Rv_2Rw]$;
 \item if $g = g_1 \land g_2$ and $g$ is evaluated to~$0$, let $g_i$ be the gate or input variable evaluated to~$0$. We then choose $\phi_g^{g_i}[Rv_1]$;
  \item if $g = g_1 \lor g_2$ and $g$ is evaluated to~$0$, we choose $\phi_g^{c_1}[Rv]$; and
  \item if $g = g_1 \lor g_2$, we choose $\phi_{c_2}^{g_{2}}[Rv_{1}]$.
 \end{itemize}
For a path query $p$, we write $\head{p}$ for the variable at the key-position of the first atom, and $\rear{p}$ for the variable at the non-key position of the last atom.

Assume for the sake of contradiction that $\rep$ satisfies $q$. 
  Then, there exists some valuation $\theta$ such that $\theta(uRv_1Rv_2Rw) \subseteq \rep$. Then the gate $g^*\defeq\theta(\head{Rv_1})$ is evaluated to~$0$ by construction. 
  Let $g_1\defeq\theta(\rear{Rv_1})$. By construction, for $g^* = g_1 \land g_2$ or $g^* = g_1 \lor g_2$, we must have $\phi_{g}^{g_1}[Rv_1] \subseteq \rep$ and $g_1$ is a gate or an input variable also evaluated to~$0$. 
By our construction of $\rep$, there is no path with trace $Rv_2Rw$ outgoing from $g_1$.
However, $\theta(Rv_2Rw)\subseteq\rep$, this can only happen when $g_1$ is an \textsf{OR} gate, 
and one of the following occurs:
\begin{itemize}
\item 
Case that $\card{Rw}\leq\card{Rv_{1}}$, and the trace of $\theta(Rv_{2}Rw)$ is a prefix of $Rvv_{2}^{+}Rv_{1}$.
Then $Rw$ is a prefix of $Rv_1$, a contradiction.
\item 
Case that $\card{Rw}>\card{Rv_{1}}$, and $Rvv_{2}^{+}Rv_{1}$ is a prefix of the trace of $\theta(Rv_{2}Rw)$.
Consequently, $Rv_{1}$ is a prefix of $Rw$.
Then, for every $k\geq 1$, $\pumpclosure{q}$ contains $uRv_{1}\kleene{Rv_{2}}{k}Rw$.
It is now easily verified that for large enough values of~$k$, $uRv_{1}Rv_{2}w$ is not a factor of $uRv_{1}\kleene{Rv_{2}}{k}Rw$.
By Lemmas~\ref{lem:pumpclosure} and~\ref{lemma:conp-hard},
$\cqa{q}$ is \coNP-hard.
\qedhere
\end{itemize}
\end{proof}

	\section{Path Queries with Constants}
\label{sec:constant}

We now extend our complexity classification of $\cqa{q}$ to path queries in which constants can occur. 

\begin{definition}[Generalized path queries]\label{def:gpq}
 A {\em generalized path query} is a Boolean conjunctive query of the following form:
\begin{equation}\label{eq:gpq}
q = \{R_{1}(\underline{s_{1}},s_{2}), R_{2}(\underline{s_{2}},s_{3}), \dots, R_{k}(\underline{s_{k}},s_{k+1})\},
\end{equation}
where $s_{1}$, $s_{2}$,\dots, $s_{k+1}$ are constants or variables, all distinct, and $R_{1}$, $R_{2}$,\dots, $R_{k}$ are (not necessarily distinct) relation names.
Significantly, every constant can occur at most twice: at a non-primary-key position and the next primary-key-position.  

The \emph{characteristic prefix} of $q$, denoted by $\chr{q}$, is the longest prefix
\begin{equation*}
\{R_{1}(\underline{s_{1}},s_{2}), R_{2}(\underline{s_{2}},s_{3}), \dots, R_{\ell}(\underline{s_{\ell}},s_{\ell+1})\}, 0\leq\ell\leq k
\end{equation*}
such that no constant occurs among $s_{1}$, $s_{2}$, \dots, $s_{\ell}$ (but $s_{\ell+1}$ can be a constant).
Clearly, if $q$ is constant-free, then $\chr{q}=q$.
%unique maximal prefix query of $q$ that contains no constant or the only constant is at the non-key position of the last atom. If $q$ contains no constants, $\chr{q} = q$.
\qed
\end{definition}

\begin{example} \label{ex:char}
If $q =\{R(\underline{x},y)$, $S(\underline{y},0)$, $T(\underline{0},1)$, $R(\underline{1},w)\}$, where $0$ and~$1$ are constants, then $\chr{q}=\{R(\underline{x},y)$, $S(\underline{y},0)\}$. 
\qed
\end{example}

The following lemma implies that if a generalized path query~$q$ starts with a constant, then $\cqa{q}$ is in \FO. This explains why the complexity classification in the remainder of this section will only depend on $\chr{q}$.

\begin{lemma} \label{lemma:tail-part-fo}
For any generalized path query $q$, $\cqa{p}$ is in \FO, where $p\defeq q\setminus\chr{q}$.
\end{lemma}

We now introduce some definitions and notations used in our complexity classification.
The following definition introduces a convenient syntactic shorthand for characteristic prefixes previously defined in~Definition~\ref{def:gpq}.

\begin{definition}\label{def:enhanced}
Let $q = \{R_{1}(\underline{x_{1}},x_{2})$, $R_{2}(\underline{x_{2}},x_{3})$, \dots, $R_{k}(\underline{x_{k}},x_{k+1}) \}$ be a path query.
We write $\pathcons{q}{c}$ for the generalized path query obtained from $q$ by replacing $x_{k+1}$ with the constant~$c$.
The constant-free path query~$q$ will be denoted by $\pathcons{q}{\top}$, where $\top$ is a distinguished special symbol.
\qed
\end{definition}

\begin{definition}[Prefix homomorphism] \label{defn:prefix-homo}
Let
\begin{eqnarray*}
q & = & \{R_{1}(\underline{s_{1}},s_{2}), R_{2}(\underline{s_{2}},s_{3}), \dots, R_{k}(\underline{s_{k}},s_{k+1})\}\\
p & = & \{S_{1}(\underline{t_{1}},t_{2}), S_{2}(\underline{t_{2}},t_{3}), \dots, R_\ell(\underline{s_{\ell}},s_{\ell+1})\}
\end{eqnarray*}
be generalized path queries.
A \emph{homomorphism from $q$ to $p$} is a substitution~$\theta$ for the variables in~$q$, extended to be the identity on constants, such that for every $i\in\{1,\dots,k\}$, $R_{i}(\underline{\theta(s_{i})},\theta(s_{i+1}))\in p$.
Such a homomorphism is a \emph{prefix homomorphism} if $\theta(s_{1})=t_{1}$.
%A \emph{prefix homomorphism} from a generalized path query $q$ to $p$ is a homomorphism from $q$ to $p$ that maps the key position of the first relation in $q$ to that in $p$. Prefix homomorphisms generalize the notion of prefix in $\cone$ and $\ctwo$.
\qed
\end{definition}

\begin{example}
Let 
$q = \{R(\underline{x}, y)$, $R(\underline{y}, 1)$, $S(\underline{1}, z)\}$, and
$p = \{R(\underline{x}, y)$, $R(\underline{y}, z)$, $R(\underline{y}, 1)\}$. 
Then $\chr{q} = \{R(\underline{x}, y), R(\underline{y}, 1)\} = \pathcons{RR}{1}$ and $p =\pathcons{RRR}{1}$.
There is a homomorphism from $\chr{q}$ to $p$, but there is no prefix homomorphism from $\chr{q}$ to $p$.
\qed
\end{example}

The following conditions generalize $\cone$, $\ctwo$, and $\cthree$ from constant-free path queries to generalized path queries. Let $\gamma$ be either a constant or the distinguished symbol $\top$.
\begin{description}
\item[$\done$:] 
Whenever $\chr{q} = \pathcons{uRvRw}{\gamma}$, there is a prefix homomorphism from $\chr{q}$ to $\pathcons{uRvRvRw}{\gamma}$.
\item[$\dtwo$:]
Whenever $\chr{q} = \pathcons{uRvRw}{\gamma}$, there is a homomorphism from $\chr{q}$ to $\pathcons{uRvRvRw}{\gamma}$; and whenever $\chr{q}= \pathcons{u Rv_{1} Rv_{2} Rw}{\gamma}$ for consecutive occurrences of $R$,  $v_{1} = v_{2}$ or there is a prefix homomorphism from $\pathcons{Rw}{\gamma}$ to $\pathcons{Rv_{1}}{\gamma}$.
\item[$\dthree$:]
Whenever $\chr{q} = \pathcons{uRvRw}{\gamma}$, there is a homomorphism from $\chr{q}$ to $\pathcons{uRvRvRw}{\gamma}$.
\end{description} 
It is easily verified that if $\gamma=\top$, then $\done$, $\dtwo$, and $\dthree$ are equivalent to, respectively, $\cone$, $\ctwo$, and $\cthree$.
Likewise, the following theorem degenerates to Theorem~\ref{thm:main} for path queries without constants.

\begin{theorem}
  \label{thm:main-constant}
 For every generalized path query~$q$, the following complexity upper bounds obtain:
  \begin{itemize}
    \item if $q$ satisfies $\done$, then $\cqa{q}$ is in \FO;
    \item if $q$ satisfies $\dtwo$, then $\cqa{q}$ is in \NL; and
    \item if $q$ satisfies $\dthree$, then $\cqa{q}$ is in \PTIME.
  \end{itemize}
The following complexity lower bounds obtain:
  \begin{itemize}
    \item if $q$ violates $\done$, then $\cqa{q}$ is \NL-hard;
    \item if $q$ violates $\dtwo$, then $\cqa{q}$ is \PTIME-hard; and
    \item if $q$ violates $\dthree$, then $\cqa{q}$ is \coNP-complete.
  \end{itemize}
\end{theorem}

Finally, the proof of Theorem~\ref{thm:main-constant} reveals that for generalized path queries~$q$ containing at least one constant, the complexity of $\cqa{q}$ exhibits a trichotomy (instead of a tetrachotomy as in Theorem~\ref{thm:main-constant}).

\begin{theorem}\label{lem:trichotomy}
For any generalized path query $q$ containing at least one constant, the problem $\cqa{q}$ is either in \FO, \NL-complete, or \coNP-complete.  
\end{theorem}

% Corrollary~\ref{coro:nl-dichotomy} shows that the \PTIME-complete case in Theorem~\ref{thm:main} only occurs for path queries without constants. The proofs of Theorem~\ref{thm:main-constant} and Corollary~\ref{coro:nl-dichotomy} are in Appendix~\ref{appx:constant}.

	\section{Related Work}
\label{sec:related-work}

Inconsistencies in databases have been studied in different contexts~\cite{10.5555/1709465.1709573,kahale2020meta,katsis2010inconsistency}. Consistent query answering (CQA) was initiated by the seminal work
by Arenas, Bertossi, and Chomicki \cite{10.1145/303976.303983}. After twenty years, their contribution was acknowledged in a \emph{Gems of PODS session}~\cite{Bertossi19}. An overview of  complexity classification results in CQA appeared recently in the \emph{Database Principles} column of SIGMOD Record \cite{10.1145/3377391.3377393}.

The term $\cqa{q}$ was coined in~\cite{wijsen2010first} to refer to CQA for Boolean queries $q$ on databases that violate primary keys, one per relation, which are fixed by $q$'s schema. The complexity classification of $\cqa{q}$ for the class of self-join-free
Boolean conjunctive queries started with the work by Fuxman and Miller~\cite{FUXMAN2007610}, and was further pursued in~\cite{KOLAITIS201277,KoutrisS14,KoutrisW15,KoutrisW17,DBLP:conf/icdt/KoutrisW19,KoutrisWTOCS20},
which eventually revealed that the complexity of $\cqa{q}$ for self-join-free conjunctive queries displays a trichotomy between \FO,  \LSPACE-complete, and \coNP-complete. 
%Moreover, if $\cqa{q}$ is in \FO, a first-order rewriting of~$q$ (i.e., a first-order query that solves $\cqa{q}$) can be effectively constructed. 
A few extensions beyond this trichotomy result are known. 
It remains decidable whether or not $\cqa{q}$ is in \FO~for self-join-free Boolean conjunctive queries with negated atoms~\cite{KoutrisW18}, with respect to multiple keys~\cite{KoutrisW20}, and with unary foreign keys~\cite{DBLP:conf/pods/HannulaW22}, all assuming that $q$ is self-join-free.
%The complexity of $\cqa{q}$ for self-join-free Boolean conjunctive queries with negated atoms was studied in~\cite{KoutrisW18}. 
%For self-join-free Boolean conjunctive queries with respect to multiple keys, it remains decidable whether or not $\cqa{q}$ is in \FO~\cite{KoutrisW20}.

Little is known about $\cqa{q}$ beyond self-join-free conjunctive queries. Fontaine \cite{10.1145/2699912} showed that if we strengthen Conjecture~\ref{conj:dichotomy} from conjunctive queries to unions of conjunctive queries, then it implies Bulatov's dichotomy theorem
for conservative CSP~\cite{10.1145/1970398.1970400}.
This relationship between CQA and CSP was further explored in~\cite{lutz_et_al:LIPIcs:2015:4995}.
In~\cite{DBLP:journals/corr/AfratiKV15}, the authors show the \FO\ boundary for $\cqa{q}$ for constant-free Boolean conjunctive queries $q$ using a single binary relation name with a singleton primary key. 
Figueira et~al.~\cite{DBLP:conf/icdt/FigueiraPSS23} have recently discovered a simple fixpoint algorithm that solves $\cqa{q}$ when $q$ is a self-join free conjunctive query or a path query such that $\cqa{q}$  is in \PTIME.

The counting variant of the problem $\cqa{q}$, denoted
$\sharp\cqa{q}$, asks to count the number of repairs that satisfy some Boolean query $q$.
For self-join-free Boolean conjunctive queries, $\sharp\cqa{q}$ exhibits a dichotomy between  \textbf{FP} and \mbox{$\sharp$\PTIME}-complete~\cite{Maslowski2013ADI}. This dichotomy has been shown to extend to self-joins if primary keys are singletons~\cite{MaslowskiW14}, and to functional dependencies~\cite{DBLP:conf/pods/CalauttiLPS22a}.

In practice, systems supporting CQA have often used efficient solvers for Disjunctive Logic Programming, Answer Set Programming (ASP) or Binary Integer Programming (BIP), regardless of whether the CQA problem admits a first-order rewriting \cite{DBLP:conf/cikm/KhalfiouiJLSW20, chomicki2004hippo, DBLP:conf/sat/DixitK19, DBLP:conf/icde/DixitK22, DBLP:journals/pacmmod/FanKOW23, GrecoGZ03, 10.14778/2536336.2536341, manna_ricca_terracina_2015, DBLP:journals/dke/MarileoB10}.

\section{Conclusion}
\label{sec:conclusion}

We established a complexity classification in consistent query answering relative to primary keys, for path queries that can have self-joins: for every path query $q$, the problem $\cqa{q}$ is in \FO, \NL-complete, \PTIME-complete, or \coNP-complete, and it is decidable in polynomial time in the size of $q$ which of the four cases applies. If $\cqa{q}$ is in \FO\ or in \PTIME, rewritings of $q$ can be effectively constructed in, respectively, first-order logic and Least Fixpoint Logic .  %\xiating{Should we change from \NL\ to \PTIME?}
%This paper extends the existing classification of $\cqa{q}$ beyond self-join-free Boolean conjunctive queries.

%Some open problems are as follows. 
%An interesting extension to this paper is to provide a complexity classification of $\cqa{q}$ on path queries in which certain relations are known to be consistent. 
For binary relation names and singleton primary keys, an intriguing open problem is to generalize the form of the queries, from paths to directed rooted trees, DAGs, or general digraphs. 
The ultimate open problem is Conjecture~\ref{conj:dichotomy}, which conjectures that
for every Boolean conjunctive query $q$, $\cqa{q}$ is either in \PTIME\ or \coNP-complete. 
%Another open problem is to decide whether $\cqa{q}$ is in \FO, and if it is, construct a consistent first-order rewriting of $q$. Both problems have been resolved for self-join-free Boolean conjunctive queries on primary constraints and remain open for Boolean conjunctive queries possibly with self-joins \cite{10.1145/2745754.2745769,10.1145/3068334}.

\smallskip
\noindent \textbf{Acknowledgements.} This work is supported by the National Science Foundation under grant IIS-1910014.

	\bibliographystyle{abbrv}
	\bibliography{reference}

\begin{thebibliography}{10}

\bibitem{DBLP:journals/corr/AfratiKV15}
F.~N. Afrati, P.~G. Kolaitis, and A.~Vasilakopoulos.
\newblock Consistent answers of conjunctive queries on graphs.
\newblock {\em CoRR}, abs/1503.00650, 2015.

\bibitem{DBLP:conf/cikm/KhalfiouiJLSW20}
A.~{Amezian El Khalfioui}, J.~Joertz, D.~Labeeuw, G.~Staquet, and J.~Wijsen.
\newblock Optimization of answer set programs for consistent query answering by
  means of first-order rewriting.
\newblock In {\em {CIKM}}, pages 25--34. {ACM}, 2020.

\bibitem{10.1145/303976.303983}
M.~Arenas, L.~E. Bertossi, and J.~Chomicki.
\newblock Consistent query answers in inconsistent databases.
\newblock In {\em {PODS}}, pages 68--79. {ACM} Press, 1999.

\bibitem{berkholz2017answering}
C.~Berkholz, J.~Keppeler, and N.~Schweikardt.
\newblock Answering conjunctive queries under updates.
\newblock In {\em {PODS}}, pages 303--318. {ACM}, 2017.

\bibitem{Bertossi19}
L.~E. Bertossi.
\newblock Database repairs and consistent query answering: Origins and further
  developments.
\newblock In {\em {PODS}}, pages 48--58. {ACM}, 2019.

\bibitem{10.1145/1970398.1970400}
A.~A. Bulatov.
\newblock Complexity of conservative constraint satisfaction problems.
\newblock {\em {ACM} Trans. Comput. Log.}, 12(4):24:1--24:66, 2011.

\bibitem{DBLP:conf/pods/CalauttiLPS22a}
M.~Calautti, E.~Livshits, A.~Pieris, and M.~Schneider.
\newblock Counting database repairs entailing a query: The case of functional
  dependencies.
\newblock In {\em {PODS}}, pages 403--412. {ACM}, 2022.

\bibitem{10.5555/1709465.1709573}
J.~Chomicki and J.~Marcinkowski.
\newblock Minimal-change integrity maintenance using tuple deletions.
\newblock {\em Inf. Comput.}, 197(1-2):90--121, 2005.

\bibitem{chomicki2004hippo}
J.~Chomicki, J.~Marcinkowski, and S.~Staworko.
\newblock Hippo: {A} system for computing consistent answers to a class of
  {SQL} queries.
\newblock In {\em {EDBT}}, volume 2992 of {\em Lecture Notes in Computer
  Science}, pages 841--844. Springer, 2004.

\bibitem{DBLP:conf/sat/DixitK19}
A.~A. Dixit and P.~G. Kolaitis.
\newblock A sat-based system for consistent query answering.
\newblock In {\em {SAT}}, volume 11628 of {\em Lecture Notes in Computer
  Science}, pages 117--135. Springer, 2019.

\bibitem{DBLP:conf/icde/DixitK22}
A.~A. Dixit and P.~G. Kolaitis.
\newblock Consistent answers of aggregation queries via {SAT}.
\newblock In {\em {ICDE}}, pages 924--937. {IEEE}, 2022.

\bibitem{DBLP:journals/pacmmod/FanKOW23}
Z.~Fan, P.~Koutris, X.~Ouyang, and J.~Wijsen.
\newblock Lin{CQA}: Faster consistent query answering with linear time
  guarantees.
\newblock {\em Proc. {ACM} Manag. Data}, 1(1):38:1--38:25, 2023.

\bibitem{DBLP:conf/icdt/FigueiraPSS23}
D.~Figueira, A.~Padmanabha, L.~Segoufin, and C.~Sirangelo.
\newblock A simple algorithm for consistent query answering under primary keys.
\newblock In {\em {ICDT}}, volume 255 of {\em LIPIcs}, pages 24:1--24:18.
  Schloss Dagstuhl - Leibniz-Zentrum f{\"{u}}r Informatik, 2023.

\bibitem{10.1145/2699912}
G.~Fontaine.
\newblock Why is it hard to obtain a dichotomy for consistent query answering?
\newblock {\em {ACM} Trans. Comput. Log.}, 16(1):7:1--7:24, 2015.

\bibitem{FreireGIM15}
C.~Freire, W.~Gatterbauer, N.~Immerman, and A.~Meliou.
\newblock The complexity of resilience and responsibility for self-join-free
  conjunctive queries.
\newblock {\em Proc. {VLDB} Endow.}, 9(3):180--191, 2015.

\bibitem{FreireGIM20}
C.~Freire, W.~Gatterbauer, N.~Immerman, and A.~Meliou.
\newblock New results for the complexity of resilience for binary conjunctive
  queries with self-joins.
\newblock In {\em {PODS}}, pages 271--284. {ACM}, 2020.

\bibitem{FUXMAN2007610}
A.~Fuxman and R.~J. Miller.
\newblock First-order query rewriting for inconsistent databases.
\newblock {\em J. Comput. Syst. Sci.}, 73(4):610--635, 2007.

\bibitem{10.1145/1008354.1008356}
L.~M. Goldschlager.
\newblock The monotone and planar circuit value problems are log space complete
  for {P}.
\newblock {\em SIGACT News}, 9(2):25–29, July 1977.

\bibitem{GrecoGZ03}
G.~Greco, S.~Greco, and E.~Zumpano.
\newblock A logical framework for querying and repairing inconsistent
  databases.
\newblock {\em {IEEE} Trans. Knowl. Data Eng.}, 15(6):1389--1408, 2003.

\bibitem{DBLP:conf/pods/HannulaW22}
M.~Hannula and J.~Wijsen.
\newblock A dichotomy in consistent query answering for primary keys and unary
  foreign keys.
\newblock In {\em {PODS}}, pages 437--449. {ACM}, 2022.

\bibitem{kahale2020meta}
L.~A. Kahale, A.~M. Khamis, B.~Diab, Y.~Chang, L.~C. Lopes, A.~Agarwal, L.~Li,
  R.~A. Mustafa, S.~Koujanian, R.~Waziry, et~al.
\newblock Meta-analyses proved inconsistent in how missing data were handled
  across their included primary trials: A methodological survey.
\newblock {\em Clinical Epidemiology}, 12:527--535, 2020.

\bibitem{katsis2010inconsistency}
Y.~Katsis, A.~Deutsch, Y.~Papakonstantinou, and V.~Vassalos.
\newblock Inconsistency resolution in online databases.
\newblock In {\em {ICDE}}, pages 1205--1208. {IEEE} Computer Society, 2010.

\bibitem{KOLAITIS201277}
P.~G. Kolaitis and E.~Pema.
\newblock A dichotomy in the complexity of consistent query answering for
  queries with two atoms.
\newblock {\em Inf. Process. Lett.}, 112(3):77--85, 2012.

\bibitem{10.14778/2536336.2536341}
P.~G. Kolaitis, E.~Pema, and W.~Tan.
\newblock Efficient querying of inconsistent databases with binary integer
  programming.
\newblock {\em Proc. {VLDB} Endow.}, 6(6):397--408, 2013.

\bibitem{DBLP:conf/pods/KoutrisOW21}
P.~Koutris, X.~Ouyang, and J.~Wijsen.
\newblock Consistent query answering for primary keys on path queries.
\newblock In {\em {PODS}}, pages 215--232. {ACM}, 2021.

\bibitem{KoutrisS14}
P.~Koutris and D.~Suciu.
\newblock A dichotomy on the complexity of consistent query answering for atoms
  with simple keys.
\newblock In {\em {ICDT}}, pages 165--176. OpenProceedings.org, 2014.

\bibitem{KoutrisW15}
P.~Koutris and J.~Wijsen.
\newblock The data complexity of consistent query answering for self-join-free
  conjunctive queries under primary key constraints.
\newblock In {\em {PODS}}, pages 17--29. {ACM}, 2015.

\bibitem{KoutrisW17}
P.~Koutris and J.~Wijsen.
\newblock Consistent query answering for self-join-free conjunctive queries
  under primary key constraints.
\newblock {\em {ACM} Trans. Database Syst.}, 42(2):9:1--9:45, 2017.

\bibitem{KoutrisW18}
P.~Koutris and J.~Wijsen.
\newblock Consistent query answering for primary keys and conjunctive queries
  with negated atoms.
\newblock In {\em {PODS}}, pages 209--224. {ACM}, 2018.

\bibitem{DBLP:conf/icdt/KoutrisW19}
P.~Koutris and J.~Wijsen.
\newblock Consistent query answering for primary keys in logspace.
\newblock In {\em {ICDT}}, volume 127 of {\em LIPIcs}, pages 23:1--23:19.
  Schloss Dagstuhl - Leibniz-Zentrum f{\"{u}}r Informatik, 2019.

\bibitem{KoutrisW20}
P.~Koutris and J.~Wijsen.
\newblock First-order rewritability in consistent query answering with respect
  to multiple keys.
\newblock In {\em {PODS}}, pages 113--129. {ACM}, 2020.

\bibitem{KoutrisWTOCS20}
P.~Koutris and J.~Wijsen.
\newblock Consistent query answering for primary keys in datalog.
\newblock {\em Theory Comput. Syst.}, 65(1):122--178, 2021.

\bibitem{DBLP:books/sp/Libkin04}
L.~Libkin.
\newblock {\em Elements of Finite Model Theory}.
\newblock Texts in Theoretical Computer Science. An {EATCS} Series. Springer,
  2004.

\bibitem{lutz_et_al:LIPIcs:2015:4995}
C.~Lutz and F.~Wolter.
\newblock On the relationship between consistent query answering and constraint
  satisfaction problems.
\newblock In {\em {ICDT}}, volume~31 of {\em LIPIcs}, pages 363--379. Schloss
  Dagstuhl - Leibniz-Zentrum f{\"{u}}r Informatik, 2015.

\bibitem{manna_ricca_terracina_2015}
M.~Manna, F.~Ricca, and G.~Terracina.
\newblock Taming primary key violations to query large inconsistent data via
  {ASP}.
\newblock {\em Theory Pract. Log. Program.}, 15(4-5):696--710, 2015.

\bibitem{DBLP:journals/dke/MarileoB10}
M.~C. Marileo and L.~E. Bertossi.
\newblock The consistency extractor system: Answer set programs for consistent
  query answering in databases.
\newblock {\em Data Knowl. Eng.}, 69(6):545--572, 2010.

\bibitem{Maslowski2013ADI}
D.~Maslowski and J.~Wijsen.
\newblock A dichotomy in the complexity of counting database repairs.
\newblock {\em J. Comput. Syst. Sci.}, 79(6):958--983, 2013.

\bibitem{MaslowskiW14}
D.~Maslowski and J.~Wijsen.
\newblock Counting database repairs that satisfy conjunctive queries with
  self-joins.
\newblock In {\em {ICDT}}, pages 155--164. OpenProceedings.org, 2014.

\bibitem{wijsen2010first}
J.~Wijsen.
\newblock On the first-order expressibility of computing certain answers to
  conjunctive queries over uncertain databases.
\newblock In {\em {PODS}}, pages 179--190. {ACM}, 2010.

\bibitem{10.1145/2188349.2188351}
J.~Wijsen.
\newblock Certain conjunctive query answering in first-order logic.
\newblock {\em {ACM} Trans. Database Syst.}, 37(2):9:1--9:35, 2012.

\bibitem{10.1145/3377391.3377393}
J.~Wijsen.
\newblock Foundations of query answering on inconsistent databases.
\newblock {\em {SIGMOD} Rec.}, 48(3):6--16, 2019.

\end{thebibliography}
	\appendix

\section{Proofs for Section~\ref{sec:syntax}}

\subsection{Preliminary Results}

We define $\kleene{q}{k}=\emptyword$ if $k=0$.
The following lemma concerns words having a proper suffix that is also a prefix.

\begin{lemma}\label{lem:repeat}
If $w$ is a prefix of the word $uw$ with $u\neq\emptyword$,
then $w$ is a prefix of $\kleene{u}{\card{w}}$. 
%%%
Symmetrically, if $u$ is a suffix $uw$ with $w\neq\emptyword$, then $u$ is a suffix of $\kleene{w}{\card{u}}$. 
\end{lemma}
\begin{proof}%[Proof of Lemma~\ref{lem:repeat}]
Assume $w$ is a prefix of $uw$ with $u\neq\emptyword$.
The desired result is obvious if $\card{w}\leq\card{u}$, in which case $w$ is a prefix of $u$.
In the remainder of the proof, assume $\card{w}>\card{u}$.
The desired result becomes clear from the following construction:
\begin{center}
\includegraphics[scale=0.5]{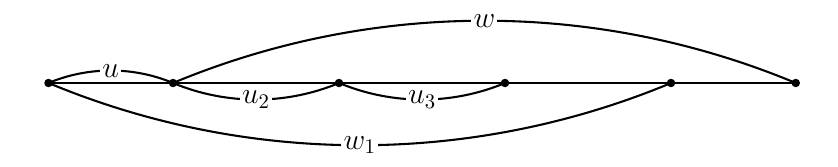}
\end{center}
The word $w_{1}$ is the occurrence of $w$ that is a prefix of~$uw$.
The word $u_{2}$ is the length-$\card{u}$ prefix of~$w$.
Obviously, $u_{2}=u$.
The word $u_{2}u_{3}$ is the length-$2\card{u}$ prefix of~$w$.
Obviously, $u_{3}=u_{2}$.
And so on.
It is now clear that $w$ is a prefix of $\kleene{u}{\card{w}}$.
Note that this construction requires $u\neq\emptyword$.
This concludes the proof.
\end{proof}

\begin{definition}[Episode]
An \emph{episode} of $q$ is a factor of $q$ of the form $RuR$ such that $R$ does not occur in~$u$.
Let $q=\ell RuRr$ where $RuR$ is an episode.
We say that this episode is \emph{right-repeating (within $q$)} if $r$ is a prefix of $\kleene{uR}{\card{r}}$.
Symmetrically, we say that this episode is \emph{left-repeating} if~$\ell$ is a suffix of $\kleene{Ru}{\card{\ell}}$.
\qed
\end{definition}

For example, let
$q= AMAA\overbrace{MAAM}^{e_{1}}A\overbrace{MAAM}^{e_{2}}AAMAB$.
Then the episode called $e_{1}$ is left-repeating, while the episode $e_{2}$ is neither left-repeating nor right-repeating.

\begin{definition}[Offset]
Let $u$ and $w$ be words.
We say that $u$ has \emph{offset} $n$ in $w$ if there exists words $p$, $s$ such that $\card{p}=n$ and $w=pus$. 
%We will use negative numbers for offset from the right, using $-0$, $-1$, $-2$, \dots
\qed
\end{definition}

\begin{lemma}[Repeating lemma]\label{lem:episode}
Let $q$ be a word that satisfies $\cthree$.
Then, every episode of $q$ is either left-repeating or right-repeating (or both).
\end{lemma}
\begin{proof}%[Proof of Lemma~\ref{lem:episode}]
Let $RuR$ be an episode in $q=\ell RuRr$.
By the hypothesis of the lemma, $q$ is a factor of $p\defeq\swipe{\ell}{Ru}{Rr}$.
Since $\card{q}-\card{p}=\card{u}+1$,
the offset of $q$ in $p$ is $\leq\card{u}+1$. 
Since $R$ does not occur in $u$, it must be that $q$ is either a prefix or a suffix of $p$.
We distinguish two cases:
\begin{description}
\item[Case that $q$ is a suffix of $p$.]
Then, it is easily verified that $\ell$ is a suffix of $\ell Ru$.
By Lemma~\ref{lem:repeat}, $\ell$ is a suffix of $\kleene{Ru}{\card{\ell}}$,
which means that $RuR$ is left-repeating within $q$.
\item[Case that $q$ is a prefix of $p$.]
We have that $r$ is a prefix of $uRr$.
By Lemma~\ref{lem:repeat}, $r$ is a prefix of $\kleene{uR}{\card{r}}$,
which means that $RuR$ is right-repeating.
\end{description}
This concludes the proof.
\end{proof}

\begin{definition}
If $q$ is a word over an alphabet $\Sigma$, then $\symbols{q}$ is the set that contains all (and only) the symbols that occur in~$q$.
\end{definition}

\begin{lemma}[Self-join-free episodes]\label{lem:sjfepisode}
Let $q$ be a word that satisfies $\cthree$.
Let $L\ell L$ be the right-most occurrence of an episode  that is left-repeating in $q$. 
Then, $L\ell$ is self-join-free.   
\end{lemma}
\begin{proof}%[Proof of Lemma~\ref{lem:sjfepisode}]
Consider for the sake of contradiction that $L\ell$ is not self-join-free.
Since $L\not\in\symbols{\ell}$, it must be that $\ell$ has a factor $MmM$ such that $Mm$ is self-join-free.
By Lemma~\ref{lem:episode}, $MmM$ must be left-repeating or right-repeating, which requires $L\in\symbols{Mm}$, a contradiction.
\end{proof}

%We give some intuition for the following lemma.
%Consider $ABBA$.
%The only self-join-free factors are $\emptyword$, $A$, $B$, $AB$, and $BA$.
%These are thus the only candidates for $u$, $v$, $w$.
%$ABBA$ is not of the forms~$\btwoa$ or~$\bthree$. 
%For instance, $ABBA=\kleene{AB}{1}\cdot\emptyword\cdot BA$ is not of the form~$\bthree$ because of the last occurrence of~$A$.
%Note that $ABBA$ is not a factor of its pumped version $ABBBA$.

\subsection{Proof of Lemma~\ref{lem:fofo}}

\begin{proof}[Proof of Lemma~\ref{lem:fofo}]
The implication \ref{it:fobone}$\implies$\ref{it:focone} is obvious.
To show \ref{it:focone}$\implies$\ref{it:fobone}, assume that $q$ satisfies~$\cone$.
The desired result is obvious if $q$ is self-join-free.
Assume from here on that $q$ is not self-join-free.
Then, we can write $q=\ell RmRr$, such that $\ell Rm$ is self-join-free.
That is, the second occurrence of $R$ is the left-most symbol that occurs a second time in~$q$.
By $\cone$, $q$ is a prefix of $\ell RmRmRr$.
It follows that $Rr$ is a prefix of $RmRr$.
By Lemma~\ref{lem:repeat}, $Rr$ is a prefix of $\kleene{Rm}{\card{r}+1}$.
It follows that there is a $k$ such that $q$ is a prefix of~$\ell\kleene{Rm}{k}$.
\end{proof}

\subsection{Proof of Lemma~\ref{lem:forms}}

\begin{proof}[Proof of Lemma~\ref{lem:forms}]
The proof of \ref{it:forms}$\implies$\ref{it:pumping} is straightforward.
We show next the direction \ref{it:pumping}$\implies$\ref{it:forms}.
To this end, assume that $q$ satisfies $\cthree$.
The desired result is obvious if $q$ is self-join-free (let $j=k=0$ in~$\btwoa$).
Assume that $q$ has a factor $L\ell L\cdot m\cdot Rr R$ where $L\ell L$ and $Rr R$ are episodes such that $\symbols{L\ell L}$, $\symbols{RrR}$, and $\symbols{m}$ are pairwise disjoint. 
Then, by Lemma~\ref{lem:episode}, $L\ell L$ must be left-repeating, and $RrR$ right-repeating.
By Lemma~\ref{lem:sjfepisode}, $L\ell$ and $rR$ are self-join-free.
Then $q$ is of the form~$\btwoa$. By letting $j=0$ or $k=0$, we obtain the situation where the number of episodes that are factors of $q$ is zero or one. 

The only difficult case is where two episodes overlap.
Assume that $q$ has an episode that is left-repeating (the case of a right-repeating episode is symmetrical).
Assume that this left-repeating episode is $e_{1}\defeq L\ell RoL$ in 
$q\defeq\dotsm\begin{array}{l}
\overbrace{L\ell Ro L}^{e_{1}}r R\\[-2.8ex]
\phantom{L\ell}\underbrace{\phantom{Ro Lr R}}_{e_{2}}
\end{array}\dotsm$,
where it can be assumed that $e_{1}$ is the right-most episode that is left-repeating.
Then, $\ell\neq\emptyword\neq r$ implies $\first{\ell}\neq\first{r}$ (or else $e_{1}$ would not be right-most, a contradiction).
By a similar reasoning, $\ell=\emptyword$ implies $r\neq\emptyword$.
Therefore, it is correct to conclude $\ell\neq r$.
It can also be assumed without loss of generality that~$r$ shares no symbols with~$e_{1}$, by choosing $R$ as the first symbol after~$e_{1}$ that also occurs in~$e_{1}$.
Now assume $e_{2}$ is right-repeating, over a length $>\card{oL}$.
Then $q$ contains a factor 
$\begin{array}{l}
\overbrace{L\ell Ro L}^{e_{1}}r R\\[-2.8ex]
\phantom{L\ell}\underbrace{\phantom{Ro Lr R}}_{e_{2}}
\end{array}\cdot oL$.
Then, $q$ rewinds to a word $p$ with factor:
$$\begin{array}{l}
\overbrace{L\ell Ro L}^{e_{1}}r R\\[-2.8ex]
\phantom{L\ell}\underbrace{\phantom{Ro Lr R}}_{e_{2}}
\end{array}\cdot o
\begin{array}{l}
\overbrace{L|\ell Ro L}^{e_{1}}r R\\[-2.8ex]
\phantom{L\ell}\underbrace{\phantom{Ro Lr R}}_{e_{2}}
\end{array}\cdot oL,$$
where the vertical bar $|$ is added to indicate a distinguished position.
It can now be verified that~$q$ is not a factor of~$p$, because of the alternation of $e_{1}$ and $e_{2}$ which does not occur in~$q$.
This contradicts the hypothesis of the lemma.
In particular, the words that start at position~$|$ are $r$ and $\ell$ in, respectively, $q$ and~$p$.
We conclude by contradiction that $e_{2}$ cannot be right-repeating over a length $>\card{oL}$.
Thus, following the right-most occurrence of~$e_{1}$, the word $q$~can contain fresh word~$r$, followed by~$RoL$, which is a suffix of~$e_{1}$.
This is exactly the form~$\btwob$. 

A remaining, and simpler,  case is where two episodes overlap by a single symbol $R=L$, giving 
$q\defeq\dotsm\begin{array}{l}
\overbrace{L\ell L}^{e_{1}}rL\\[-2.8ex]
\phantom{L\ell}\underbrace{\phantom{LrR}}_{e_{2}}
\end{array}\dotsm$,
where $e_{1}$ is the right-most episode that is left-repeating,
and $L$ is the first symbol after $e_{1}$ that also occurs in~$e_{1}$.
Therefore, $L$~does not occur in~$\ell\cdot r$, and $\ell\neq r$.
Indeed, 
if $\ell=\emptyword=r$, then $e_{1}$ is not right-most; and if $\ell\neq\emptyword\neq r$, then $\first{\ell}\neq\first{r}$, or else $e_{1}$ would not be right-most, a contradiction.
The word~$q$ rewinds to a word~$p$ with factor
$
\begin{array}{l}
\overbrace{L\ell L}^{e_{1}}r\overbrace{L\ell L}^{e_{1}}r L
,%%% punctuation symbol following the expression
\\[-2.8ex]
\phantom{L\ell}\underbrace{\phantom{LrL}}_{e_{2}}    
\phantom{\ell}\underbrace{\phantom{LrL}}_{e_{2}}   
\end{array}
$ 
and $\card{p}-\card{q}=\card{\ell}+\card{r}+2$. 
It is easily verified that $e_{2}$ cannot be right-repeating for~$>0$ symbols.
For instance, consider the case where $r\neq\emptyword$ and $e_{2}$ is right-repeating for~$1$ symbol, meaning that $q$ has suffix 
$\begin{array}{l}
\overbrace{L\ell L}^{e_{1}}rL\cdot\first{r}
,%%% punctuation symbol following the expression
\\[-2.8ex]
\phantom{L\ell}\underbrace{\phantom{LrR}}_{e_{2}}
\end{array}$
and $p$ has suffix
$
\begin{array}{l}
\overbrace{L\ell L}^{e_{1}}r\overbrace{L\ell L}^{e_{1}}r L\cdot\first{r}
.%%% punctuation symbol following the expression
\\[-2.8ex]
\phantom{L\ell}\underbrace{\phantom{LrL}}_{e_{2}}    
\phantom{\ell}\underbrace{\phantom{LrL}}_{e_{2}}   
\end{array}
$ 
If we left-align these suffixes, then there is a mismatch between $\first{r}$ and the leftmost symbol of $L\ell$.
The other possibility is to right-align these suffixes, but then $e_{1}$ cannot be genuinely left-repeating within~$q$.
\end{proof}

\subsection{Proof of Lemma~\ref{lem:notctwo}}

\begin{proof}[Proof of Lemma~\ref{lem:notctwo}]
Assume that $q$ satisfies $\cthree$.
By Lemma~\ref{lem:forms}, $q$ satisfies $\btwoa$, $\btwob$, or $\bthree$.

\framebox{\ref{it:notctwo}$\implies$\ref{it:falsifies}}
By contraposition.
Assume that \eqref{it:falsifies} does not hold.
Then, either $q$ satisfies~$\btwoa$ or $q$ satisfies~$\btwob$.
Assume that $q=aRb_{1}Rb_{2}Rc$ for three consecutive occurrences of~$R$ such that $b_{1}\neq b_{2}$.
It suffices to show that $Rc$ is a prefix of $Rb_{1}$.
It is easily verified that $b_{1}\neq b_{2}$ cannot happen if $q$ satisfies $\btwoa$.
Therefore, $q$ satisfies~$\btwob$.
The word in $\kleene{uv}{k}wv$ in $\btwob$ indeed allows for suffix $vu\cdot vw\cdot v$ where the first and second occurrence of $v$ are followed, respectively, by $u$ and~$w$.
Then, in~$q$, we have that $w$ is followed by a prefix of~$v$, and therefore $\ctwo$ is satisfied.

\framebox{\ref{it:falsifies}$\implies$\ref{it:factors}}
The hypothesis is that $q$ satisfies $\bthree$, but falsifies both $\btwoa$ and $\btwob$. 
We can assume $k\geq 0$ and self-join-free word $uvw$ such that $q$ is a factor of $uw\kleene{uv}{k}$, but $q$ falsifies~$\btwoa$ and~$\btwob$.
It must be that $u\neq\emptyword$ and the offset of $q$ in $uw\kleene{uv}{k}$ is $<\card{u}$, for otherwise $q$ is a factor of $w\kleene{uv}{k}$ and therefore satisfies~$\btwoa$, a contradiction. 
Also, one of $v$ or $w$ must not be the empty word, or else $q$ is a factor of $u\kleene{u}{k}$,   and therefore satisfies~$\btwoa$ (and also satisfies~$\btwob$).
%%%
We now consider the length of~$q$.
The word $uwuvu$ is a factor of $\kleene{wu}{2}vu$, and thus satisfies~$\btwob$.
If $v=\emptyset$, then the word $uwuu$ is a factor of $\kleene{wu}{2}u$, and thus satisfies~$\btwob$.
It is now correct to conclude that one of the following must occur:
\begin{itemize}
\item
$v\neq\emptyset$ and $\last{u}\cdot wuvu\cdot\first{v}$ is a factor of $q$; or
\item
$v=\emptyset$, $w\neq\emptyset$ and $\last{u}\cdot w\kleene{u}{2}\cdot\first{u}$ is a factor of $q$. 
\end{itemize}

\framebox{\ref{it:factors}$\implies$\ref{it:notctwo}}
Assume~\eqref{it:factors}.
Consider first the case $v\neq\emptyword$.
Let $u=\hat{u}R$ and $v=S\hat{v}$.
We have $R\neq S$, since $uv$ is self-join-free.
By item~\eqref{it:hercules}, $q$ has a factor $R\cdot w\hat{u}RS\hat{v}\hat{u}R\cdot S$, with three consecutive occurrences of~$R$.
It is easily verified that $w\hat{u}\neq S\hat{v}\hat{u}$, and that $RS$ is not a prefix of $Rw\hat{u}$.
Therefore $q$ falsifies $\ctwo$.

Consider next the case $v=\emptyword$ (whence $w\neq\emptyword$). 
Let $u=\hat{u}R$.
By item~\eqref{it:tritan}, $q$ has a factor $R\cdot w\hat{u}R\hat{u}R\cdot\first{u}$, with three consecutive occurrences of $R$.
Since $w\hat{u}\neq\hat{u}$ and $\first{u}\neq\first{w}$, it follows that $q$ falsifies $\ctwo$.
\end{proof} 

%\section{Proofs for Section~\ref{sec:key}}\label{app:key}

% The constant classification
\section{Proofs for Section~\ref{sec:constant}}
\label{appx:constant}

\subsection{Proof of Lemma~\ref{lemma:tail-part-fo}}

Lemma~\ref{lemma:tail-part-fo} is an immediate corollary of Lemma~\ref{lemma:fo-constant-head}, which states that whenever a generalized path query starts with a constant, then $\cqa{q}$ is in \FO. 
Its proof needs two helping lemmas.

\begin{lemma} \label{lemma:connected-components}
	Let $q=q_{1}\cup q_{2}\cup\dotsm\cup q_{k}$ be a Boolean conjunctive query such that for all $1 \leq i < j \leq k$, $\var{q_i}\cap \var{q_j}=\emptyset$. Then, the following are equivalent for every database instance $\db$:
	\begin{enumerate}
	\item\label{it:entire-graph}
	$\db$ is a ``yes''-instance for $\cqa{q}$; and
	\item\label{it:two-cc}
	for each $1\leq i\leq k$, $\db$ is a ``yes''-instance for $\cqa{q_i}$.
	\end{enumerate}
\end{lemma}
\begin{proof}
We give the proof for $k=2$.
The generalization to larger~$k$ is straightforward.

\framebox{\ref{it:entire-graph}$\implies$\ref{it:two-cc}}
Assume that~\eqref{it:cqaq} holds true.
Then each repair $\rep$ of $\db$ satisfies $q$, and therefore satisfies both $q_1$ and $q_2$. 
Therefore, $\db$ is a ``yes''-instance for both $\cqa{q_1}$ and $\cqa{q_2}$.

\framebox{\ref{it:two-cc}$\implies$\ref{it:entire-graph}}
Assume that~\eqref{it:two-cc} holds true.
Let $\rep$ be any repair of $\db$. 
Then there are valuations $\mu$ from $\var{q_1}$ to $\adom{\db}$, and $\theta$ from $\var{q_2}$ to $\adom{\db}$ such that $\mu(q_1) \subseteq\rep$ and $\theta(q_2) \subseteq\rep$. 
Since $\var{q_1} \cap \var{q_2} = \emptyset$ by construction, we can define a valuation~$\sigma$ as follows, for every variable $z\in\var{q_1}\cup\var{q_2}$:
	$$\sigma(z) = \begin{cases}
						\mu(z) & \mbox{if $z \in \var{q_1}$}\\
						\theta(z) & \mbox{if $z \in \var{q_2}$}
%						c & \text{ if } z = c \text{ is a constant,}
				\end{cases}
	$$
From $\sigma(q) = \sigma(q_1) \cup \sigma(q_2) = \mu(q_1) \cup \theta(q_2) \subseteq\rep$, it follows that $\rep$ satisfies $q$.
Therefore, $\db$ is a ``yes''-instance for $\cqa{q}$.
\end{proof}

\begin{lemma} \label{lemma:constant-end-reduction}
	Let $q$ be a generalized path query with
	$$q = \{R_{1}(\underline{s_{1}},s_{2}), R_{2}(\underline{s_{2}},s_{3}), \dots, R_{k}(\underline{s_{k}},c)\},$$ where $c$ is a constant, and each $s_i$ is either a constant or a variable for all $i \in \{1, \dots, k\}$. Let 
$$p = \{R_{1}(\underline{s_{1}},s_{2}), R_{2}(\underline{s_{2}},s_{3}), \dots, R_{k}(\underline{s_{k}},s_{k+1}), N(\underline{s_{k+1}}, s_{k+2})\},$$ where $s_{k+1}$, $s_{k+2}$ are fresh variables to~$q$ and $N$ is a fresh relation to~$q$. Then there exists a first-order reduction from $\cqa{q}$ to $\cqa{p}$.
\end{lemma}
\begin{proof}
Let $\db$ be an instance for $\cqa{q}$ and consider the instance $\db \cup \{N(\underline{c},d)\}$ for $\cqa{p}$ where $d$ is a fresh constant to $\adom{\db}$.

We show that $\db$ is a ``yes''-instance for $\cqa{q}$ if and only if $\db\cup\{N(\underline{c},d)\}$ is a ``yes''- instance for $\cqa{p}$.
 
\framebox{$\implies$}
Assume $\db$ is a ``yes''-instance for $\cqa{q}$. 
	Let $\rep$ be any repair of $\db \cup \{N(\underline{c},d)\}$, and thus $\rep\setminus \{N(\underline{c},d)\}$ is a repair for $\db$. Then there exists a valuation $\mu$ with $\mu(q) \subseteq\rep \setminus \{N(\underline{c},d)\}$. Consider the valuation $\mu^+$ from $\var{q} \cup \{s_{k+1}, s_{k+2}\}$ to $\adom{\db} \cup \{c, d\}$ that agrees with $\mu$ on $\var{q}$ and maps additionally $\mu^+(s_{k+1}) = c$ and $\mu^+(s_{k+2}) = d$. We thus have $\mu^+(p) \subseteq \rep$. 
It is correct to conclude that $\db \cup \{N(\underline{c},d)\}$ is a ``yes''-instance for $\cqa{p}$.	

\framebox{$\impliedby$}
Assume that $\db \cup \{N(\underline{c},d)\}$ is a ``yes''-instance for the problem $\cqa{p}$.
	Let $\rep$ be any repair of $\db$. 
Then $\rep \cup \{N(\underline{c},d)\}$ is a repair of $\db \cup \{N(\underline{c},d)\}$, and thus there exists some valuation $\theta$ with $\theta(p) \subseteq\rep \cup\{N(\underline{c},d)\}$. 
Since $\db$ contains only one $N$-fact, we have $\theta(s_{k+1})=c$.
It follows that $\theta(q)\subseteq\rep$, as desired.
\end{proof}

\begin{lemma} \label{lemma:fo-constant-head}
Let $q$ be a generalized path query with
	$$q = \{R_{1}(\underline{s_{1}},s_{2}), R_{2}(\underline{s_{2}},s_{3}), \dots, R_{k}(\underline{s_{k}},s_{k+1})\}$$ where $s_1$ is a constant, and each $s_i$ is either a constant or a variable for all $i \in \{2, \dots, k+1\}$. Then the problem $\cqa{q}$ is in \FO.
\end{lemma}
\begin{proof}
	Let the $1 = j_1 < j_2 < \dots < j_{\ell} \leq k + 1$ be all the indexes $j$ such that $s_{j}$ is a constant for some $\ell \geq 1$. Let $j_{\ell+1} = k+1$. Then for each $i \in \{1, 2, \dots, \ell\}$, the query
	$$q_i = \bigcup_{j_i \leq j < j_{i+1}} \{R_{j}(\underline{s_{j}}, s_{j+1})\}$$ is a generalized path query where each $s_{j_i}$ is a constant. 

	We claim that $\cqa{q_i}$ is in \FO\ for each $1 \leq i \leq \ell$. Indeed, if $s_{j_{i+1}}$ is a variable, then the claim follows by Lemma~\ref{lemma:constant-head-fo}; if $s_{j_{i+1}}$ is a constant, then the claim follows by Lemma~\ref{lemma:constant-end-reduction} and Lemma~\ref{lemma:constant-head-fo}.

	Since by construction, $q = q_1 \cup q_2 \cup \dots \cup q_{\ell}$, we conclude that $\cqa{q}$ is in \FO\ by Lemma~\ref{lemma:connected-components}.
\end{proof}

The proof of Lemma~\ref{lemma:tail-part-fo} is now simple.

\begin{proof}[Proof of Lemma~\ref{lemma:tail-part-fo}]
If $q$ contains no constants, the lemma holds trivially. Otherwise, $\cqa{p}$ is in \FO\ by Lemma~\ref{lemma:fo-constant-head}.
\end{proof}

\subsection{Elimination of Constants}

In this section, we show how constants can be eliminated from generalized path queries.
The \emph{extended query} of a generalized path query is defined next.

\begin{definition}[Extended query]
 Let $q$ be a generalized path query. The \emph{extended query} of $q$, denoted by $\extend{q}$, is defined as follows:  
\begin{itemize}
\item 
if $q$ does not contain any constant, then $\extend{q}\defeq q$; 
\item
otherwise, $\chr{q}=\{R_1(\underline{x_1}, x_2)$, $R_2(\underline{x_2}, x_3)$, \dots, $R_{\ell}(\underline{x_{\ell}}, c)\}$ for some constant~$c$.
In this case, we define
$$\extend{q}\defeq
\{R_1(\underline{x_1},x_2), \dots, 
  R_{\ell}(\underline{x_{\ell}}, x_{\ell+1}), 
  N(\underline{x_{\ell+1}},x_{\ell+2})\},$$
where $x_{\ell+1}$ and $x_{\ell+2}$ are fresh variables and $N$ is a fresh relation name not occurring in $q$.  
\qed
\end{itemize}
%\qed
\end{definition}

By definition, $\extend{q}$ does not contain any constant.

\begin{example}
Let $q = R(\underline{x}, y), S(\underline{y}, 0), T(\underline{0}, 1), R(\underline{1}, w)$ where $0$ and~$1$ are constants. We have $\extend{q} = R(\underline{x}, y), S(\underline{y}, z), N(\underline{z}, u)$. 
\qed
\end{example}

We show two lemmas which, taken together, show that the problem $\cqa{q}$ is first-order reducible to $\cqa{\extend{q}}$, for every generalized path query~$q$.

\begin{lemma} \label{lemma:prefix-suffices}
For every generalized path query $q$,
there is a first-order reduction from $\cqa{q}$ to $\cqa{\chr{q}}$.
\end{lemma}
\begin{proof}
Let $p \defeq q \setminus \chr{q}$.
Since $\var{\chr{q}} \cap \var{p} = \emptyset$, Lemmas~\ref{lemma:connected-components} and~\ref{lemma:fo-constant-head} imply that the following are equivalent for every database instance $\db$:
\begin{enumerate}
\item\label{it:cqaq}
$\db$ is a ``yes''-instance for $\cqa{q}$; and
\item\label{it:cqatwo}
$\db$ is a ``yes''-instance for $\cqa{\chr{q}}$ and a ``yes''-instance for $\cqa{p}$.
\end{enumerate}
To conclude the proof, it suffices to observe that $\cqa{p}$ is in \FO\ by Lemma~\ref{lemma:fo-constant-head}.
\end{proof}

\begin{lemma} \label{lem:last-step}
For every generalized path query $q$, there is a first-order reduction from $\cqa{\chr{q}}$ to $\cqa{\extend{q}}$.  
\end{lemma}

\begin{proof}
Let $q$ be a generalized path query.
If $q$ contains no constants, the lemma trivially obtains because $\chr{q} = \extend{q} = q$. If $q$ contains at least one constant, then there exists a first-order reduction from $\cqa{\chr{q}}$ to $\cqa{\extend{q}}$ by Lemma~\ref{lemma:constant-end-reduction}.
\end{proof}

\subsection{Complexity Upper Bounds in Theorem~\ref{thm:main-constant}}

\begin{lemma} \label{lem:ptime-degenerate}
Let $q$ be a generalized path query that contains at least one constant. 
If $q$ satisfies $\dthree$, then $q$ satisfies $\dtwo$ and $\extend{q}$ satisfies $\ctwo$.
\end{lemma}

\begin{proof}
Assume that $q$ satisfies $\dthree$. Let $\chr{q} = \pathcons{p}{c}$ for some constant $c$. We have $\extend{q} = p \cdot N$ where $N$ is a fresh relation name not occurring in~$p$.

We first argue that $\extend{q}$ is a factor of every word to which $\extend{q}$ rewinds.   
To this end, let $\extend{q}=uRvRwN$ where $p=uRvRw$. 
Since $q$ satisfies $\dthree$, there exists a homomorphism from $\chr{q}=\pathcons{uRvRw}{c}$ to $\pathcons{uRvRvRw}{c}$, implying that $uRvRw$ is a suffix of $uRvRvRw$.
It follows that $uRvRwN$ is a suffix of $uRvRvRwN$. 
Hence $\extend{q}$ satisfies $\cthree$.

The remaining test for $\ctwo$ is where
$\extend{q} = uRv_1Rv_2RwN$ for consecutive occurrences of~$R$.
We need to show that either $v_{1}=v_{2}$ or $RwN$ is a prefix of $Rv_{1}$ (or both).
%To this end, assume that $RwN$ is not a prefix of $Rv_1$. 
We have $p = uRv_1Rv_2Rw$. 
Since $q$ satisfies $\dthree$, there exists a homomorphism from $\chr{q} = \pathcons{uRv_1 Rv_2 Rw}{c}$ to $\pathcons{uRv_1Rv_2 Rv_2 Rw}{c}$. 
Since $c$ is a constant, the homomorphism must map $Rv_1$ to $Rv_2$, implying that~$v_1=v_2$.
It is correct to conclude that $q$ satisfies $\dtwo$ and $\extend{q}$ satisfies~$\ctwo$.
\end{proof}

\begin{lemma} \label{lem:never-higher}
For every generalized path query $q$, 
	\begin{itemize}
		\item if $q$ satisfies $\done$, then  $\extend{q}$ satisfies $\cone$;
		\item if $q$ satisfies $\dtwo$, then  $\extend{q}$ satisfies $\ctwo$; and
		\item if $q$ satisfies $\dthree$, then  $\extend{q}$ satisfies $\cthree$.
	\end{itemize}
\end{lemma}

\begin{proof}
The lemma holds trivially if $q$ contains no constant. 
Assume from here on that $q$ contains at least one constant.

Assume that $q$ satisfies $\done$. Then $\chr{q}$ must be self-join-free. In this case, $\extend{q}$ is self-join-free, and thus $\extend{q}$ satisfies $\cone$.

For the two remaining items, assume that $q$ satisfies $\dtwo$ or $\dthree$.
Since $\dtwo$ logically implies $\dthree$, $q$ satisfies $\dthree$.
By Lemma~\ref{lem:ptime-degenerate}, $\extend{q}$ satisfies $\ctwo$.
Since $\ctwo$ logically implies $\cthree$, $q$ satisfies $\cthree$.
\end{proof}

We can now prove the upper bounds in Theorem~\ref{thm:main-constant}.

\begin{proof}[Proof of upper bounds in Theorem~\ref{thm:main-constant}]
Since first-order reductions compose, by Lemmas~\ref{lemma:prefix-suffices} and~\ref{lem:last-step}, there is a first-order reduction from the problem $\cqa{q}$ to $\cqa{\extend{q}}$. The upper bound results then follow by Lemma~\ref{lem:never-higher}. 
\end{proof}

\subsection{Complexity Lower Bounds in Theorem~\ref{thm:main-constant}}

The complexity lower bounds in Theorem~\ref{thm:main-constant} can be proved by slight modifications of the proofs in Sections~\ref{sec:nlhard} and~\ref{sec:conphard}. We explain these modifications below for a generalized path query~$q$ containing at least one constant.
Note incidentally that the proof in Section~\ref{sec:phard} needs no revisiting, because, by Lemma~\ref{lem:ptime-degenerate}, a violation of~$\dtwo$ implies a violation of $\dthree$.
%Indeed, if $q$ violates $\dtwo$, then, by Lemma~\ref{lem:ptime-degenerate}, it also violates $\dthree$.

In the proof of Lemma~\ref{lemma:nl-hard}, let $\chr{q} = \pathcons{uRvRw}{c}$ where $c$ is a constant and there is no prefix homomorphism from $\chr{q}$ to $\pathcons{uRvRvRw}{c}$. Let $p = q \setminus \chr{q}$. Note that the path query $uRv$ does not contain any constant. We revise the reduction description in Lemma~\ref{lemma:nl-hard} to be  
\begin{itemize}
\item for each vertex $x \in V \cup \{s'\}$, we add $\phi_\bot^x[u]$; 
\item for each edge $(x, y) \in E \cup \{(s',s),(t,t')\}$, we add $\phi_x^y[Rv]$;
\item for each vertex $x \in V$, we add $\phi^c_x[Rw]$;  and
\item add a canonical copy of $p$ (which starts in the constant $c$).
\end{itemize}
An example is shown in Figure~\ref{fig:NL-constant}. 
Since the constant $c$ occurs at most twice in~$q$ by Definition~\ref{def:gpq}, the query $q$ can only be satisfied by a repair including each of $\phi_\bot^x[u]$, $\phi_x^y[Rv]$, $\phi^c_y[Rw]$, and the canonical copy of $p$. 
\NL-hardness can now be proved as in the proof of~Lemma~\ref{lemma:nl-hard}.

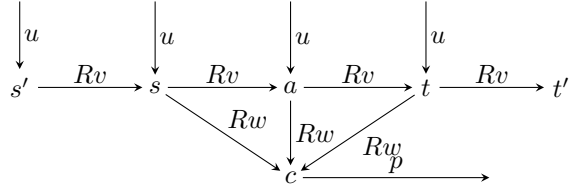
\begin{figure}[!h]\centering
    \begin{tikzpicture}[->,>=stealth,auto=left, scale=1.8,vnode/.style={circle,black,inner sep=1pt,scale=1},el/.style = {inner sep=3/2pt}]
      \node[vnode] (s') at (0,0) {$s'$};
      \node[vnode] (s) at (1,0) {$s$};
      \node[vnode] (a) at (2,0) {$a$};
      \node[vnode] (t) at (3,0) {$t$};
      \node[vnode] (t') at (4,0) {$t'$};
      
      \node[vnode] (ps') at (0,2/3) {};
      \node[vnode] (ps) at (1,2/3) {};
      \node[vnode] (pa) at (2,2/3) {};
      \node[vnode] (pt) at (3,2/3) {};
      
       \node[vnode] (ts) at (1,-2/3) {};
      \node[vnode] (ta) at (2,-2/3) {$c$};
      \node[vnode] (tt) at (3,-2/3) {};
      \node[vnode] (tp) at (3.5, -2/3) {};

       \path[->] (ps') edge node[el] {$u$} (s'); 
       \path[->] (ps) edge node[el] {$u$} (s); 
       \path[->] (pa) edge node[el] {$u$} (a); 
       \path[->] (pt) edge node[el] {$u$} (t); 
       
       \path[->] (s) edge node[el] {$Rw$} (ta); 
       \path[->] (a) edge node[el] {$Rw$} (ta); 
       \path[->] (t) edge node[el] {$Rw$} (ta); 
       \path[->] (ta) edge node[el] {$p$} (tp);

       \path[->] (s') edge node[el] {$Rv$} (s); 
       \path[->] (s) edge node[el] {$Rv$} (a); 
       \path[->] (a) edge node[el] {$Rv$} (t); 
       \path[->] (t) edge node[el] {$Rv$} (t'); 
    \end{tikzpicture}
    \caption{Database instance for the revised \NL-hardness reduction from the graph $G$ with  $V=\{s,a,t\}$ and $E=\{(s,a),(a,t)\}$.}
    \label{fig:NL-constant}
\end{figure}

\begin{figure}[!h]\centering
    \begin{tikzpicture}[->,>=stealth,auto=left, scale=1.8,vnode/.style={circle,black,inner sep=1pt,scale=1},el/.style = {inner sep=3/2pt}]
      \node[vnode] (c2) at (-1/2,0) {$(x_2 \vee \overline{x_3})$};
       \node[vnode] (c1) at (-1/2,1) {$(x_1 \vee \overline{x_2})$};
      
      \node[vnode] (x3) at (1,-1/2) {$x_3$};
      \node[vnode] (x2) at (1,1/2) {$x_2$};
      \node[vnode] (x1) at (1,3/2) {$x_1$};
      
      % \node[vnode] (x3f) at (2,-1+1/4) {$-$};
      % \node[vnode] (x3t) at (2,-1/2+1/4) {$+$};
      % \node[vnode] (x2f) at (2,1/4) {$-$};
      % \node[vnode] (x2t) at (2,1/2+1/4) {$+$};
      % \node[vnode] (x1f) at (2,1+1/4) {$-$};
      % \node[vnode] (x1t) at (2,3/2+1/4) {$+$};
      
      \node[vnode] (c) at (2.5, 1/2) {$c$};
      \node[vnode] (p) at (3, 1/2) {};
      
      \path[->] (c1) edge node[el] {$u$} (x1); 
      \path[->] (c1) edge node[el] {$uRv$} (x2); 
      \path[->] (c2) edge node[el] {$u$} (x2); 
      \path[->] (c2) edge node[el] {$uRv$} (x3); 
      
      \path[->] (c) edge node[el] {$p$} (p);

      \path[->] (x1) edge[bend left] node [pos=.023,above] {$+$} node[el] {$Rw$} (c); 
       \path[->] (x1) edge node [pos=.05,below] {$-$} node[el,left] {$RvRw$} (c); 
       \path[->] (x2) edge[bend left,below] node [pos=.05,above] {$+$} node[el] {$Rw$} (c); 
       \path[->] (x2) edge[bend right] node [pos=.05,below] {$-$} node[el] {$RvRw$} (c); 
       \path[->] (x3) edge node [pos=.05,above] {$+$} node[el] {$Rw$} (c); 
       \path[->] (x3) edge[bend right] node [pos=.025,below] {$-$} node[el,right] {$RvRw$} (c); 
    \end{tikzpicture}
    \caption{Database instance for the revised \coNP-hardness reduction from the formula $\psi = (x_1 \vee \overline{x_2}) \wedge (x_2 \vee \overline{x_3})$.}
    \label{fig:coNP-constant}
\end{figure}
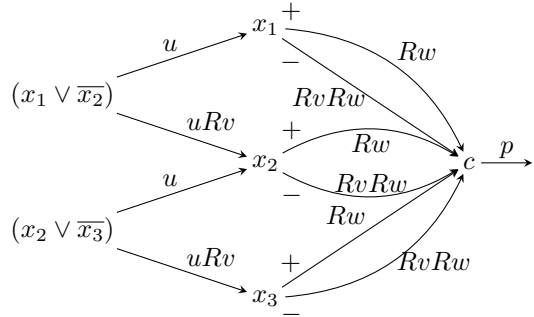

In the proof of Lemma~\ref{lemma:conp-hard}, let $\chr{q} = \pathcons{uRvRw}{c}$ where $c$ is a constant and there is no homomorphism from $\chr{q}$ to $\pathcons{uRvRvRw}{c}$. Let $p = q \setminus \chr{q}$. Note that both path queries $uRv$ and $u$ do not contain any constant. We revise the reduction description in Lemma~\ref{lemma:conp-hard} to be  
\begin{itemize}
\item for each variable $z$, we add $\phi_{z}^c[Rw]$ and $\phi_z^c[RvRw]$; 
\item for each clause $C$ and positive literal $z$ of $C$, we add $\phi_{C}^z[u]$;
\item for each clause $C$ and variable $z$ that occurs in a negative literal of $C$, we add $\phi_{C}^z[uRv]$; and
\item add a canonical copy of $p$ (which starts in the constant~$c$).
\end{itemize}
An example is shown in Figure~\ref{fig:coNP-constant}. Since the constant $c$ occurs at most twice in~$q$, the query $q$ can only be satisfied by a repair~$\rep$ such that either
\begin{itemize}
\item
$\rep$ contains $\phi_{C}^z[uRv]$, $\phi_{z}^c[Rw]$, and the canonical copy of $p$; or 
\item
$\rep$ contains $\phi_{C}^z[u]$, $\phi_z^c[RvRw]$, and the canonical copy of $p$.
\end{itemize}
\coNP-hardness can now be proved as in the proof of~Lemma~\ref{lemma:conp-hard}.

\subsection{Proof of Theorem~\ref{lem:trichotomy}}

\begin{proof}[Proof of Theorem~\ref{lem:trichotomy}]
Immediate consequence of Theorem~\ref{thm:main-constant} and Lemma~\ref{lem:ptime-degenerate}.
\end{proof}

%\onecolumn
%\section{Corrigendum}
%\input{erratum}
% \input{new-algo}
% \input{ufa}
% \input{circuit}

\end{document}